\documentclass{llncs}
\usepackage{hyperref}

\usepackage{amssymb} 
\usepackage{bbm} 
\usepackage{amsmath}
\usepackage{bussproofs} 
\usepackage{stmaryrd} 
\usepackage{latexsym} 
\usepackage[all]{xy} 
\usepackage[draft]{minted} 
\usepackage{thmtools, thm-restate} 
\usepackage{imakeidx} 

\usepackage{savesym}
\savesymbol{comment}
\usepackage[nompar]{commenting}
\declareauthor{lo}{Luke}{red}
\authorcommand{lo}{comment}
\declareauthor{sr}{Steven}{blue}
\authorcommand{sr}{comment}
\declareauthor{lp}{Long}{brown}
\authorcommand{lp}{comment}

\renewcommand{\OpenCommentBraket}{$\lceil$}
\renewcommand{\ClosedCommentBraket}{$\rfloor$}

\title{Defunctionalization of Higher-Order Constrained Horn Clauses}

\author{
Long Pham \inst{1} \and
Steven Ramsay \inst{2} \and
Luke Ong \inst{1}
}

\institute{University of Oxford, UK \\
\email{long.pham@keble.ox.ac.uk, Luke.Ong@cs.ox.ac.uk} \and
University of Bristol, UK \\
\email{steven.ramsay@bristol.ac.uk}}

\begin{document}

\maketitle


\begin{abstract}

Building on the successes of satisfiability modulo theories (SMT), Bj{\o}rner et al.~initiated a {research programme} advocating Horn constraints as a suitable basis for automatic program verification \cite{Bjorner2012}.
The notion of first-order constrained Horn clauses has recently been extended to higher-order logic by Cathcart Burn et al.~\cite{Ramsay2017}. 
To exploit the remarkable efficiency of SMT solving, a natural approach to solve systems of higher-order Horn constraints is to reduce them to systems of first-order Horn constraints.
This paper presents a defunctionalization algorithm to achieve the reduction.

Given a well-sorted higher-order constrained Horn clause (HoCHC) problem instance, the defunctionalization algorithm constructs a first-order well-sorted constrained Horn clause problem. 
In addition to well-sortedness of the algorithm's output, we prove that if an input HoCHC is solvable, then the result of its defunctionalization is solvable. 
The converse also holds, which we prove using a recent result on the continuous semantics of HoCHC.
To our knowledge, this defunctionalization algorithm is the first sound and complete reduction from systems of higher-order Horn constraints to systems of first-order Horn constraints.

We have constructed DefMono\footnotemark, a prototype implementation of the defunctionalization algorithm.
It first defunctionalizes an input HoCHC problem and then feeds the result into a backend SMT solver. 
We have evaluated the performance of DefMono empirically by comparison with two other higher-order verification tools.

\footnotetext{The web interface is available at \url{http://mjolnir.cs.ox.ac.uk/dfhochc/}. }
\end{abstract}


\section{Introduction}

\subsection{Background}

Notwithstanding the existence of undecidable problems, over the past decades, formal verification has proved to be useful and even essential to a number of computing applications. Hardware industries, particularly the semiconductor industry, have
long embraced the verification technology because the cost of manufacturing faulty hardware products is too costly. Hence, in such industries, formal verification has been used to detect bugs in early development stages. By contrast, software formal verification
had been less widely used than hardware verification because more advanced verification technology is required due to increased complexity in software. 
However, recent advances in the theory and practice of formal verification have led to wider use of formal methods in software as well. Recognising the value of formal verification, the 2007 Turing Awards were given to
Edmund Melson Clarke, E. Allen Emerson, and Joseph Sifakis for their contributions to model checking. 

Amongst the enabling technologies in the development of formal verification is satisfiability modulo theories (SMT) solvers \cite{Beckert2014}. 
Many approaches in formal verification reduce input programs to first-order constraints such as loop invariants and dependent types \cite{Bjorner2012}. These constraints are then fed into SMT solvers to check their satisfiability with respect to certain
background theories. The standardisation of input formats for SMT solvers is instrumental in accelerating the development of SMT solvers, allowing larger collections of benchmarks to be built. Also, with respect to formal verification, the standardisation
of SMT problem formats achieves separation of concerns by dividing the verification process into constraint generation and SMT solving.

Motivated by the standardisation of SMT problem formats, Bj{\o}rner et al.~propose standardization at a higher level: first-order verification problems \cite{Bjorner2012}. 
They suggest the use of constrained Horn clauses to express first-order verification problems, and their claim that Horn clauses serve as a suitable format of first-order verification problems is substantiated in \cite{Bjorner2015}. 
First-order constrained Horn clauses are subsequently extended to higher-order logic by Cathcart Burn et al.~\cite{Ramsay2017}. 

Whilst numerous verification techniques and tools have been created to verify first-order constrained Horn clauses, higher-order constrained Horn clause problems have not seen as much progress as first-order ones. 
We can exploit the advances in first-order Horn-clause solving by reducing higher-order constrained Horn clause problems to semantically equivalent first-order ones. 
This approach is pursued by Cathcart Burn et al.~\cite{Ramsay2017} using refinement types. 
In this refinement type-based approach, each free top-level relational variable is associated with a type. A valid type assignment can then be thought of as a model of an input HoCHC problem. 
One drawback of this method is incompleteness. 
Cathcart Burn et al.~\cite{Ramsay2017} report an instance of solvable HoCHC for which the refinement type-based approach produces an untypable logic program (that is, no model is found by this approach). 

In this work, I take a different approach and develop a defunctionalization algorithm to reduce higher-order constrained Horn clauses to first order ones. 
This is inspired by Reynolds's defunctionalization, a well-established method of reducing higher-order functional programs to first-order ones. 

\subsection{Related work}

\paragraph{First-order constrained Horn clause problems}

Using first-order Horn clauses to express first-order verification problems was originally proposed by Bj{\o}rner et al.~\cite{Bjorner2012}. 
They maintain that the Horn clause can serve as a suitable standard format of verification problems, enabling the development of a larger collection of benchmarks in the same format. 
In \cite{Bjorner2015}, they explain the relationship between Horn clauses and existential fixed-point logic (E+LFP), which is equivalent to Hoare logic. 
They also provide an overview of how to obtain first-order Horn clauses from first-order programs and how to solve first-order Horn clauses. 
The paper also gives a number of pointers to more detailed accounts of various Horn-clause verification methods.

\paragraph{Higher-order Horn clause problems and refinement types}

Cathcart Burn et al.~\cite{Ramsay2017} have extended the notion of constrained Horn clause problems to higher-order logic, and introduced the monotone semantics. 
Unlike the standard semantics, Horn clause problems have canonical models in monotone semantics, which is a very useful property in automated formal verification.
As an alternative representation of the higher-order constrained Horn clause problem, the monotone safety problem is introduced. Unlike the Horn clause problem,
the monotone safety problem does not contain logical implication, which is not monotone. Thus, the monotone safety problem is a more suitable representation in the monotone semantics, 
although the difference between the monotone safety problem and Horn clause problem is purely syntactic. 
The paper also explores the connection between the standard and monotone semantics, proving that any higher-order constrained Horn clause problem in the standard semantics can be converted into a semantically equivalent monotone safety problem. 

In the second half of the paper, a refinement type-based approach to verifying monotone safety problems is presented. 

\paragraph{Defunctionalization}

In the conclusion of \cite{Ramsay2017}, Cathcart Burn et al.~propose the use of Reynolds's defunctionalization to reduce higher-order Horn clause problems to first-order ones as done by the refinement type-based approach. 
This is what motivates the present work on defunctionalization of HoCHC. 
The idea of representing higher-order functions by closures to verify higher-order programs can also be found in \cite{Bjorner2013}, although this only gives a brief overview of the approach. 

Defunctionalization is explained in a detailed yet readable manner in its original paper by Reynolds \cite{Reynolds1972}. 
In this paper, typability of the apply function created as a result of defunctionalization is not considered. 
A problem arises when we deal with polymorphic languages. 
This issue is resolved using type specialization in \cite{Bell1997}. 
Another work on defunctionalization of polymorphic languages is \cite{Pottier2004}. 
Although the present work on defunctionalization of monotone safety problems does not involve polymorphic types, 
the idea of formulating a defunctionalization algorithm using inference rules comes from \cite{Pottier2004}.

\subsection{Contributions}

The chief contribution of this work is the development of a defunctionalization algorithm to reduce HoCHC to first-order constrained Horn clauses. 
With respect to the correctness of the algorithm, I prove type preservation, completeness, and soundness. 
The output of the defunctionalization algorithm is proved to be well-sorted, given that the input is well-sorted. 
Using the idea of valuation extraction, I also prove that if an input higher-order constrained Horn clause problem is solvable, then its defunctionalized problem is also solvable. 
The proof for the converse is achieved by using a recent result on the continuous semantics of HoCHC \cite{Jochems18}.
As far as I am aware, this is the first sound and complete reduction from HoCHC to first-order constrained Horn clauses. 

\subsection{Outline of this report}

This document is structured as follows.

Section~\ref{chapter on preliminaries} introduces higher-order logic, logic program safety problems, and monotone semantics. 

Section~\ref{chapter on defunctionalization of monotone problems with a concrete example} illustrates how defunctionalization works on a concrete example. 

In the first half of Section~\ref{chapter on the formal presentation of the algorithm}, the defunctionalization algorithm is formulated using inference rules. 
In the second half of this section, completeness and soundness of the algorithm are established.

Section~\ref{chapter on implementation and evaluation} presents a prototype tool based on the defunctionalization algorithm and compares its performance with other higher-order verification tools. 

Section~\ref{chapter on conclusion} summarises the work and proposes a few directions for future work.

Appendix~\ref{chapter on the algorithm in the appendix} presents details of the preprocessing in the defunctionalization algorithm. Also, the rationale for the algorithm's design is given. 

Appendix~\ref{chapter on monotonicity of alpha'} describes how to obtain monotone valuations for outputs of the defunctionalization algorithm. 

Appendix~\ref{chapter on the supplements for meaning preservation} provides detailed proofs for the lemmas and theorems presented in Section~\ref{chapter on the formal presentation of the algorithm}. 

Appendix~\ref{chapter on type preservation} gives a formal proof of type preservation. 


\section{Preliminaries} \label{chapter on preliminaries}

This section introduces the basics of higher-order logic, logic program safety problems, and monotone semantics. 
Higher-order constrained Horn clauses (HoCHC) are not formally introduced, since the defunctionalization algorithm works on logic program safety problems, which are alternative representations of HoCHC \cite{Ramsay2017}.
It is therefore sufficient to understand that HoCHC and logic program safety problems are equivalent. 

\subsection{Higher-order logic}

In this subsection I review the syntax and semantics of higher-order logic based on a simply typed lambda calculus. The presentation style of this subsection follows the one in \cite{Ramsay2017}. 

\subsubsection{Syntax}
In a simply typed lambda calculus\index{simply typed lambda calculus}, each value is associated with a sort\index{sort} that denotes the category of elements to which the value belongs. Let $(b \in) \mathbb{B}$ be
a fixed set of user-defined base sorts\index{base sort} including a sort $\iota$ of individuals\index{individuals} and a sort $o$ of propositions\index{propositions}. Using the base sorts, simple sorts\index{simple sort} are inductively defined as follows:
\begin{equation*}
\sigma ::= b \mid \sigma_1 \to \sigma_2,
\end{equation*}
where $b \in \mathbb{B}$. As standard, the sort constructor $\to$ associates to the right. The order\index{sort!order|see {order}}\index{order} of a sort is defined by
\begin{alignat*}{2}
{\tt ord}(b) & = 1 && \qquad \text{if } b \in \mathbb{B} \\
{\tt ord}(\sigma_1 \to \sigma_2) & = \max \{ {\tt ord}(\sigma_1) + 1, {\tt ord}(\sigma_2) \} && \qquad \text{otherwise.}
\end{alignat*}

Let $\Sigma = (\mathbb{B}, \mathbb{S})$ denote a first-order signature\index{signature}, where $\mathbb{B}$ is a set of base sorts that includes the propositional sort $o$ and at least one sort
of individuals. $\mathbb{S}$ is a set of constant symbols, each of which is associated with a first-order sort (i.e.~a sort whose order is at most 2). As $\mathbb{S}$ can be viewed as a
mapping from constant symbols to simple sorts, I write $\mathbb{S}(c)$ for the sort assigned to $c$ by $\mathbb{S}$. Note that because a lambda calculus does not
distinguish between functions and values of base sorts, `constant symbols' in $\mathbb{S}$ include not only those symbols with base sorts but also symbols of arrows types; i.e.~function symbols. 

Given $\Sigma = (\mathbb{B}, \mathbb{S})$, terms are inductively defined by
\begin{equation*}
M, N ::= x \mid c \mid M \ N \mid \lambda x {:} \sigma. M,
\end{equation*} 
where $x$ is a variable and $c \in \mathbb{S}$. Standardly, function application associates to the left. Also, the scopes of lambda abstractions extend as far to the right as possible. 
If a term $M$ has sort $\sigma_1 \to \cdots \to \sigma_m \to b$, where $b \in \mathbb{B}$, the arity\index{arity} of $M$ is defined as
\begin{equation*}
{\bf ar}(M) = m.
\end{equation*}
The set of free variables\index{free variables} occurring in term $M$ is denoted by ${\tt FV}(M)$. 

A sort environment\index{sort!environment} $\Delta$ is a finite sequence of pairs $x: \sigma$, where $x$ is a variable and $\sigma$ is a simple type. The sort environment is required to have no conflicts; that is, it must
not assign multiple sorts to the same variable. The sorts of terms\index{term} are defined by the following sorting rules:
\begin{center}
\AxiomC{}
\LeftLabel{\sc (SCst)}
\UnaryInfC{$\Delta \vdash c: \mathbb{S}(c)$}
\DisplayProof
\qquad
\AxiomC{}
\LeftLabel{\sc (SVar)}
\UnaryInfC{$\Delta_1, x: \sigma, \Delta_2 \vdash x: \sigma$}
\DisplayProof
\end{center}
\begin{prooftree}
\AxiomC{$\Delta \vdash s: \sigma_1 \to \sigma_2$}
\AxiomC{$\Delta \vdash t: \sigma_1$}
\LeftLabel{\sc (SApp)}
\BinaryInfC{$\Delta \vdash s \ t: \sigma_2$}
\end{prooftree}
\begin{prooftree}
\AxiomC{$\Delta, x : \sigma_1 \vdash s: \sigma_2$}
\LeftLabel{\sc (SAbs)}
\RightLabel{$x \notin \text{dom}(\Delta)$}
\UnaryInfC{$\Delta \vdash \lambda x {:} \sigma_1. s: \sigma_1 \to \sigma_2$}
\end{prooftree}

Notice that the sorts of constant symbols are specified by a signature, whilst the sorts of free variables are specified by a sort environment. 

Next, to define formulas of higher-order logic, logical connectives are introduced as constant symbols outside $\Sigma$. Let {\tt LSym} be the set of the following logical constant symbols\index{logical constant symbols}:
\begin{align*}
{\tt true}, {\tt false} & : o & \neg & : o \to o \\
\land, \lor, \Rightarrow & : o \to o \to o & \forall_{\sigma}, \exists_{\sigma} & : (\sigma \to o) \to o.
\end{align*}
I adopt the convention that $\exists_{\sigma}(\lambda x {:} \sigma. M)$ is shortened to $\exists x {:} \sigma. M$ or $\exists_{\sigma} x. M$. Furthermore, if the sort of $x$ is clear from the context, $\exists x. M$
can be written. 

Formulas\index{formula} are defined as well-sorted terms that have the sort $o$ and whose constant symbols are from either $\mathbb{S}$ or {\tt LSym}.

Lastly, relational sorts\index{relational sort} are formally defined by
\begin{equation*}
\rho ::= o \mid b \to o \mid \rho_1 \to \rho_2,
\end{equation*}
where $b \in \mathbb{B}$. 

\subsubsection{Semantics}
Given a first-order signature\index{first-order signature} $\Sigma = (\mathbb{B}, \mathbb{S})$, a structure $A$ assigns a non-empty set of elements $A_{\iota}$ to each $\iota \in \mathbb{B}$, where $\iota \neq o$.
The sets $A_{\iota}$ are often called universes\index{universe}. To the sort $o$ is assigned the distinguished lattice\index{lattice} $\mathbbm{2} = \{0 \leq 1 \}$. The full sort frame\index{full sort frame} over $A$ is defined inductively on a sort
as follows:
\begin{alignat*}{2}
\mathcal{S} \llbracket \iota \rrbracket &:= A_{\iota} & \qquad \iota \in \mathbb{B}, \iota \neq o \\
\mathcal{S} \llbracket o \rrbracket & := \mathbbm{2} \\
\mathcal{S} \llbracket \sigma_1 \to \sigma_2 \rrbracket & := \mathcal{S} \llbracket \sigma_1 \rrbracket \Rightarrow \mathcal{S} \llbracket \sigma_2 \rrbracket,
\end{alignat*}
where $X \Rightarrow Y$ is the full set-theoretic function space between sets $X$ and $Y$. To each constant symbol $c$ in $\mathbb{S}$, $A$ assigns an element from $\mathcal{S} \llbracket \mathbb{S}(c) \rrbracket$.
Let $c^{A}$ denote this element. 

The lattice $\mathbbm{2}$ supports the following functions:
\begin{align*}
{\tt or}(b_1)(b_2) & = \max \{b_1, b_2 \} & {\tt not}(b) & = 1- b \\
{\tt and}(b_1)(b_2) & = \min \{b_1, b_2 \} & {\tt implies}(b_1)(b_2) & = {\tt or}({\tt not}(b_1))(b_2) \\
{\tt exists}_{\sigma}(f) & = \max \{f(v) \mid v \in \mathcal{S} \llbracket \sigma \rrbracket \} & {\tt forall}_{\sigma}(f) & = {\tt not}({\tt exists}_{\sigma}({\tt not} \circ f)).
\end{align*}
For each logical constant symbol $c \in \text{\tt LSym}$, I denote the corresponding Boolean function given above by $c^{\tt LFun}$.

The order on $\mathbbm{2}$ can be extended to define an order $\subseteq_{\rho}$ on $\mathcal{S} \llbracket \rho \rrbracket$, where $\rho$ is a relational sort:
\begin{itemize}
\item For all $b_1, b_2 \in \mathcal{S} \llbracket o \rrbracket$, if $b_1 \leq b_2$, then $b_2 \subseteq_{o} b_2$;
\item For all $r_1, r_2 \in \mathcal{S} \llbracket b \to \rho \rrbracket$, if $r_1(n) \subseteq_{\rho} r_2(n)$ for all $n \in \mathcal{S} \llbracket b \rrbracket$, then $r_1 \subseteq_{b \to \rho} r_2$;
\item For all $r_1, r_2 \in \mathcal{S} \llbracket \rho_1 \to \rho_2 \rrbracket$, if $r_1(s) \subseteq_{\rho} r_2(s)$ for all $s \in \mathcal{S} \llbracket \rho_1 \rrbracket$, then $r_1 \subseteq_{\rho_1 \to \rho_2} r_2$.
\end{itemize}
The full sort frame\index{full sort frame!sort environment} can be defined on a sort environment $\Delta$ using an indexed Cartesian product:
\begin{equation*}
\mathcal{S} \llbracket \Delta \rrbracket := \prod_{x \in \text{dom}(\Delta)} \mathcal{S} \llbracket \Delta(x) \rrbracket.
\end{equation*}
In other words, this is the set of all functions mapping each variable $x$ in $\text{dom}(\Delta)$ to an element in $\mathcal{S} \llbracket \Delta(x) \rrbracket$. These functions are called valuations\index{valuation}.
The order on $\mathcal{S}\llbracket \Delta \rrbracket$ can be defined in the same fashion as above: for all $f_1, f_2 \in \mathcal{S}\llbracket \Delta \rrbracket$, if $f_1(x) \subseteq_{\rho} f_2(2)$
for all $x : \rho \in \Delta$, then $f_1\subseteq_{\Delta} f_2$. 

The interpretation\index{interpretation} of a term $\Delta \vdash M: \sigma$ is given by an inductively defined function $\mathcal{S} \llbracket \Delta \vdash M: \sigma\rrbracket : \mathcal{S} \llbracket \Delta \rrbracket \Rightarrow
\mathcal{S} \llbracket \sigma \rrbracket$. When $M$ consists only of one symbol, $\mathcal{S} \llbracket \Delta \vdash M: \sigma\rrbracket$ is defined by
\begin{alignat*}{2}
\mathcal{S} \llbracket \Delta \vdash x: \sigma \rrbracket (\alpha) & = \alpha(x) && \qquad \text{if } x \text{ is a variable} \\
\mathcal{S} \llbracket \Delta \vdash c: \sigma \rrbracket (\alpha) & = c^{A} && \qquad \text{if } c \in \mathbb{S} \\
\mathcal{S} \llbracket \Delta \vdash c: \sigma \rrbracket (\alpha) & = c^{\tt LFun} && \qquad \text{otherwise},
\end{alignat*}
where $\alpha$ is a valuation from $\mathcal{S} \llbracket \Delta \rrbracket$. If $M$ has a compound structure, we have
\begin{align*}
\mathcal{S} \llbracket \Delta \vdash M \ N: \sigma_2 \rrbracket (\alpha) & = \mathcal{S} \llbracket \Delta \vdash M: \sigma_1 \to \sigma_2 \rrbracket (\alpha) (\mathcal{S} \llbracket \Delta \vdash N: \sigma_1 \rrbracket (\alpha)) \\
\mathcal{S} \llbracket \Delta \vdash \lambda x {:} \sigma_1. M : \sigma_1 \to \sigma_2 \rrbracket (\alpha) & = \lambda v \in \mathcal{S} \llbracket \sigma_1 \rrbracket. 
\mathcal{S} \llbracket \Delta, x: \sigma_1 \vdash M:\sigma_2 \rrbracket (\alpha [x \mapsto v]).
\end{align*}
Notice that the interpretation of non-logical constant symbols is given by a structure, whereas the interpretation of free variables is given by a valuation.

Assume we are given a $\Sigma$-structure $A$, a formula $\Delta \vdash M: o$, and a valuation $\alpha \in \mathcal{S} \llbracket \Delta \rrbracket$. Then $\langle A, \alpha \rangle$ satisfies $M$
if and only if $\mathcal{S} \llbracket \Delta \vdash M: o \rrbracket (\alpha) = 1$. This satisfaction relation is denoted by $A, \alpha \vDash M$. 

\subsection{Logic program safety problems}
\label{sec:safety-problem}
Each verification problem comprises two components: a definite formula component, 
which describes an input program, and a goal formula component, which is the property of the input program that we want to verify. This subsection introduces verification problems whose
definite formula components are expressed using logic programs. Again, the presentation style of this subsection follows that in \cite{Ramsay2017}.

\subsection{Constraint languages}
Given a first-order signature $\Sigma$, a constraint language\index{constraint language} is defined as $(Tm, Fm, Th)$, where $Tm$ is a distinguished subset of first-order terms that can be built from $\Sigma$,
$Fm$ is a distinguished subset of first-order formulas that can be built from $\Sigma$, and $Th$ is a theory in which to interpret $Fm$. Any formula from $Fm$ is called a constraint\index{constraint} and
$Th$ is called a background theory\index{background theory}. We allow $Tm$ and $Fm$ to be strict subsets of all terms and formulas built from $\Sigma$ as some background theories only consider strict subsets
of formulas; e.g.~quantifier-free formulas.

In this document, formulas in a constraint language refer to terms of sort $o$. Therefore, we have $Fm \subseteq Tm$, unlike in usual presentations of predicate logic, where $Tm \cap Fm = \emptyset$. 

\subsubsection{Goal terms} \label{subsection on goal terms}
Fix a first-order signature $\Sigma = (\mathbb{B}, \mathbb{S})$ and a constraint language $(Tm, Fm, Th)$ over $\Sigma$. The class of well-sorted goal terms\index{goal term}\index{term!goal term|see {goal term}} $\Delta\vdash G: \rho$, where $\rho$ is a relational sort, 
is given by these sorting rules:
\begin{prooftree}
\AxiomC{}
\LeftLabel{\sc (GCst)}
\RightLabel{$c \in \{\land, \lor, \exists_{\iota} \} \cup \{ \exists_{\rho} \mid \rho \}$}
\UnaryInfC{$\Delta \vdash c: \rho_{c}$}
\end{prooftree}
\begin{prooftree}
\AxiomC{}
\LeftLabel{\sc (GVar)}
\UnaryInfC{$\Delta_1, x: \rho, \Delta_2 \vdash x: \rho$}
\end{prooftree}
\begin{prooftree}
\AxiomC{}
\LeftLabel{\sc (GConstr)}
\RightLabel{$\Delta \vdash \varphi: o \in Fm$}
\UnaryInfC{$\Delta \vdash \varphi: o$}
\end{prooftree}
\begin{prooftree}
\AxiomC{$\Delta, x: \sigma \vdash G: \rho$}
\LeftLabel{\sc (GAbs)}
\RightLabel{$x \notin \text{\tt dom}(\Delta)$}
\UnaryInfC{$\Delta \vdash \lambda x {:} \sigma. G: \sigma \to \rho$}
\end{prooftree}
\begin{prooftree}
\AxiomC{$\Delta \vdash G: b \to \rho$}
\LeftLabel{\sc (GAppl)}
\RightLabel{$\Delta \vdash N: b \in Tm$}
\UnaryInfC{$\Delta \vdash G \ N: \rho$}
\end{prooftree}
\begin{prooftree}
\AxiomC{$\Delta \vdash G: \rho_1 \to \rho_2$}
\AxiomC{$\Delta \vdash H: \rho_1$}
\LeftLabel{\sc (GAppR)}
\BinaryInfC{$\Delta \vdash G \ H: \rho_2$}
\end{prooftree}
Throughout the above six rules, $b$ denotes a base sort from $\mathbb{B}$, $\rho$ (with or without subscripts) denotes a relational sort, and $\sigma$ is either a base sort or a relational sort.
Henceforth, I assume that goal terms are well-sorted. 

\subsubsection{Logic programs}
Assume that a first-order signature and a constraint language are fixed. A higher-order constrained logic program\index{logic program} $P$ over a sort environment $\Delta = x_1: \rho_1, \ldots, x_m: \rho_m$, 
where each $\rho_i$ is a relational sort, is a finite system of (mutual) recursive definitions of shape:
\begin{equation*}
x_1 {:} \rho_1 = G_1, \ldots, x_m {:} \rho_m = G_m,
\end{equation*}
where each $G_i$ is a goal term and each $x_i$ is distinct. I will call each $x_i$ a top-level relational variable\index{top-level relational variable}. $P$ is said to be well-sorted whenever $\Delta \vdash G_i: \rho_i$ (i.e.~$G_i$ is well-sorted and
has relational sort $\rho_i$) for each $1 \leq i \leq m$. It follows that if $P$ is well-sorted, 
\begin{equation*}
{\tt FV}(G_i) \subseteq \{x_1, \ldots, x_m \}
\end{equation*}
for all $1 \leq i \leq m$. 

Since each $x_i$ is distinct, we can regard $P$ as a finite map from variables to goal terms. Thus, let $P(x_i)$ denote the goal term $G_i$ that is bound to $x_i$.
I write $\vdash P: \Delta$ to mean that $P$ is a well-sorted program over $\Delta$. 

To interpret logic programs\index{logic program!interpretation}, I use the standard semantics\index{standard semantics}. Let $A$ be a $\Sigma$-structure and $P$ be a well-sorted logic program over a sort environment $\Delta$. 
The one-step consequence operator\index{one-step consequence operator} of $P$ is the functional $T^{\mathcal{S}}_{P: \Delta} : \mathcal{S} \llbracket \Delta \rrbracket \Rightarrow \mathcal{S} \llbracket \Delta \rrbracket$ defined by
\begin{equation*}
T^{\mathcal{S}}_{P: \Delta}(\alpha)(x) = \mathcal{S} \llbracket \Delta \vdash P(x): \Delta(x) \rrbracket (\alpha).
\end{equation*}
A valuation $\alpha$ is a prefixed point of $T^{\mathcal{S}}_{P: \Delta}$ if and only if we have $T^{\mathcal{S}}_{P: \Delta} (\alpha) \subseteq_{\Delta} \alpha$. 

\subsubsection{Logic program safety problems}
Suppose that $\Sigma$ is a first-order signature and $L = (Tm, Fm, Th)$ is a constraint language over $\Sigma$. A logic program safety problem\index{logic program safety problem} is defined as a triple $(\Delta, P, G)$, where $\Delta$ is a
sort environment of relational variables, $P$ is a well-sorted logic program over $\Delta$, and $G$ is a goal term that has sort $o$ and is built from $\Sigma$ and $L$. The problem is solvable\index{solvable} if and only if for all models of $Th$,
there exists a valuation $\alpha$ such that $\alpha$ is a prefixed point of $T^{\mathcal{S}}_{P: \Delta}$ and $\mathcal{S} \llbracket \Delta \vdash G: o \rrbracket (\alpha) = 0$. $G$ is usually the negation of a property
that we want $P$ to satisfy. 

\subsection{Monotone semantics}

The monotone semantics\index{monotone semantics} for logic programs is introduced by Cathcart Burn et al.~\cite{Ramsay2017} as an alternative to the standard semantics. 
The importance of the monotone semantics in defunctionalization will be explained in Subsection~\ref{section on the importance of the monotone semantics}.

\subsubsection{Semantics}

Given a first-order signature $\Sigma = (\mathbb{B}, \mathbb{S})$, structure $A$ assigns a non-empty discrete poset\index{discrete poset}\index{poset}\index{poset!discrete|see {discrete poset}} $A_{\iota}$ to each $\iota \in \mathbb{B}$, 
where $\iota \neq o$. As in the standard semantics, $A$ assigns $\mathbbm{2}$ to the sort $o$. Discrete posets are defined as partially ordered sets in which no two distinct elements are comparable. 
The monotone sort frame\index{monotone sort frame} over $A$ is then inductively defined as
\begin{equation*}
\mathcal{M} \llbracket \iota \rrbracket := A_{\iota} \qquad \mathcal{M} \llbracket o \rrbracket := \mathbbm{2} \qquad \mathcal{M} \llbracket \sigma_1 \to \sigma_2 \rrbracket :=
\mathcal{M} \llbracket \sigma_1 \rrbracket \Rightarrow_{m} \mathcal{M} \llbracket \sigma_2 \rrbracket,
\end{equation*}
where $X \Rightarrow_{m} Y$ is the monotone function space between posets $X$ and $Y$. The universe $A_{\iota}$ is regarded as a discrete ``poset'' rather than simply a set because we want the definition
$\mathcal{M} \llbracket \sigma_1 \to \sigma_2 \rrbracket :=\mathcal{M} \llbracket \sigma_1 \rrbracket \Rightarrow_{m} \mathcal{M} \llbracket \sigma_2 \rrbracket$ to encompass the cases when 
$\sigma_1 \in \mathbb{B}$. Since any set can be considered as a discrete poset, when $\sigma_1 \in \mathbb{B} \setminus \{ o \}$, $\Rightarrow_{m}$ is the same as $\Rightarrow$ in the definition
of full sort frames. 

$A$ also maps each constant symbol $c: \sigma \in \mathbb{S}$ to an element from $\mathcal{M} \llbracket \sigma \rrbracket $.  

The order in $\mathbbm{2}$ is extended to $\mathcal{M} \llbracket \rho \rrbracket$, where $\rho$ is a relational sort, in the same manner as $\mathcal{S} \llbracket \rho \rrbracket$. 
Also, the set of valuations with respect to sort environment $\Delta$ is defined analogously to the standard semantics: 
\begin{equation*}
\mathcal{M} \llbracket \Delta \rrbracket := \prod_{x \in \text{dom}(\Delta)} \mathcal{M} \llbracket \Delta(x) \rrbracket.
\end{equation*}

The monotone interpretation of goal terms is inductively defined in the same way as the standard interpretation. As we consider only monotone functions, the definition of {\tt exists} becomes 
\begin{equation*}
{\tt exists}_{\sigma}(f) = \max\{f(v) \mid v \in \mathcal{M} \llbracket \sigma \rrbracket \}. 
\end{equation*}

Fix a first-order signature, a constraint language, and a structure for interpretation of a logic program. The one-step consequence operator $T_{P: \Delta}^{\mathcal{M}}$ is defined as
\begin{equation*}
T_{P: \Delta}^{\mathcal{M}} (\alpha)(x) = \mathcal{M} \llbracket \Delta \vdash P(x): \Delta(x) \rrbracket (\alpha). 
\end{equation*}
A prefixed point of $T_{P: \Delta}^{\mathcal{M}}$ is called a model of the logic program $P$. The term `model' is overloaded because a model of a logic program is a valuation, whereas a model
of a theory is a structure. 

\subsubsection{Monotone logic program safety problems}

Suppose that $\Sigma$ is a first-order signature and $L = (Tm, Fm, Th)$ is a constraint language over $\Sigma$. A monotone logic program safety problem\index{logic program safety problem!monotone|see {monotone problem}}
(oftentimes abbreviated as a monotone problem\index{monotone problem})
is defined as a triple $(\Delta, P, G)$, where $\Delta$ is a sort environment of relational variables, $P$ is a well-sorted logic program over $\Delta$, and $G$ is a goal formula. Both $P$ and $G$ are built from
$\Sigma$ and $L$. The monotone problem is solvable if and only if for all models of $Th$, there exists a valuation $\alpha$ such that $\alpha$ is a prefixed point of $T^{\mathcal{M}}_{P: \Delta}$ and 
$\mathcal{M} \llbracket \Delta \vdash G: o \rrbracket (\alpha) = 0$. $G$ is usually the negation of a property that we want $P$ to satisfy. 

Theorem~2 in \cite{Ramsay2017} establishes a bridge between constrained Horn clause problems and monotone logic program safety problems:
\begin{theorem} \label{relationship between Horn clause problems and monotone problems original}
A higher-order constrained Horn clause problem $(\Delta, D, G)$ is solvable if and only if the associated monotone logic program safety problem $(\Delta, P_{D}, G)$ is solvable. 
\end{theorem}
The transformation from the definite Horn formula $D$ to the corresponding logic program $P_{D}$ is provided in Section~4.1 of \cite{Ramsay2017}.


\section{Defunctionalization of monotone problems} \label{chapter on defunctionalization of monotone problems with a concrete example}

This section illustrates how defunctionalization works on a concrete example of a logic program safety problem, which is interpreted using the standard semantics. 
An issue that arises from higher-order existential quantification is then explained. The monotone semantics is instrumental in resolving this issue. 

\subsection{Overview} \label{concrete example of defunctionalization}

In this subsection, I will illustrate the workings of the defunctionalization algorithm for logic program safety problems using a concrete example. 
Because the standard semantics is more natural and intuitive than the monotone semantics, the standard semantics allows us to use our own intuition to interpret safety problems. 
Consequently, we can follow how defunctionalization proceeds without being concerned about semantics. 
Therefore, we will use the standard semantics to interpret the example logic program safety problem. 

Henceforth, for readability, I omit subscripts of $\exists$ that denote the sorts of quantified variables. 

Consider the safety problem $\mathcal{P} = (\Delta, P, G)$, where $\Delta$ is given by
\begin{align*}
\Delta = \{ & Main: {\bf nat} \to {\bf natlist} \to o, \\
& TwiceMap: ({\bf nat} \to {\bf nat} \to o) \to {\bf natlist} \to {\bf natlist} \to o, \\
& Map: ({\bf nat} \to {\bf nat} \to o) \to {\bf natlist} \to {\bf natlist} \to o, \\
& Twice: ({\bf nat} \to {\bf nat} \to o) \to {\bf nat} \to {\bf nat} \to o \},
\end{align*}
$P$ is
\begin{align*}
Main & = \lambda n, ns. TwiceMap \ (\lambda a, b. a + n = b) \ ({\tt cons} \ 0 \ {\tt nil}) \ ns \\
TwiceMap & = \lambda f. Map \ (Twice \ f) \\
Map & = \lambda f, a, b. (a = {\tt nil} \land b = {\tt nil}) \lor \\
& \quad (\exists n, ns, m, ms. a = {\tt cons} \ n \ ns \land f \ n \ m \land Map \ f \ ns \ ms \land b = {\tt cons} \ m \ ms) \\
Twice & = \lambda f, a, b. (\exists c. f \ a \ c \land f \ c \ b),
\end{align*}
and $G$ is
\begin{equation*}
G = \exists n, ns. Main \ n \ ns \land ns = {\tt nil}. 
\end{equation*}
The signature for $\mathcal{P}$ is $\Sigma = (\mathbb{B}, \mathbb{S})$, where $\mathbb{B}$ and $\mathbb{S}$ are
\begin{align*}
\mathbb{B} & = \{{\bf nat}, {\bf natlist}, o \} \\
\mathbb{S} & = \{{\tt nil}: {\bf natlist}, {\tt cons}: {\bf nat} \to {\bf natlist} \to {\bf natlist}, +: {\bf nat} \to {\bf nat} \to {\bf nat} \\
& \quad =_{\bf natlist} : {\bf natlist} \to {\bf natlist} \to o, =_{\bf nat} : {\bf nat} \to {\bf nat} \to o \} \cup \{ n: {\bf nat} \mid n \in \mathbb{N} \}.
\end{align*}
To be precise, $\mathbb{S}$ must be finite. However, having all natural numbers included in $\mathbb{S}$ does not affect the fundamental nature of this verification problem.

Observe that $\{ n: {\bf nat} \mid n \in \mathbb{N} \}$ is a set of symbols rather than a set of mathematical entities.

As $=_{\bf natlist}$ and $=_{\bf nat}$ have different types, they must be distinguished. However, I will denote both of them by $=$ for simplicity. From the sorts of their arguments, we can infer which equality is in use. 	

The background theory we use to interpret $\mathcal{P}$ is the one constructed by a structure that maps each $c \in \mathbb{S}$ to the naturally corresponding element in the universe
of natural numbers and their lists. 

Defunctionalization is the conversion of higher-order programs to a semantically equivalent first-order programs. To achieve it, higher-order parameters appearing in the programs need to be removed. 
Parameters are classified into formal parameters\index{formal parameter} and actual parameters\index{actual parameter}. Higher-order formal parameters in logic programs are always found in lambda abstractions. Higher-order actual parameters are
generated by curried functions, which can be (strictly) partially applied. Hence, the higher-order property of logic programs is attributed to
\begin{itemize}
\item Higher-order formal parameters
\item Curried functions.
\end{itemize}

Higher-order formal parameters can be identified by looking at type annotations. However, formal parameters can be missing due to currying. For instance, in the logic program $P$ given above, 
the definition of $TwiceMap$ is a lambda abstraction with only one parameter $f$. However, since the sort of $TwiceMap$ is $({\bf nat} \to {\bf nat} \to o) \to {\bf natlist} \to {\bf natlist} \to o$, 
we have ${\bf ar}(TwiceMap) = 3$. Therefore, $TwiceMap$ has two more formal parameters. The first preprocessing step is thus to uncover hidden formal parameters. This is known as $\eta$-expansion\index{$\eta$-expansion}
in the literature on lambda calculi. We only uncover the formal parameters of outermost lambda abstractions. In $P$, the only place where formal parameters of top-level relational variables are
hidden is $TwiceMap$. Applying $\eta$-expansion to it, we obtain
\comment[sr]{Just give the equation for TwiceMap here.}
\begin{align*}
Main & = \lambda n, ns. TwiceMap \ (\lambda a, b. a + n = b) \ ({\tt cons} \ 0 \ {\tt nil}) \ ns \\
TwiceMap & = \lambda f, xs, ys. Map \ (Twice \ f) \ xs \ ys \\
Map & = \lambda f, a, b. (a = {\tt nil} \land b = {\tt nil}) \\
& \quad \lor (\exists n, ns, m, ms. a = {\tt cons} \ n \ ns \land f \ n \ m \land Map \ f \ ns \ ms \land b = {\tt cons} \ m \ ms) \\
Twice & = \lambda f, a, b. (\exists c. f \ a \ c \land f \ c \ b).
\end{align*}
Note that although $Twice \ f$ inside the definition of $TwiceMap$ has a functional sort and hence its formal parameters are hidden, we do not apply $\eta$-expansion to it, because it is not an outermost
function. Now all formal parameters of outermost lambda abstractions in $P$ are visible. 

Curried functions in logic programs appear either as the definitions of top-level relational variables in the form $x_i {:} \rho_i = G_i$ or as anonymous lambda abstractions. The constant symbols from
$\mathbb{S}$ cannot be strictly partially applied in a logic program. Hence, we do not need to defunctionalize the functions declared in the signature. 

In the logic program above, $\lambda a, b. a + n = b$ inside the definition of $Main$ is an example of anonymous functions\index{anonymous function}. 
Later in the process of defunctionalization, higher-order actual parameters are replaced with first-order data, each of which is labelled with an associated curried function. 
Whilst top-level relational variables can be used as labels for the functions that define them, anonymous functions do not have any unique name or variable associated with it.
Therefore, for convenience, I create fresh top-level relational variables for anonymous functions so that they are not `anonymous' anymore. 
Consequently, $P$ becomes
\comment[sr]{Consequently, the equation for $Main$ becomes... }
\begin{align*}
Main & = \lambda n, ns. TwiceMap \ (Add \ n) \ ({\tt cons} \ 0 \ {\tt nil}) \ ns \\
TwiceMap & = \lambda f, ns, ms. Map \ (Twice \ f) \ ns \ ms \\
Map & = \lambda f, a, b. (a = {\tt nil} \land b = {\tt nil}) \\
& \quad \lor (\exists n, ns, m, ms. a = {\tt cons} \ n \ ns \land f \ n \ m \land Map \ f \ ns \ ms \land b = {\tt cons} \ m \ ms) \\
Twice & = \lambda f, a, b. (\exists c. f \ a \ c \land f \ c \ b) \\
Add & = \lambda n, a, b. a + n = b.
\end{align*}

Next, higher-order parameters are replaced with first-order data of a new base type. Let us denote the new type {\bf closr}, which is short for `closure'\index{closure}. Because I will defunctionalize $P$ step by step,
the intermediate states may not be valid logic programs. 

Partially applied curried functions are represented by algebraic data types. They store all actual parameters that have been supplied so far. It is necessary to define distinct data constructors
according to the number of actual parameters. For instance, partially applied instances of function $Add$ can take one of the following forms:
\begin{align*}
& Add \\
& Add \ x \\
& Add \ x \ y \\
& Add \ x \ y \ z,
\end{align*}
where the last form is not strictly partially applied. 

Let $C^i_{F}$ be a data constructor that, after $i$ many arguments are supplied, represents a partially applied instance of function $F$ with $i$ many actual parameters. To simulate lambda/function
application using these constructors, we need to define function $Apply$. For the top-level relational variable $Main$ defined in $P$, the corresponding $Apply$ is given by
\begin{align*}
Apply & = \lambda x, y, z. x = C^0_{Main} \land z = C^1_{Main}\ y \\
Apply & = \lambda x, ns. (\exists n. x = C^1_{Main} \ n \land TwiceMap \ (Add \ n) \ ({\tt cons} \ 0 \ {\tt nil}) \ ns).
\end{align*}
Notice that $TwiceMap \ (Add \ n) \ ({\tt cons} \ 0 \ {\tt nil}) \ ns$ in the second line is derived from the definition of $Main$ in $P$. $Apply$ on the first line has sort 
${\bf closr} \to {\bf nat} \to {\bf closr} \to o$ and takes three arguments. The first argument represents a partially applied instance of $Main$. The second argument is an input to the function represented
by the first argument. The third argument is the result of applying the second argument to the function represented by the first argument. 

By contrast, $Apply$ on the second line has sort ${\bf closr} \to {\bf natlist} \to o$ and has arity 2. The first argument represents a partially applied instance of $Main$ in which only the last parameter
of $Main$ is missing. The second argument corresponds to this missing final parameter. In a logic program, if a top-level relational sort $F$ has arity $n$, the first $n - 1$ parameters of $F$ can be
interpreted as inputs (as in functional programming) and the last parameter of $F$ can be interpreted as the corresponding output. Using this interpretation, the second $Apply$ is considered as linking
the input and output of $Main$. 

In this way, the first $Apply$ function simulates function application, whereas the second $Apply$ works out whether the first argument evaluates to the second argument. As these two $Apply$ functions
have different roles, I rename the second $Apply$ to `IOMatch'. This yields
\begin{align*}
Apply & = \lambda x, y, z. x = C^0_{Main} \land z = C^1_{Main}\ y \\
IOMatch & = \lambda x, ns. (\exists n. x = C^1_{Main} \ n \land TwiceMap \ (Add \ n) \ ({\tt cons} \ 0 \ {\tt nil}) \ ns).
\end{align*}
Applying the same step to the remaining top-level relational variables in $P$ gives
\comment[sr]{Could be moved out to the appendix.}
\begin{align*}
Apply & = \lambda x, y, z. x = C^0_{Main} \land z = C^1_{Main}\ y \\
IOMatch & = \lambda x, ns. (\exists n. x = C^1_{Main} \ n \land TwiceMap \ (Add \ n) \ ({\tt cons} \ 0 \ {\tt nil}) \ ns) \\[10pt]
Apply & = \lambda x, y, z. x= C^0_{TwiceMap} \land z = C^1_{TwiceMap} \ y \\
Apply & = \lambda x, y, z. (\exists f. x = C^1_{TwiceMap} \ f \land z = C^2_{TwiceMap} \ f \ y) \\
IOMatch & = \lambda x, ms. (\exists f, ns. x = C^2_{TwiceMap}\ f \ ns \land Map \ (Twice \ f) \ ns \ ms) \\[10pt]
Apply & = \lambda x, y, z. x = C^0_{Map} \land z = C^1_{Map} \ y \\
Apply & = \lambda x, y, z. (\exists f. x = C^1_{Map} \ f \land z = C^2_{Map} \ f \ y) \\
IOMatch & = \lambda x, b. (\exists f, a. x = C^2_{Map} \ f \ a \land ((a = {\tt nil} \land b = {\tt nil}) \\
& \quad \lor (\exists n, ns, m, ms. a = {\tt cons} \ n \ ns \land f \ n \ m \land Map \ f \ ns \ ms \land b = {\tt cons} \ m \ ms))) \\[10pt]
Apply & = \lambda x, y, z. x = C^0_{Twice} \land z = C^1_{Twice} \ y \\
Apply & = \lambda x, y, z. (\exists f. x = C^1_{Twice} \ f \land z = C^2_{Twice} \ f \ y) \\
IOMatch & = \lambda x, b. (\exists f, a. x = C^2_{Twice} \ f \ a \land (\exists c. f \ a \ c \land f \ c \ b)) \\[10pt]
Apply & = \lambda x, y, z. x = C^0_{Add} \land z = C^1_{Add} \ y \\
Apply & = \lambda x, y, z. (\exists n. x = C^1_{Add} \ n \land z = C^2_{Add} \ n \ y) \\
IOMatch & = \lambda x, b. (\exists n, a. x = C^2_{Add} \ n \ a \land a + n = b).
\end{align*}
In the original $P$, $f$ is used as a higher-order formal parameter, but here it has sort {\bf closr}. The new logic program has multiple equations for $Apply$ and $IOMatch$, which violates the rule that
every top-level relational variable must be distinct. Each of the equations defines a conditional branch of a relational variable. Hence, we combine them by taking their disjunction. Specifically, suppose we have
\begin{align*}
X & = \lambda x_1, \ldots, x_n. G_1 \\ 
X & = \lambda y_1, \ldots, y_n. G_2.
\end{align*}
Then the disjunction of these two equations is given by
\begin{equation*}
X = \lambda x_1, \ldots, x_n. (G_1 \lor G_2 [x_1 / y_1] \cdots [x_n / y_n]),
\end{equation*}
where $G_2 [x_i / y_i]$ denotes the result of substituting $x_i$ for every free occurrence of $y_i$ in $G_2$. I assume that $x_i$ does not occur bound in $G_2$ (or at least the substitution $[x_i / y_i]$ does
not cause variable capture) for all $1 \leq i \leq n$. For readability, however, I will leave the logic program unchanged. 

At this point, $Apply$ is not well-sorted, since the second arguments take various sorts; e.g.~{\bf nat}, {\bf natlist}, and {\bf closr}. 
Therefore, we must create clones of $Apply$, each specializing in  a particular sort of the second argument. 
Let $Apply_{A}$ denote a clone of $Apply$ whose sort is ${\bf closr} \to A \to {\bf closr} \to o$. 
Similarly for $IOMatch$, I write $IOMatch_{A}$ for a clone of $IOMatch$ whose sort is ${\bf closr} \to A \to o$. 
Inserting these clones to appropriate places of the program, we obtain
\comment[sr]{Instead just give one example here and move a copy to the appendix.}
\begin{align*}
Apply_{\bf nat} & = \lambda x, y, z. x = C^0_{Main} \land z = C^1_{Main}\ y \\
Apply_{\bf nat} & = \lambda x, y, z. (\exists f. x = C^1_{Twice} \ f \land z = C^2_{Twice} \ f \ y) \\
Apply_{\bf nat} & = \lambda x, y, z. x = C^0_{Add} \land z = C^1_{Add} \ y \\
Apply_{\bf nat} & = \lambda x, y, z. (\exists n. x = C^1_{Add} \ n \land z = C^2_{Add} \ n \ y) \\[10pt]
Apply_{\bf natlist} & = \lambda x, y, z. (\exists f. x = C^1_{TwiceMap} \ f \land z = C^2_{TwiceMap} \ f \ y) \\
Apply_{\bf natlist} & = \lambda x, y, z. (\exists f. x = C^1_{Map} \ f \land z = C^2_{Map} \ f \ y) \\[10pt]
Apply_{\bf closr} & = \lambda x, y, z. x= C^0_{TwiceMap} \land z = C^1_{TwiceMap} \ y \\
Apply_{\bf closr} & = \lambda x, y, z. x = C^0_{Map} \land z = C^1_{Map} \ y \\
Apply_{\bf closr} & = \lambda x, y, z. x = C^0_{Twice} \land z = C^1_{Twice} \ y \\[10pt]
IOMatch_{\bf nat} & = \lambda x, b. (\exists f, a. x = C^2_{Twice} \ f \ a \land (\exists c. f \ a \ c \land f \ c \ b)) \\
IOMatch_{\bf nat} & = \lambda x, b. (\exists n, a. x = C^2_{Add} \ n \ a \land a + n = b) \\[10pt]
IOMatch_{\bf natlist} & = \lambda x, ns. (\exists n. x = C^1_{Main} \ n \land TwiceMap \ (Add \ n) \ ({\tt cons} \ 0 \ {\tt nil}) \ ns) \\
IOMatch_{\bf natlist} & = \lambda x, ms. (\exists f, ns. x = C^2_{TwiceMap}\ f \ ns \land Map \ (Twice \ f) \ ns \ ms) \\
IOMatch_{\bf natlist} & = \lambda x, b. (\exists f, a. x = C^2_{Map} \ f \ a \land ((a = {\tt nil} \land b = {\tt nil}) \\
& \quad \lor (\exists n, ns, m, ms. a = {\tt cons} \ n \ ns \land f \ n \ m \land Map \ f \ ns \ ms \land b = {\tt cons} \ m \ ms))). 
\end{align*}

On the right hand sides of equations defining $IOMatch_{A}$, we have function application that involves formal higher-order parameters such as $f$. However, because their sorts are changed to
{\bf closr} by defunctionalization, the next step is to insert $C^i_{F}$ , $Apply_{A}$, and $IOMatch_{A}$ to the right hand sides of the equations defining $IOMatch_{A}$. For example, the first equation of
$IOMatch_{{\bf nat}}$ has the function application $f \ a \ c$. This is transformed into
\begin{equation*}
\exists d. Apply_{\bf nat} \ f \ a \ d \land IOMatch_{\bf nat} \ d \ c.
\end{equation*}
Analogously, $TwiceMap \ (Add \ n) \ ({\tt cons} \ 0 \ {\tt nil}) \ ns$ in the first equation of $IOMatch_{\bf natlist}$ is transformed into
\begin{align*}
& \exists a. Apply_{\bf nat} \ C^0_{Add} \ n \ a \\
& \land (\exists b, c. Apply_{\bf closr} \ C^0_{TwiceMap} \ a \ b \land Apply_{\bf natlist} \ b \ ({\tt cons} \ 0 \ {\tt nil}) \ c \land IOMatch_{\bf natlist} \ c \ ns).
\end{align*}
Applying the same step to all the remaining equations that define $IOMatch_{A}$ in $P$ gives
\comment[sr]{Similar}
\begin{align*}
IOMatch_{\bf nat} & = \lambda x, b. (\exists f, a. x = C^2_{Twice} \ f \ a \land (\exists c. (\exists d. Apply_{\bf nat} \ f \ a \ d \land IOMatch_{\bf nat} \ d \ c) \\
& \quad \land (\exists e. Apply_{\bf nat} \ f \ c \ e \land IOMatch_{\bf nat} \ e \ b))) \\
IOMatch_{\bf nat} & = \lambda x, b. (\exists n, a. x = C^2_{Add} \ n \ a \land a + n = b) \\[10pt]
IOMatch_{\bf natlist} & = \lambda x, ns. (\exists n. x = C^1_{Main} \ n \land (\exists a. Apply_{\bf nat} \ C^0_{Add} \ n \ a \\
& \quad \land (\exists b, c. Apply_{\bf closr} \ C^0_{TwiceMap} \ a \ b \land Apply_{\bf natlist} \ b \ ({\tt cons} \ 0 \ {\tt nil}) \ c \\
& \qquad \land IOMatch_{\bf natlist} \ c \ ns))) \\
IOMatch_{\bf natlist} & = \lambda x, ms. (\exists f, ns. x = C^2_{TwiceMap}\ f \ ns \land (\exists a. Apply_{\bf closr} \ C^0_{Twice} \ f \ a \\
& \quad \land (\exists b, c. Apply_{\bf closr} \ C^0_{Map} \ a \ b \land Apply_{\bf natlist} \ b \ ns \ c \land IOMatch_{\bf natlist} \ c \ ms))) \\
IOMatch_{\bf natlist} & = \lambda x, b. (\exists f, a. x = C^2_{Map} \ f \ a \land ((a = {\tt nil} \land b = {\tt nil}) \\
& \quad \lor (\exists n, ns, m, ms. a = {\tt cons} \ n \ ns \\
& \qquad \land (\exists c. Apply_{\bf nat} \ f \ n \ c \land IOMatch_{\bf nat} \ c \ m) \\
& \qquad \land (\exists d, e. Apply_{\bf closr} \ C^0_{Map} \ f \ d \land Apply_{\bf natlist} \ d \ ns \ e \land IOMatch_{\bf natlist} \ e \ ms) \\
& \qquad \land b = {\tt cons} \ m \ ms))). 
\end{align*}
Here we use type annotations to determine appropriate clones of $Apply$ and $IOMatch$ to be used. As every function application is now done through a clone of $Apply$ and $IOMatch$, 
$Apply_{A}$ and $IOMatch_{A}$ are self-contained. Hence, we can delete all equations defining the top-level relational variables from the source program. This completes the defunctionalization of $P$. 

Since the top-level relational variables in the original $P$ are removed, we need to defunctionalize the goal formula component $G$ as well. It produces
\begin{equation*}
G' = \exists n, ns. ((\exists a. Apply_{\bf nat} \ C^0_{Main} \ n \ a \land IOMatch_{\bf natlist} \ a \ ns) \land ns = {\tt nil}).
\end{equation*}

\comment[sr]{Could just refer to the appendix.}
To sum up, the defunctionalized logic program safety problem is $\mathcal{P}' = (\Delta', P', G')$ with the new signature $\Sigma' = (\mathbb{B}', \mathbb{S}')$, where $\mathbb{B}'$ and
$\mathbb{S}'$ are given by
\begin{equation*}
\mathbb{B}' = \mathbb{B} \cup \{ {\bf closr} \}
\end{equation*}
and
\begin{align*}
\mathbb{S}' = \mathbb{S} \cup \{ & C^0_{Main}: {\bf closr}, C^1_{Main}: {\bf nat} \to {\bf closr}, \\
& C^0_{TwiceMap}: {\bf closr}, C^1_{TwiceMap}: {\bf closr} \to {\bf closr}, C^2_{TwiceMap}: {\bf closr} \to {\bf natlist} \to {\bf closr}, \\
& C^0_{Map}: {\bf closr}, C^1_{Map}: {\bf closr} \to {\bf closr}, C^2_{Map}: {\bf closr} \to {\bf natlist} \to {\bf closr}, \\
& C^0_{Twice}: {\bf closr}, C^1_{Twice}: {\bf closr} \to {\bf closr}, C^2_{Twice}: {\bf closr} \to {\bf nat} \to {\bf closr}, \\
& C^0_{Add}: {\bf closr}, C^1_{Add}: {\bf nat} \to {\bf closr}, C^2_{Add}: {\bf nat} \to {\bf nat} \to {\bf closr} \}.
\end{align*}
The new sort environment is 
\begin{align*}
\Delta'= \{ & Apply_{\bf nat}: {\bf closr} \to {\bf nat} \to {\bf closr} \to o, Apply_{\bf natlist}: {\bf closr} \to {\bf natlist} \to {\bf closr} \to o, \\
& Apply_{\bf closr}: {\bf closr} \to {\bf closr} \to {\bf closr} \to o, IOMatch_{nat}: {\bf closr} \to {\bf nat} \to o, \\
& IOMatch_{natlist}: {\bf closr} \to {\bf natlist} \to o \}.
\end{align*}
$P'$ consists of
\begin{align*}
Apply_{\bf nat} & = \lambda x, y, z. x = C^0_{Main} \land z = C^1_{Main}\ y \\
Apply_{\bf nat} & = \lambda x, y, z. (\exists f. x = C^1_{Twice} \ f \land z = C^2_{Twice} \ f \ y) \\
Apply_{\bf nat} & = \lambda x, y, z. x = C^0_{Add} \land z = C^1_{Add} \ y \\
Apply_{\bf nat} & = \lambda x, y, z. (\exists n. x = C^1_{Add} \ n \land z = C^2_{Add} \ n \ y) \\[10pt]
Apply_{\bf natlist} & = \lambda x, y, z. (\exists f. x = C^1_{TwiceMap} \ f \land z = C^2_{TwiceMap} \ f \ y) \\
Apply_{\bf natlist} & = \lambda x, y, z. (\exists f. x = C^1_{Map} \ f \land z = C^2_{Map} \ f \ y) \\[10pt]
Apply_{\bf closr} & = \lambda x, y, z. x= C^0_{TwiceMap} \land z = C^1_{TwiceMap} \ y \\
Apply_{\bf closr} & = \lambda x, y, z. x = C^0_{Map} \land z = C^1_{Map} \ y \\
Apply_{\bf closr} & = \lambda x, y, z. x = C^0_{Twice} \land z = C^1_{Twice} \ y \\[10pt]
IOMatch_{\bf nat} & = \lambda x, b. (\exists f, a. x = C^2_{Twice} \ f \ a \land (\exists c. (\exists d. Apply_{\bf nat} \ f \ a \ d \land IOMatch_{\bf nat} \ d \ c) \\
& \quad \land (\exists e. Apply_{\bf nat} \ f \ c \ e \land IOMatch_{\bf nat} \ e \ b))) \\
IOMatch_{\bf nat} & = \lambda x, b. (\exists n, a. x = C^2_{Add} \ n \ a \land a + n = b) \\[10pt]
IOMatch_{\bf natlist} & = \lambda x, ns. (\exists n. x = C^1_{Main} \ n \land (\exists a. Apply_{\bf nat} \ C^0_{Add} \ n \ a \\
& \quad \land (\exists b, c. Apply_{\bf closr} \ C^0_{TwiceMap} \ a \ b \land Apply_{\bf natlist} \ b \ ({\tt cons} \ 0 \ {\tt nil}) \ c \\
& \qquad \land IOMatch_{\bf natlist} \ c \ ns))) \\
IOMatch_{\bf natlist} & = \lambda x, ms. (\exists f, ns. x = C^2_{TwiceMap}\ f \ ns \land (\exists a. Apply_{\bf closr} \ C^0_{Twice} \ f \ a \\
& \quad \land (\exists b, c. Apply_{\bf closr} \ C^0_{Map} \ a \ b \land Apply_{\bf natlist} \ b \ ns \ c \land IOMatch_{\bf natlist} \ c \ ms))) \\
IOMatch_{\bf natlist} & = \lambda x, b. (\exists f, a. x = C^2_{Map} \ f \ a \land ((a = {\tt nil} \land b = {\tt nil}) \\
& \quad \lor (\exists n, ns, m, ms. a = {\tt cons} \ n \ ns \\
& \qquad \land (\exists c. Apply_{\bf nat} \ f \ n \ c \land IOMatch_{\bf nat} \ c \ m) \\
& \qquad \land (\exists d, e. Apply_{\bf closr} \ C^0_{Map} \ f \ d \land Apply_{\bf natlist} \ d \ ns \ e \land IOMatch_{\bf natlist} \ e \ ms) \\
& \qquad \land b = {\tt cons} \ m \ ms))). 
\end{align*}
$G'$ is 
\begin{equation*}
G' = \exists n, ns. ((\exists a. Apply_{\bf nat} \ C^0_{Main} \ n \ a \land IOMatch_{\bf natlist} \ a \ ns) \land ns = {\tt nil}).
\end{equation*}
$P'$ and $G'$ are well-sorted and indeed first-order. 

Suppose that the model of the background theory for $\mathcal{P}$ is $A$, which interprets the constant symbols from $\mathbb{S}$ in a standard way. A model $A'$ over $\Sigma'$ is then defined as
follows:
\begin{itemize}
\item For all constant symbols inherited from $\Sigma$, $A'$ has the same interpretation as $A$.
\item To sort {\bf closr}, $A'$ assigns a universe of objects created by data constructors $C^i_{F}$. 
\item Functions $C^i_{F}$ are interpreted in a natural way as data constructors for the algebraic data type {\bf closr}. 
\end{itemize}
The background theory for $\mathcal{P}'$ is the background theory for $\mathcal{P}$ extended with additional theorems for the {\bf closr} universe.
It is worth noting that not all functions have their respective representatives in the universe of {\bf closr}. 

$\mathcal{P}$ and $\mathcal{P}'$ have the same semantics in the sense that
\begin{equation*}
\mathcal{P} \text{ is solvable} \iff \mathcal{P}' \text{ is solvable}.
\end{equation*}
This is because no relational variable is quantified in $G$. However, if $G$ has quantified relational variables, there is a problem as discussed in the next subsection.

\subsection{Quantification of higher-order variables}

In the example of Subsection~\ref{concrete example of defunctionalization}, we do not have existential quantifiers over variables of order more than 1. 
If we had higher-order existential quantifiers, we would have an issue with preserving the semantics (i.e.~solvability) of $\mathcal{P}$. 

By way of example, suppose $G = \exists f. Twice \ f \ 1 \ 3$ and that $P$ only contains $Twice$ defined as above. The resulting safety problem is not solvable, since for every prefixed point of $P$, 
\begin{equation*}
f = \lambda x, y. (x + 1 = y)
\end{equation*}
makes $Twice \ f \ 1 \ 3$ hold. However, one reasonable way to defunctionalize $G$ gives
\begin{equation*}
G' = \exists f. (\exists a, b. Apply_{\bf closr} \ C^0_{Twice}\ f \ a \land Apply_{\bf nat} \ a \ 1 \ b \land IOMatch_{\bf nat} \ b \ 3),
\end{equation*}
where the sort of $f$ is now {\bf closr}. Then there exist valuations that are prefixed points of $P'$ but do not satisfy $G'$. The reason is because $Apply_{A}$ and $IOMatch_{A}$ are defined in such a way that
only functions that can be created within $P$ (i.e.~partially applied functions that are represented by $C^i_{F} \ t_1 \ \cdots t_i$, where $F$ is a top-level relational variable) are considered by $P'$. 
As $Add$ is not defined in $P$ anymore, $\lambda x, y. (x + 1 = y)$ cannot arise from $P$. Thus, there exists a prefixed point of $P'$ that does not satisfy $G'$.

Therefore, the semantics of the source safety problem are not preserved if relational variables (except for variables of sort $o$) are quantified in $G$. 

In essence, my defunctionalization fails due to the fact that $P'$ ignores any function that cannot be built from the top-level relational variables defined in $P$. 

\subsection{Elimination of higher-order quantifiers} \label{section on the importance of the monotone semantics}

The monotone semantics can resolve the issue with existential quantification over higher-order variables. 

Monotonicity lets us eliminate higher-order existential quantifiers from all goal terms in a monotone problem. Consider $\exists_{\rho} \lambda x {:} \rho. F$, where $\rho$ is a higher-order relational sort. Any function
interpreted using the monotone semantics is monotone due to the use of $\Rightarrow_{m}$ in the definition of monotone sort frames. Hence, $\mathcal{M} \llbracket \lambda x {:} \rho. F \rrbracket (\alpha)$, where 
${\tt FV}(F) \subseteq {\tt dom}(\alpha)$, is a monotone function of $x$.
If there exists $u \in \mathcal{M} \llbracket \rho \rrbracket$ such that $\mathcal{M} \llbracket \lambda x {:} \rho. F \rrbracket (\alpha)(u) = 1$ for a fixed valuation $\alpha$, any $v \in \mathcal{M} \llbracket \rho \rrbracket$
such that $u \subseteq_{\rho} v$ should satisfy $\mathcal{M} \llbracket \lambda x {:} \rho. F \rrbracket (\alpha)(v) = 1$ as well. The maximum element in $\mathcal{M} \llbracket \rho \rrbracket$ is the relation that always returns 1.
This is called the universal relation of sort $\rho$. It follows that $\mathcal{M} \llbracket \exists_{\rho} \lambda x {:} \rho. F \rrbracket (\alpha) = \mathcal{M} \llbracket F[x \mapsto \lambda x_1, \ldots, x_k. {\tt true}] \rrbracket (\alpha)$
for any $\alpha$ such that ${\tt FV}(F) \subseteq {\tt dom}(\alpha)$. 
Here, I assume that {\tt true} is declared in the signature and is included in the background theory. If this is not the case, we can simply add {\tt true} to the signature and the background theory. 
Henceforth, I assume that a monotone problem does not have quantifiers over higher-order relational variables.


\section{Algorithm} \label{chapter on the formal presentation of the algorithm}

This section first  formally presents the defunctionalization algorithm. 
It then introduces valuation extraction, which is a crucial idea in the proofs of the algorithm's completeness and soundness. 
The section concludes with proofs of completeness and soundness.

\subsection{Preprocessing} \label{preprocessing}

Let the source monotone problem be $\mathcal{P} = (\Delta, P, G)$. Prior to defunctionalization $\mathcal{P}$, we need to eliminate all anonymous functions in $P$ and $G$ and then perform $\eta$-expansion to fully expand the outermost lambda
abstractions defining top-level relational variables. For reasons of space, the details are not presented here. They can be found in Appendix~\ref{preprocessing in the appendix}. 

\subsection{Defunctionalization algorithm} \label{section on the defunctionalization algorithm}
In this subsection, I formulate the defunctionalization algorithm via parametrised relations. 
A goal term to be defunctionalized is called a ``source goal term'', and a defunctionalized goal term is called a ``target goal term''.
An input monotone safety problem is called a ``source monotone problem'' and a defunctionalized monotone safety problem is called a ``target monotone problem''. 

One approach to formulating a defunctionalization algorithm of a functional programming language is to define a relation between source terms and target terms \cite{Pottier2004}. 
I will denote the relation by $\leadsto$ and call it a transformation. $s \leadsto t$ means that $s$ is defunctionalized into $t$. 
The word ``transformation'' might suggest that $\leadsto$ is not only a relation but also a function. 
It is in fact possible to  show that $\leadsto$ returns unique outputs and hence is a function. However, it suffices to regard $\leadsto$ as a relation in this document. 

Prior to presenting the core of the defunctionalization algorithm, I explain why I make use of parametrised transformations in the formulation of the algorithm. 

\subsubsection{Parametrised transformation}

The use of parametrised transformations gives us control over variable symbols in target goal terms. To appreciate the importance of being able to specify variable symbols, consider the goal term
\begin{equation*}
f \ x \ y,
\end{equation*}
where the sort of each variable is 
\begin{align*}
f & : {\bf int} \to {\bf int} \to o \\
x & : {\bf int}\\
y & : {\bf int}.
\end{align*}
Further, assume $f \notin {\tt dom}(\Delta)$. This goal term can be defunctionalized into
\begin{equation}
\exists a. Apply_{\bf int} \ f \ x \ a \land IOMatch_{\bf int} \ a \ y, \label{defunctionalization of f x y}
\end{equation}
where the sort of each variable is now
\begin{align*}
f & : {\bf closr} \\
x & : {\bf int}\\
y & : {\bf int} \\
a & : {\bf closr}. 
\end{align*}
Note that although \eqref{defunctionalization of f x y} is not identical to the output of the defunctionalization algorithm presented in Subsection~\ref{core of the defunctionalization algorithm}, they are logically equivalent.  

As the grammar of goal terms is defined inductively, it is natural to defunctionalize goal terms inductively. \eqref{defunctionalization of f x y} consists of two components:
\begin{align*}
& Apply_{\bf int} \ f \ x \ a \\
& IOMatch_{\bf int} \ a \ y.
\end{align*}
The former corresponds to the partial application of $f$ to $x$ and the latter corresponds to the application of $(f \ x)$ to $y$. 
Hence, the structure of \eqref{defunctionalization of f x y} roughly reflects the structure of $f \ x \ y$, where the curried function $f$ is applied to $x$ first and then to $y$. 
However, the two components in \eqref{defunctionalization of f x y} cannot be separated cleanly. 
The problem is that both components refer to the same quantified variable $a$. 
Hence, it is necessary to establish a ``communication channel'' between the defunctionalization of $f \ x$ and the defunctionalization of $(f \ x) \ y$. 
More specifically, we need to either specify what variable should be used in the first component or inspect it and then copy the variable symbol used in it to the second component. 


As it is certainly not straightforward to define helper functions that extract variable symbols from target goal terms, I opted to pass variable symbols to target goal terms. 
One approach is to use contexts.
For the above example, we can defunctionalize $f \ x$ into a context $Apply_{\bf int} \ f \ x \ X$, where $X$ is a hole. We can then substitute a concrete variable symbol into $X$. 

Another approach to passing variable symbols is to use parametrised transformations. 
If a parametrised transformation $\leadsto^{X}$ is defined in such a way that $(f \ x) \leadsto^{X} (Apply_{\bf int} \ f \ x \ X)$ holds, we can invoke $\leadsto^{X}$ with a specific variable symbol substituted into $X$.

The transformation in the context-based approach sometimes returns contexts and other times returns goal terms with no holes. 
Since this can be confusing to readers, I adopted the approach based on parametrised transformations. 

\subsubsection{Defunctionalization} \label{core of the defunctionalization algorithm}

\paragraph{Formal presentation}

Let $\mathcal{P} = (\Delta, P, G)$ be the source monotone problem that has already been preprocessed by the procedure explained in Subsection~\ref{preprocessing}. Due to the preprocessing, the formal parameters of each outermost
lambda abstraction defining a top-level relational variable are visible. Also, $P$ and $G$ contain no anonymous functions. Given $P = \lambda x_1, \ldots, x_m. F$, where $F$ is not a lambda abstraction, let us call $F$ the body
of $P$. Because $G: o$ and hence is not a function, the body of $G$ is $G$ itself. Since lambda abstractions only appear at the top level of syntax trees of $P$ and $G$, the bodies of $P$ and $G$ are free of lambda abstractions.
The defunctionalization of the bodies is guided by the following inference rules:
\begin{prooftree}
\AxiomC{$c \in \{ \land, \lor \}$}
\AxiomC{$E \leadsto E'$}
\AxiomC{$F \leadsto F'$}
\RightLabel{\sc (LogSym)}
\TrinaryInfC{$(c \ E \ F) \leadsto (c \ E' \ F')$}
\end{prooftree}
\begin{prooftree}
\AxiomC{$F \leadsto F'$}
\RightLabel{\sc (Exi)}
\UnaryInfC{$\exists_{b} x. F \leadsto \exists_{b} x. F'$}
\end{prooftree}
\begin{prooftree}
\AxiomC{$\text{head}(E) \notin \{\land, \lor \}$}
\AxiomC{$\Delta \vdash (E \ F): \rho$}
\AxiomC{$\rho \neq o$}
\AxiomC{$(E \ F) \leadsto_{A}^{X} H$}
\RightLabel{\sc (App)}
\QuaternaryInfC{$(E \ F) \leadsto^{X} H$}
\end{prooftree}
\begin{prooftree}
\AxiomC{$\text{head}(E) \notin \{\land, \lor \}$}
\AxiomC{$\Delta \vdash (E \ F): o$}
\AxiomC{$(E \ F) \leadsto_{M} H$}
\RightLabel{\sc (Match)}
\TrinaryInfC{$(E \ F) \leadsto H$}
\end{prooftree}
\begin{prooftree}
\AxiomC{$E \leadsto^{x} E'$}
\AxiomC{$\Delta \vdash F: \sigma$}
\AxiomC{$\sigma \leadsto_{T} \sigma$}
\AxiomC{$F \leadsto F'$}
\RightLabel{\sc (App-Base)*}
\QuaternaryInfC{$(E\ F) \leadsto_{A}^{X} \exists_{\bf closr} x. (E' \land Apply_{\sigma} \ x  \ F' \ X)$}
\end{prooftree}
\begin{prooftree}
\AxiomC{$E \leadsto^{x} E'$}
\AxiomC{$\Delta \vdash F: \sigma$}
\AxiomC{$\sigma \leadsto_{T} {\bf closr}$}
\AxiomC{$F \leadsto^{y} F'$}
\RightLabel{\sc (App-Arrow)*}
\QuaternaryInfC{$(E\ F) \leadsto_{A}^{X} \exists_{\bf closr} x. (E' \land \exists_{\bf closr} y. (F' \land Apply_{\bf closr} \ x  \ y \ X))$}
\end{prooftree}
\begin{prooftree}
\AxiomC{$E \leadsto^{x} E'$}
\AxiomC{$\Delta \vdash F: \sigma$}
\AxiomC{$\sigma \leadsto_{T} \sigma$}
\AxiomC{$F \leadsto F'$}
\RightLabel{\sc (Match-Base)*}
\QuaternaryInfC{$(E\ F) \leadsto_{M} \exists_{\bf closr} x. (E' \land IOMatch_{\sigma} \ x  \ F')$}
\end{prooftree}
\begin{prooftree}
\AxiomC{$E \leadsto^{x} E'$}
\AxiomC{$\Delta \vdash F: \sigma$}
\AxiomC{$\sigma \leadsto_{T} {\bf closr}$}
\AxiomC{$F \leadsto^{y} F'$}
\RightLabel{\sc (Match-Arrow)*}
\QuaternaryInfC{$(E\ F) \leadsto_{M} \exists_{\bf closr} x. (E' \land \exists_{\bf closr} y. (F' \land IOMatch_{\bf closr} \ x  \ y))$}
\end{prooftree}
\begin{center}
\AxiomC{$\varphi \in Fm \cup Tm$}
\RightLabel{\sc (ConstrLan)}
\UnaryInfC{$\varphi \leadsto \varphi$}
\DisplayProof
\qquad 
\AxiomC{$x$ is a variable}
\AxiomC{$x: o$}
\RightLabel{\sc (Var-Base)}
\BinaryInfC{$x \leadsto x$}
\DisplayProof
\end{center}
\begin{prooftree}
\AxiomC{$x$ is a variable}
\AxiomC{$x: \rho$}
\AxiomC{$\rho \neq o$}
\RightLabel{\sc (Var-Arrow)}
\TrinaryInfC{$x \leadsto^{X} X = x$}
\end{prooftree}
\begin{prooftree}
\AxiomC{$x \in \Delta$}
\RightLabel{\sc (TopVar)}
\UnaryInfC{$x \leadsto^{X} X = C^0_{x}$}
\end{prooftree}
\begin{center}
\AxiomC{${\tt order}(b) = 1$}
\RightLabel{\sc (Base)}
\UnaryInfC{$b \leadsto_{T} b$}
\DisplayProof
\qquad
\AxiomC{${\tt order}(\tau)> 1$}
\RightLabel{\sc (Arrow)}
\UnaryInfC{$\tau \leadsto_{T} {\bf closr}$}
\DisplayProof
\end{center}
The function head is defined by
\begin{alignat*}{2}
\text{head}(x) & = x && \qquad x \text{ is a variable} \\
\text{head}(c) & = c && \qquad c \in \{\land, \lor \} \\
\text{head}(\varphi) & = \varphi && \qquad \varphi \in Fm \cup Tm \\
\text{head}(E) & = E && \qquad E \text{ is in the form } \lambda x. F \text{ or } \exists x. F \\
\text{head}(E \ F) & = \text{head}(E).
\end{alignat*}
This function returns the head symbols of goal terms. 

In the conclusions of the rules whose names are marked with *, quantified variables $x$ and $y$ are assumed to be different from any variable symbol occurring in $E'$ and $F'$ before substitutions $[X \mapsto x]$ and $[X \mapsto y]$ are
applied. 

\paragraph{Transformation types}

The abbreviations of the rules' names and what they stand for are summarised below:
\begin{center}
\begin{tabular}[c]{c c}
Abbreviation & Full form \\
\hline
\textsc{LogSym} & Logical constant symbols \\
\textsc{Exi} & Existential quantifier \\
\textsc{App} & Apply \\
\textsc{Match} & IOMatch \\
\textsc{App-Base} & Apply for a base sort \\
\textsc{App-Arrow} & Apply for an arrow sort \\
\textsc{Match-Base} & IOMatch for a base sort \\
\textsc{Match-Arrow} & IOMatch for an arrow sort \\
\textsc{ConstrLan} & Constraint language \\
\textsc{Var-Base} & Variable of base sort \\
\textsc{Var-Arrow} & Variable of an arrow sort\\
\textsc{TopVar} & Top-level relational variable \\
\textsc{Base} & Base sort \\
\textsc{Arrow} & Arrow sort \\
\end{tabular}
\end{center}

In the above inference rules, we have five types of transformations: $\leadsto$, $\leadsto^{X}$, $\leadsto_{A}^{X}$, $\leadsto_{M}$, and $\leadsto_{T}$. Superscripts of $\leadsto$ store parameters, and subscripts denote the classes
of the transformations. `A' in $\leadsto_{A}^{X}$ is short for Apply, `M' in $\leadsto_{M}$ is short for Match, and `T' in $\leadsto_{T}$ is short for Types (i.e.~sorts). 

The transformations $\leadsto$ and $\leadsto^{X}$ are applied to goal terms that contain no lambda abstractions. They are used to transform the bodies of lambda abstractions defining top-level relational variables. 
In $\leadsto^{X}$, $X$ is a parameter into which a variable symbol is substituted. Hence, $\leadsto^{X}$ is a parametrised relation that returns an appropriate goal term according to the parameter $X$ passed to the relation.

The transformation $\leadsto_{A}^{X}$ is for function application that produces goal terms with arrow sorts. Because the result of the function application is not of base sort, the argument of the function application cannot be the last
parameter of any relational variable (otherwise, the result of the function application would have the sort $o$). Hence, $\leadsto_{A}^{X}$ replaces such function application with an instance of $Apply$. Like $\leadsto^{X}$, 
$\leadsto_{A}^{X}$ has a parameter for variable symbols. Notice that $X$ in the inference rules defining $\leadsto^{X}$ and $\leadsto_{A}^{X}$ acts as a ``metavariable'' and hence is used as a ``pattern'' in pattern matching.  

In contrast to $\leadsto_{A}^{X}$, $\leadsto_{M}$ is for function application that produces goal terms of base sort. Such function application is replaced with an instance
of $IOMatch$. 

Lastly, $\leadsto_{T}$ transforms sorts. 

\subsubsection{Bindings of variables and quantifiers}

In the four inference rules marked with *, new quantified variables are introduced in the result of transformation. To avoid variable capture and disambiguate bindings of quantifiers and quantified variables, the newly introduced quantified variables
must be distinct from all variable symbols occurring in $E'$ and $F'$ before we apply substitutions $[X \mapsto x]$ and $[X \mapsto y]$. Hence, it is necessary to calculate $E'$ and $F'$ before we can select suitable symbols
for the quantified variables in the four rules' conclusions. 

To illustrate the need for using fresh variables, consider the partially applied goal term
\begin{equation}
Add \ x, \label{Add x}
\end{equation}
where the sort of each variable is
\begin{align*}
Add & : {\bf int} \to {\bf int} \to {\bf int} \to o \\
x & : {\bf int}.
\end{align*}
Further, assume $Add \in \Delta$. By {\sc (TopVar)}, we have
\begin{equation}
Add \leadsto^{X} X = C^0_{Add}, \label{defunctionalization of Add}
\end{equation}
where $X$ is to be specified when the enclosing goal term is defunctionalized. We also have
\begin{equation*}
x \leadsto x
\end{equation*}
by {\sc (Var-Base)}. 

To produce a target term of $Add \ x$, we apply {\sc (App)} and {\sc (App-Base)}. If $x$ is used for a new quantified variable, we obtain
\begin{equation}
Add \ x \leadsto^{X} \exists_{\bf closr} x. x = C^0_{Add} \land Apply_{\bf int} \ x \ x \ X. \label{ambiguous defunctionalization of Add x}
\end{equation}
Variable capture happens in \eqref{ambiguous defunctionalization of Add x} since the second argument of $Apply_{\bf int}$, which ought to be a free variable, is now bound by $\exists_{\bf closr}$. 
To disambiguate this expression, a fresh variable symbol $y$ is used for the quantified variable:
\begin{equation}
Add \ x \leadsto^{X} \exists_{\bf closr} y. y = C^0_{Add} \land Apply_{\bf int} \ y \ x \ X. \label{defunctionalization of Add x}
\end{equation}

$X$ in the parametrised transformations specifies what variable symbol is used to denote the entity of sort {\bf closr} that represents the source goal term. 

For example, consider $Add \ x$ in \eqref{Add x}. In the target term of $Add$ in \eqref{defunctionalization of Add}, $X$ denotes the entity of sort {\bf closr} that represents $Add$, i.e.~$C^0_{Add}$, because $X$ and $C^0_{Add}$
are connected by equality. 

Also, in the target term of $Add \ x$ in \eqref{defunctionalization of Add x}, we have $Apply_{\bf int} \ y \ x \ X$, where $y = C^0_{Add}$. The third parameter of $Apply$, which always has the sort {\bf closr}, represents the
result of applying the function represented by the first parameter to the entity represented by the second parameter. Hence, $X$ denotes a closure that represents $Add \ x$. 

\subsubsection{Target monotone problems} \label{definition of target monotone problems}

The defunctionalized monotone problem is given by $\mathcal{P}' = (\Delta', P', G')$ with a new signature $\Sigma' = (\mathbb{B}', \mathbb{S}')$, where $\mathbb{B}'$ and $\mathbb{S}'$ are given by
\begin{align*}
\mathbb{B}' & = \mathbb{B} \cup \{ {\bf closr} \} \\
\mathbb{S}' & = \mathbb{S} \cup \{ (=_{\bf closr}): {\bf closr} \to {\bf closr} \to o \} \\
& \quad \cup \{ C^{i}_{X}: \sigma_1' \to \cdots \to \sigma_i' \to {\bf closr} \mid X: \sigma_1 \to \cdots \to \sigma_m \to o \in \Delta, \\
& \qquad \quad 0 \leq i < m, \sigma_j \leadsto_{T} \sigma_j' \text{ for all } 1 \leq j \leq i \}.
\end{align*}
Here, {\bf closr}, $(=_{\bf closr})$, and $C^{i}_{X}$ are assumed to be fresh. The new sort environment $\Delta'$ is
\begin{equation} \label{new Delta}
\begin{split}
\Delta' & = \{Apply_{A}: {\bf closr} \to A \to {\bf closr} \to o \mid A \in \mathbb{B}' \} \\
& \quad \cup \{IOMatch_{A}: {\bf closr} \to A \to o \mid A \in \mathbb{B}' \}. 
\end{split}
\end{equation}
Some of $Apply_{A}$ and $IOMatch_{A}$ may be redundant. $P'$ is defined as
\begin{equation} \label{new P}
P' = P'_{\text{Apply}} \cup P'_{\text{IOMatch}}, 
\end{equation}
where $P'_{\text{Apply}}$ and $P'_{\text{IOMatch}}$ are
\begin{equation} \label{new P for Apply}
\begin{split}
P'_{\text{Apply}} & = \{ Apply_{\sigma_{n+1}'} = \lambda x, y, z. (\exists a_1, \ldots, a_n. x = C^{n}_{X} \ a_1 \ \cdots \ a_n \land z = C^{n+1}_{X} \ a_1 \ \cdots  \ a_n \ y) \\
& \qquad \mid (X = \lambda x_1 {:} \sigma_1, \ldots, x_m {:} \sigma_m. F) \in P, {\bf ar}(X) = m, 0 \leq n \leq m- 2, \sigma_{n+1} \leadsto_{T} \sigma_{n+1}' \}
\end{split}
\end{equation}
\begin{equation} \label{new P for IOMatch}
\begin{split}
P'_{\text{IOMatch}} & = \{ IOMatch_{\sigma_{m}'} = \lambda x, x_{m}. (\exists x_1, \ldots, x_{m-1}. x = C^{m-1}_{X} \ x_1 \ \cdots \ x_{m-1} \land F') \\
& \qquad \mid (X = \lambda x_1 {:} \sigma_1, \ldots, x_m {:} \sigma_m. F) \in P, {\bf ar}(X) = m, \sigma_{m} \leadsto_{T} \sigma_{m}', F \leadsto F' \}
\end{split}
\end{equation}
Note that $F$ in \eqref{new P for IOMatch} must have the sort $o$. This indeed holds and follows from Lemma~\ref{lemma for well-sortedness of P' IOMatch}).
In the definition of $P'_{\text{IOMatch}}$, every occurrence of $x_1, \ldots, x_m$ in $F'$ is bound by the outermost $\exists $ in the body of $IOMatch_{\sigma'_{m}}$. 

The defunctionalized goal formula $G'$ is given by
\begin{equation*}
G \leadsto G'. 
\end{equation*}

Lastly, the constraint language, particularly the background theory, for $\mathcal{P}'$ need to be defined. 
Let the constraint language for $\mathcal{P}$ be $(Tm, Fm, Th)$ and the constraint language for $\mathcal{P}'$ be $(Tm', Fm', Th')$. 
$Tm'$ and $Fm'$ are informally defined as extensions of $Tm$ and $Fm$ with terms and formulas containing $(=_{\bf closr})$ and $C^i_X$ for some relational variable $X \in \Delta$. 
Because formal definitions of $Tm'$ and $Fm'$ are not critical to the proofs of the defunctionalization algorithm's correctness, I will not formally define them. 

The background theory $Th$ can always be characterized by a set of structures $S$ such that $F \in Th$ if and only $A \models  F$ for all $A \in S$. 
For each $A \in S$, a structure $A'$ for $\mathcal{P}'$ is defined as follows:
\begin{itemize}
\item To all sorts inherited from $\mathbb{B}$, $A'$ assigns the same universe as $A$.
\item To the sort {\bf closr}, $A'$ assigns the universe of objects that can be constructed by the data constructors $C^i_{X} \in \mathbb{S}' \setminus \mathbb{S}$. Informally, the universe assigned to {\bf closr} is 
\begin{align*}
A'_{\bf closr} = \{ (& X, t_1, \ldots, t_k) \mid X: \sigma_1 \to \cdots \to \sigma_m \to o \in \Delta, 0 \leq k < m, \\
& \sigma_i \leadsto_{T} \sigma_i' \text{ for each } 1 \leq i \leq k, t_i \in A'_{\sigma_i'} \text{ for each } 1 \leq i \leq k\}. 
\end{align*}
In other words, the universe is the set of tuples in which the first component denotes a top-level relational variable and the remaining components represent the actual parameters that have been supplied to the relational variable.
This definition is informal because we may have $A'_{\sigma_i'} = A'_{\bf closr}$; i.e.~the definition may be circular. 
In that case, the definition does not qualify as a formal definition. 
Another problem we have with this informal definition is that infinitely nested tuples are admitted. 
To get around this issue, I provide an inference rule to construct elements in $A'_{\bf closr}$:
\begin{prooftree}
\AxiomC{$X: \sigma_1 \to \cdots \to \sigma_m \to o \in \Delta$}
\AxiomC{$0 \leq k < m$}
\AxiomC{$t_i \in A'_{b_i} \text{ for each } 1 \leq i \leq m$}
\TrinaryInfC{$(X, t_1, \ldots, t_k) \in A'_{\bf closr}$}
\end{prooftree}
where $\sigma_i \leadsto_{T} b_i$ for each $1 \leq i \leq k$. 
\item For all constant symbols inherited from $\mathbb{S}$, $A'$ interprets them in the same way as $A$. 
\item The interpretation of $(=_{\bf closr}): A_{\bf closr}' \to A_{\bf closr}' \to \mathbbm{2}$ is determined by this inference rule
\begin{prooftree}
\AxiomC{$X: \sigma_1 \to \cdots \to \sigma_k \to o \in \Delta$}
\AxiomC{$0 \leq k < m$}
\AxiomC{$t_i =_{b_i} s_i \text{ for each } 1 \leq i \leq k$}
\TrinaryInfC{$(X, t_1, \cdots, t_k) =_{\bf closr} (X, s_1, \cdots, s_k)$}
\end{prooftree}
where $\sigma_i \leadsto_{T} b_i$ for each $1 \leq i \leq k$. The interpretation of  $(=_{\bf closr})$ is well-defined since the equality $(=_{b})$ for each $b \in \mathbb{B}$ exists. 
\item $C^i_{X}$ is interpreted as a function that takes in $i$ many arguments and returns an appropriate object from $A'_{\bf closr}$. Formally, it is defined by
\begin{equation*}
C^i_{X}(t_1, \ldots, t_i) = (X, t_1, \ldots, t_i). 
\end{equation*}
\end{itemize}
The new background theory for the target monotone problem is obtained by extending each model in $S$:
\begin{equation*}
S' = \{A' \mid A \in S \},
\end{equation*}
where $S'$ is a set of models characterizing $Th'$.


\subsection{Valuation extraction}

Given a monotone problem $\mathcal{P} = (\Delta, P,G)$, a model of $P$ is an element of $\mathcal{M} \llbracket \Delta \rrbracket$ such that it is a prefixed point of the one-step consequence operator $T^{\mathcal{M}}_{P: \Delta}$.
If $\alpha$ is a model of $P$ and $X \in {\tt dom}(\Delta)$, I will call $\alpha(X)$ a model of $X$. $\mathcal{P}$ is said to be solvable if for every model of the background theory, there exists a model $\alpha$ of $P$ such that 
$\mathcal{M} \llbracket G \rrbracket (\alpha) = 0$. I will call such a model of $P$ a solution to $\mathcal{P}$.

To prove completeness of the algorithm, my approach is to extract a solution to the target monotone problem from a solution to the source monotone problem. 
Hence, I will start with explaining how valuations can be extracted. 

In this subsection, a source monotone problem is assumed to have been preprocessed. As it is relatively easy to see that the preprocessing step preserves semantics, a formal proof for that will not be provided.

\subsubsection{Demonstration} \label{demonstration of valuation extraction}

To illustrate how extraction works, consider a source monotone problem
$\mathcal{P} = (\Delta, P, G)$ with the first-order signature being $\Sigma = (\mathbb{B}, \mathbb{S})$, where
\begin{align*}
\mathbb{B} & = \{ {\bf nat}, o \} \\
\mathbb{S} & = \{+: {\bf nat} \to {\bf nat} \to {\bf nat}, (=): {\bf nat} \to {\bf nat} \to o \} \cup \{ n: {\bf nat} \mid n \in \mathbb{N} \}. 
\end{align*}
$P$ contains
\begin{align*}
Add & = \lambda a, b, c. a + b = c \\
Twice & = \lambda f, a, b. (\exists c. f \ a \ c \land f \ c \ b).
\end{align*}
The sort environment for these two top-level relational variables is
\begin{equation*}
\Delta = \{Add: {\bf nat} \to {\bf nat} \to {\bf nat} \to o, Twice: ({\bf nat} \to {\bf nat} \to o) \to {\bf nat} \to {\bf nat} \to o \}. 
\end{equation*}
Since $G$ is irrelevant to the discussion of how to extract solutions, it is unnecessary to specify $G$. Let $A$ be the structure that assigns the universe of natural numbers, denoted by $\mathbb{N}$, to the sort {\bf nat}. $A$ interprets the symbols
in $\mathbb{S}$ as they are. 

I will consider a specific model $\alpha$ of $P$ that is defined by
\begin{equation*}
\alpha: Add \mapsto add \qquad \alpha: Twice \mapsto twice. 
\end{equation*}
The model of $Add$ is $add: \mathbb{N} \to \mathbb{N} \to \mathbb{N} \to \mathbbm{2}$ defined as
\begin{equation*}
add \ a \ b \ c =
\begin{cases}
1 & \text{if } a + b = c \\
0 & \text{otherwise}.
\end{cases}
\end{equation*}
The model of $Twice$ is $twice: (\mathbb{N} \to \mathbb{N} \to \mathbbm{2}) \to \mathbb{N} \to \mathbb{N} \to \mathbbm{2}$ is defined as
\begin{equation*}
twice \ f \ a \ b =
\begin{cases}
1 & \text{if } \exists c. f \ a \ c = 1 \land f \ c \ b = 1 \\
0 & \text{otherwise}.
\end{cases}
\end{equation*}
In fact, regardless of valuation $\alpha$, $add$ coincides with $\mathcal{M} \llbracket \Delta \vdash (\lambda a, b, c. a + b = c) : o \rrbracket (\alpha)$, and similarly $twice$ coincides with 
$\mathcal{M} \llbracket \Delta \vdash (\lambda f, a, b. (\exists c. f \ a \ c \land f \ c \ b)) : o \rrbracket (\alpha)$. In other words, any model of $P$ is larger than or equal to $\alpha$; hence, $\alpha$ is the least model of $P$.

$\mathcal{P}'$ is defunctionalized into $\mathcal{P}' = (\Delta', P', G')$, where $P'$ contains
\begin{align*}
Apply_{\bf nat} & = \lambda x, y, z. x = C^0_{Add} \land z = C^1_{Add} \ y \\
Apply_{\bf nat} & = \lambda x, y, z. (\exists n. x = C^1_{Add} \ n \land z = C^2_{Add} \ n \ y) \\
IOMatch_{\bf nat} & = \lambda x, c. (\exists n, a. x = C^2_{Add} \ a \ b \land a + b = c) \\[10pt]
Apply_{\bf closr} & = \lambda x, y, z. x = C^0_{Twice} \land z = C^1_{Twice} \ y \\
Apply_{\bf nat} & = \lambda x, y, z. (\exists f. x = C^1_{Twice} \ f \land z = C^2_{Twice} \ f \ y) \\
IOMatch_{\bf nat} & = \lambda x, b. (\exists f, a. x = C^2_{Twice} \ f \ a \land (\exists c. (\exists d. Apply_{\bf nat} \ f \ a \ d \land IOMatch_{\bf nat} \ d \ c) \\
& \quad \land (\exists e. Apply_{\bf nat} \ f \ c \ e \land IOMatch_{\bf nat} \ e \ b))).
\end{align*}

I will now work out a model of $P'$ induced by $\alpha$. The universe of {\bf nat} in $P'$ remains $\mathbb{N}$. The universe of {\bf closr} in $P'$, denoted by $A_{\bf closr}'$, is constructed by
\begin{prooftree}
\AxiomC{$X: \sigma_1 \to \cdots \to \sigma_m \to o \in \Delta$}
\AxiomC{$0 \leq k < m$}
\AxiomC{$t_i \in A'_{b_i} \text{ for each } 1 \leq i \leq m$}
\TrinaryInfC{$(X, t_1, \ldots, t_k) \in A'_{\bf closr}$}
\end{prooftree}
where $\sigma_i \leadsto_{T} b_i$ for each $1 \leq i \leq k$. 
A model for $Apply_{\bf nat}$ is the function $apply_{\bf nat}: A_{\bf closr}' \to \mathbb{N} \to A_{\bf closr}' \to \mathbbm{2}$ defined as
\begin{equation*}
apply_{\bf nat} \ m_1 \ n \ m_2 =
\begin{cases}
1 & \text{if } m_1 = C^0_{X} \land m_2 = C^1_{X} \ n \\
& \text{or } \exists n_1. (m_1 = C^1_{X} \ n_1 \land m_2 = C^2_{X} \ n_1 \ n) \\
0 & \text{otherwise},
\end{cases}
\end{equation*}
where $X \in \{Add, Twice\}$. 

More generally, the model for $Apply_{B}$ is a function $apply_{B}: A_{\bf closr}' \to A'_{B} \to A_{\bf closr}' \to \mathbbm{2}$ that takes three inputs: $m_1, n, m_2$. Here, the universes $A_{b}'$, where $b \in \mathbb{B} \cup \{ {\bf closr}\}$,
are defined in Subsection~\ref{definition of target monotone problems}. The parameter $m_1$ is a closure that represents a partially applied function, and
$n$ is an input to be augmented to $m_1$. Thus, $apply_{B}$ is defined as
\begin{equation*}
apply_{B} \ m_1 \ n \ m_2 =
\begin{cases}
1 & \text{if } m_2 =_{\bf closr} \text{append}(m_1, n) \\
0 & \text{otherwise},
\end{cases}
\end{equation*}
where $\text{append}: A_{\bf closr}' \to A_{B}' \to A_{\bf closr}'$ is 
\begin{equation*}
\text{append}((X, t_1, \ldots, t_k), t_{k+1}) = (X, t_1, \ldots, t_k, t_{k+1}). 
\end{equation*}
A model for $Apply_{\bf closr}$ in $P'$ is therefore given by $apply_{\bf closr}$. 

Next, I consider $IOMatch_{\bf nat}$. Because it has two branches corresponding to different top-level relational variables from the source problem, I will derive a model for each branch separately. These two models will be merged later to form
a single model for $IOMatch_{\bf nat}$. 

The first branch of $IOMatch_{\bf nat}$ is obtained by defunctionalizing $Add$. This branch has the role of determining whether the first input, which should be a closure of $Add$, can be evaluated to the second input. The first input is expected to have
the form $(Add, n_1, n_2)$. Hence, the model for the first branch of $IOMatch_{\bf nat}$ is the function $add': A_{\bf closr}' \to \mathbb{N} \to \mathbbm{2}$ defined as
\begin{equation*}
add' \ m \ n = 
\begin{cases}
1 & \text{if } m = (Add, n_1, n_2), n = n_1 + n_2 \\
0 & \text{otherwise}.
\end{cases}
\end{equation*}
Notice that this function is similar to $add$ in that both of them perform addition and return 1 whenever the last input matches the result of addition. They only differ in the representation of the inputs:
in $add$, two numbers to be summed are stored in the first two parameters, whereas in $add'$, they are stored inside the closure in the first parameter. Capturing the similarity between $add$ and $add'$, we can easily formalize how to convert
a model for a top-level relational variable to a model for the corresponding $IOMatch$ branch when {\bf closr} is not involved. 

Next, I work out how to interpret the second branch of $IOMatch_{\bf nat}$, which is obtained by defunctionalizing $Twice$. Elements of base sort in $\mathcal{P}'$ can be converted to corresponding elements in $\mathcal{P}$ as follows:
\begin{alignat*}{2}
\text{expand}_{\alpha} (t) & = t && \qquad \text{if $t$ is not of sort {\bf closr}} \\
\text{expand}_{\alpha} ((X, t_1, \ldots, t_k)) & = \alpha (X) \ \text{expand}(t_1) \ \cdots \ \text{expand}(t_k) && \qquad \text{otherwise}.
\end{alignat*}
Using the $\text{expand}_{\alpha}$ function, I define $twice': A'_{\bf closr} \to \mathbb{N} \to \mathbbm{2}$ as 
\begin{equation*}
twice' \ m \ n = 
\begin{cases}
1 & \text{if } m = (Twice, f, n_1) \land twice \ \text{expand}_{\alpha}(f) \ n_1 \ n = 1 \\
0 & \text{otherwise}.
\end{cases}
\end{equation*}
Alternatively, $twice \ \text{expand}_{\alpha}(f) \ n_1 \ n$ can be written as $\text{expand}(m) \ \text{expand}(n)$. The interpretation of the second branch of $IOMatch_{\bf nat}$ is $twice'$. 

By construction, $twice' \ (Twice, f, n_1) \ n_2 = 1$ implies $twice \ \text{expand}(f) \ n_1 \ n_2 = 1$. However, the converse does not hold. For example, $twice' \ (Twice, (Add, 1), 0) \ 2$ holds, as does $twice \ (add \ 1) \ 0 \ 2$. 
By contrast, $twice \ (\lambda a, b. a - 1 = b) \ 2 \ 0$ holds, whilst $twice' \ t \ 2 \ 0$ does not hold for any $t$ of sort {\bf closr}. This is owing to the fact that $A_{\bf closr}'$ only contains closures representing partially applied functions
that are expressible using the top-level relational variables in $P$.

As the two branches of $IOMatch_{\bf nat}$ are (syntactically) combined by taking their disjunction, the interpretation of $IOMatch_{\bf nat}$ is obtained by taking the disjunction of its constituent interpretations. Hence, the resulting model for 
$IOMatch_{\bf nat}$ is $iomatch_{\bf nat}: A_{\bf closr}' \to \mathbb{N} \to \mathbbm{2}$ defined as
\begin{equation*}
iomatch_{\bf nat} \ m \ n = add' \ m \ n \lor twice' \ m \ n. 
\end{equation*}
This can be made more general:
\begin{equation*}
iomatch_{\bf nat} \ m \ n = \text{expand}_{\alpha}(m) \ \text{expand}_{\alpha}(n). 
\end{equation*}

\subsubsection{Formalization of valuation extraction} \label{formalization of valuation extraction}
Given a first-order signature $\Sigma = (\mathbb{B}, \mathbb{S})$, suppose that a source problem is $\mathcal{P} = (\Delta, P, G)$. Assume that the target problem of $\mathcal{P}$ is $\mathcal{P}' = (\Delta', P', G')$ and that
the new signature is $\Sigma'= (\mathbb{B}', \mathbb{S}')$, where $\mathbb{B}' = \mathbb{B} \cup \{ {\bf closr} \}$. The derivation of $\mathbb{S}'$ is presented in Subsection~\ref{definition of target monotone problems}. 

Let $A$ be a $\Sigma$-structure used to interpret $\mathcal{P}$ and $A'$ be the structure for $\mathcal{P}'$ obtained from $A$ as explained in Subsection~\ref{definition of target monotone problems}. 
I write $A_{B}'$ for the universe assigned to $B \in \mathbb{B}'$ by $A'$. 
Also, assume that $\alpha$ is a valuation drawn from $\mathcal{M} \llbracket \Delta \rrbracket$. I will now explain how to derive a valuation $\alpha' \in \mathcal{M} \llbracket \Delta' \rrbracket$ from $\alpha$.

Each top-level relational variable in $\Delta'$ is either $Apply_{B}$ or $IOMatch_{B}$, where $B \in \mathbb{B}'$. 

As for $Apply_{B}$, $\alpha'$ maps it to $apply_{B}: A_{\bf closr}' \to A_{B}' \to A_{\bf closr}' \to \mathbbm{2}$ defined as
\begin{equation*}
apply_{B} \ m_1 \ n \ m_2=
\begin{cases}
1 & \text{if } m_2 =_{\bf closr} \text{append}(m_1, n) \\
0 & \text{otherwise},
\end{cases}
\end{equation*}
where $\text{append}: A_{\bf closr}' \to A_{B}' \to A_{\bf closr}'$ is 
\begin{equation*}
\text{append}((X, t_1, \ldots, t_k), t_{k+1}) = (X, t_1, \ldots, t_k, t_{k+1}). 
\end{equation*}

With respect to $IOMatch_{B}$, its interpretation is given by $iomatch_{B}: A_{\bf closr}' \to A_{B}' \to \mathbbm{2}$ defined as
\begin{equation*}
iomatch_{B} \ m \ n =\text{expand}_{\alpha} (m) \ \text{expand}_{\alpha} (n)
\end{equation*}
The function $\text{expand}_{\alpha}$ is defined as
\begin{alignat*}{2}
\text{expand}_{\alpha}(s) & = s && \qquad \text{if } s: b, b \in \mathbb{B} \\
\text{expand}_{\alpha}((Y, s_1, \ldots, s_l)) & = \alpha(Y) \ \text{expand}_{\alpha}(s_1) \ \cdots \ \text{expand}_{\alpha}(s_l) && \qquad \text{otherwise}.
\end{alignat*}
If $iomatch_{B} \ m \ n = \text{expand}_{\alpha} (m) \ \text{expand}_{\alpha} (n)$ is not well-defined due to type mismatch, then it is set to 0. 

The valuation $\alpha'$ is therefore
\begin{align*}
\alpha' & = \{(Apply_{B}, apply_{B}) \mid B \in \mathbb{B}' \} \} \\
& \quad \cup \{(IOMatch_{B}, iomatch_{B}) \mid B \in \mathbb{B}' \} \}. 
\end{align*}

Henceforth, I will write $\alpha' = T_f(\alpha)$ to mean that $\alpha'$ is derived from $\alpha$ by the above procedure, where $\alpha$ is a valuation for $P$. 

\subsubsection{Monotonicity of $\alpha'$}

We need to check whether the model of each $X \in {\tt dom}(\Delta')$ assigned by $\alpha'$ is monotone. 
In fact, $\alpha'$ is ``nearly'' monotone but is not truly monotone, since $\alpha' (Apply_{o})$ is not monotone. 
Appendix~\ref{chapter on monotonicity of alpha'} describes how to get around this issue. 

\subsection{Meaning preservation} \label{section on meaning preservation}

Meaning preservation means the preservation of source problems' semantics. 
Hence, meaning preservation is achieved when target monotone problems are solvable if and only if source monotone problems are solvable. 

\subsubsection{First direction}
In this subsection, I prove that it is possible to produce a solution to the target monotone problem from a solution to the source monotone problem. 

As usual, given a first-order signature $\Sigma = (\mathbb{B}, \mathbb{S})$, suppose that a source problem is $\mathcal{P} = (\Delta, P, G)$. Assume that $\mathcal{P}$ is defunctionalized into $\mathcal{P}' = (\Delta', P', G')$
and that its new signature is $\Sigma' = (\mathbb{B}', \mathbb{S}')$. 

Let $Th$ be a background theory for $\mathcal{P}$ and $Th'$ be the background theory for $\mathcal{P}'$ derived from $Th$. Assume $A \in Th$ and $A' \in Th'$, where $A'$ is built from $A$ as presented in
Subsection~\ref{definition of target monotone problems}. 

First, I establish the relationship between the semantics of source goal terms and semantics of target goal terms. A source goal term has either an arrow sort or a base sort. I will first illustrate the connection between the semantics of source and target goal terms
in the case when the source goal terms are of arrow sorts.

If $s\leadsto^{X} t$ and $s$ has a relational arrow sort, $X$ can be thought of as a variable (more precisely, a placeholder/hole for a variable) of sort {\bf closr} that represents $s$. For instance, the partially applied function
$Add \ 1$ is defunctionalized into
\begin{equation}
\exists_{\bf closr} x. (x = C^0_{Add} \land Apply_{\bf nat} \ x \ 1 \ X), \label{defunctionalizing Add 1}
\end{equation}
where $X$ is to be specified by a defunctionalization step at a higher level. In \eqref{defunctionalizing Add 1}, $X$ appears in the last parameter of $Apply_{\bf nat}$. Hence, $X$ can be considered as a variable that represents the result of applying
$Add$, which is represented by $x = C^0_{Add}$, to 1. To put it differently, 
\begin{equation*}
\mathcal{M} \llbracket \exists_{\bf closr} x. (x = C^0_{Add} \land Apply_{\bf nat} \ x \ 1 \ X) \rrbracket ([X \mapsto (Add, 1)]) = 1;
\end{equation*}
that is, \eqref{defunctionalizing Add 1} holds when we substitute $X = C^1_{Add} \ 1$. Moreover, it is worth observing that
\begin{equation*}
\text{expand}_{\alpha}((Add, 1)) = \mathcal{M} \llbracket Add \ 1 \rrbracket (\alpha),
\end{equation*}
where $\alpha$ is a valuation of the source goal term $Add \ 1$.

On the other hand, if $\Gamma \vdash s: b$, where $b \in \mathbb{B}$, and $s \leadsto t$, then $s$ and $t$ have the same semantics. For example, consider $s = Add \ 1 \ 1 \ 2$ that is to be interpreted using the valuation $\alpha = [Add \mapsto add]$.
$s$ is defunctionalized into 
\begin{equation*}
\exists_{\bf closr} x, y, z. \left( x = C^0_{Add} \land Apply_{\bf nat} \ x \ 1 \ y \land Apply_{\bf nat} \ y \ 1 \ z \land IOMatch_{\bf nat} \ z \ 2 \right).
\end{equation*}
Let this target goal term be denoted by $t$. $t$ is interpreted using $\alpha'$, which is obtained by applying the procedure presented in Subsection~\ref{formalization of valuation extraction}. It gives 
$\alpha' = [Apply_{\bf nat} \mapsto apply_{\bf nat}, IOMatch_{\bf nat} \mapsto add']$. 

Now we have 
\begin{equation*}
\mathcal{M} \llbracket s \rrbracket (\alpha) = \mathcal{M} \llbracket t \rrbracket (\alpha')
\end{equation*}
since both sides of the equation evaluate to 1. 

To express this formally, consider a well-sorted source goal term $s$ over $\Sigma$ that contains no lambda abstractions. Suppose the following:
\begin{itemize}
\item $s$ is a goal term (or a subgoal term) from $P$. The structure $A$ is used to interpret $s$. 
\item $\Gamma \vdash s: \sigma$, where $\sigma$ is either a relational arrow sort or a base sort. Because $s$ is well-sorted, ${\tt FV}(s) \subseteq {\tt dom}(\Gamma)$. 
\item $s \leadsto^{X} t$. The structure $A'$ is used to interpret $t$. In $s$, both ordinary variables and top-level relational variables are treated as variables. 
However, in $t$, ordinary variables from $s$ have the same status, whilst top-level relational variables from $s$ become constant symbols in $t$. 
\item $\Gamma' \vdash t: \sigma'$, where $\sigma \leadsto_{T} \sigma'$ and ${\tt FV}(t) \subseteq \Gamma'$. $\Gamma'$ can be equal to 
\begin{equation*}
\begin{split}
& \{u: \sigma' \mid u \in {\tt FV}(s) \setminus {\tt dom}(\Delta), u: \sigma \in \Gamma \} \\
& \cup \{Apply_{B}: {\bf closr} \to B \to {\bf closr} \to o \mid B \in \mathbb{B}' \} \\
& \cup \{IOMatch_{B}: {\bf closr} \to B \to o \mid B \in \mathbb{B}' \},
\end{split}
\end{equation*}
although this contains top-level relational variables from $\Delta$, which never appear in $t$. 
\item $\alpha$ is a valuation of $s$ such that if $v \in {\tt FV}(s): \rho$, where $\rho$ is a relational arrow sort, there exists $c \in A_{\bf closr}'$ such that $\text{expand}_{\alpha}(c) = \alpha (v)$. 
\item $\alpha'$ is a valuation of $t$ satisfying
\begin{itemize}
\item $\alpha'(v) = c$ for $v \in {\tt FV}(s) \setminus {\tt dom}(\Delta)$ such that $\text{expand}_{\alpha}(c) = \alpha(v)$.
\item $\alpha'(Apply_{B}) = apply_{B}$ for $B \in \mathbb{B}'$.
\item $\alpha'(IOMatch_{B}) = iomatch_{B}$ for $B \in \mathbb{B}'$.
\end{itemize}
Note that $X \notin {\tt dom}(\alpha')$. Here, $apply_{B}$ and $iomatch_{B}$ are defined in Subsection~\ref{formalization of valuation extraction}. 
\end{itemize}

Notice that $apply_{B}$ in $\alpha'$ is not monotone if $B = o$ as explained in Appendix~\ref{chapter on monotonicity of alpha'}. 
This is only a minor issue since I do not rely on monotonicity of $\alpha'$ to prove the first direction of meaning preservation.
Appendix~\ref{chapter on monotonicity of alpha'} explains how to resolve this issue.

Moreover, given a term $u$, if no existential quantifiers in $u$ bind higher-order variables and all symbols in $u$ have order at most 2, $\mathcal{M} \llbracket u \rrbracket (\alpha') = \mathcal{S} \llbracket u \rrbracket (\alpha')$.
The monotone and standard semantics differ when we have existential quantifiers over higher-order variables. 
Hence, if we use $\alpha'$ to interpret first-order goal terms, the monotone and standard semantics give the same interpretation. 
This can be formally proved by induction on the grammar of goal terms. 
Thus, within this subsection, I write $\mathcal{M} \llbracket u \rrbracket (\alpha')$ even though $\alpha'$ is not truly monotone.

In addition, $T^{\mathcal{M}}_{P': \Delta'}$ is equivalent to $T^{\mathcal{S}}_{P': \Delta'}$, although $T^{\mathcal{M}}_{P': \Delta'}$ is not guaranteed to be monotone if an input is not drawn from $\mathcal{M} \llbracket \Delta' \rrbracket$. 
The monotonicity of $T^{\mathcal{M}}_{P': \Delta'}$ is not used in this subsection. 

The next lemma establishes a semantic relationship between a source goal term and a target goal term. 

\begin{restatable}{lemma}{relationshipbetweenthesemanticsofsourceandtargetgoalterms} \label{relationship between the semantics of source and target goal terms}
If $\sigma \notin \mathbb{B}$, we have a unique $c \in A_{\bf closr}'$ such that $\mathcal{M} \llbracket \Gamma' \vdash t: {\bf closr} \rrbracket (\alpha' \cup [X \mapsto c]) = 1$. In addition, this $c$ satisfies
$\text{expand}_{\alpha}(c) = \mathcal{M} \llbracket \Gamma \vdash s: \sigma \rrbracket (\alpha)$. Otherwise, if $\sigma$ is a base sort, we have
$\mathcal{M} \llbracket \Gamma \vdash s: b \rrbracket (\alpha) = \mathcal{M} \llbracket \Gamma' \vdash t: b \rrbracket (\alpha')$. 
\end{restatable}

The uniqueness of $c \in A_{\bf closr}'$ in the above lemma is important for the inductive proof to work. 
It is worth noting that this $c$ is a closure object that represents the partially applied function $s$.

This lemma allows us to establish the first direction of meaning preservation. 

\begin{restatable}{theorem}{completenessofthedefunctionalizationalgorithm}
If $\mathcal{P}$ is solvable, so is $\mathcal{P}'$. 
\end{restatable}

\subsubsection{Continuity of one-step consequence operators}

There are difficulties with applying the idea of valuation extracting to prove the second direction of meaning preservation (see Appendix~\ref{section on difficulties with valuation extracting in the second direction}).
Hence, for the second direction, I adopt a different approach that does not involve valuation extraction. 
My approach was originally inspired by the work on Communicating Sequential Processes by Roscoe \cite{Roscoe1997}, 
though it later turned out that in the literature on logic programming, the same approach has been used for a long time \cite{Lloyd1987,Hogger1990}. 

In this subsection, I introduce the notion of ``continuity''\index{continuity}\index{Scott continuity|see {continuity}}, also known as Scott continuity.

Given a partially ordered set (poset) $P$ and a subset $X \subseteq P$, the greatest lower bound\index{greatest lower bound} of $X$ is denoted by $\bigsqcap X$ and the least upper bound\index{least upper bound} of $X$ is denoted by $\bigsqcup X$. 

It is explained in \cite{Ramsay2017} that $\mathcal{M} \llbracket \Delta \rrbracket$ is a complete lattice. 

I will first introduce several key definitions taken from \cite{Roscoe1997}. 

\begin{definition}
Given a poset $P$, a subset $D \subseteq P$ is said to be directed if each finite subset $F$ of $D$ has an upper bound in $D$; in other words, there is $y \in D$ such that $x \leq y$ for all $x \in F$. 
\end{definition}

\begin{definition}
A complete partial order\index{complete partial order}\index{cpo|see {complete partial order}} (often abbreviated cpo) is a partial order in which every directed set has a least upper bound, and which has a least element (denoted by $\bot$). 
\end{definition}

A complete lattice is also a complete partial order. 

\begin{definition}
If $P$ and $Q$ are two complete partial orders and $f: P \to Q$, then $f$ is said to be continuous if, whenever $R \subseteq P$ is directed, $\bigsqcup \{ f(x) \mid x \in R \}$ exists and equals $f (\bigsqcup R)$. 
\end{definition}

A continuous function can be shown to be monotone, although I will not do it here. The next proposition establishes that $T^{\mathcal{M}}_{P: \Delta}$ is continuous when the underlying complete lattice is finite. 

\begin{restatable}{proposition}{finitedomainsimplycontinuity}
If $\mathcal{M} \llbracket \Delta \rrbracket$ is finite, then $T^{\mathcal{M}}_{P: \Delta}: \mathcal{M} \llbracket \Delta \rrbracket \to \mathcal{M} \llbracket \Delta \rrbracket$ is continuous. 
\end{restatable}

Continuity of $T^{\mathcal{M}}_{P: \Delta}$ is a strictly weaker condition than the finiteness of $\mathcal{M} \llbracket \Delta \rrbracket$. For instance, even when some universes $A_{b}$ are infinite, if $T^{\mathcal{M}}_{P: \Delta}$ is
an identity function, it is continuous. 

The next theorem shows that if $T^{\mathcal{M}}_{P: \Delta}$ is continuous, then there exists a constructive way to obtain a fixed point. 

\begin{restatable}{theorem}{iterativeapplicationofacontinuousfunctiononthebottomyieldstheleastfixedpoint} \label{iterative application of a continuous function on the bottom yields the least fixed point}
If $f$ is continuous, then $\bigsqcup \{ f^{n} (\bot) \mid n \in \mathbb{N} \}$ is the least fixed point of $f$. 
\end{restatable}

\subsubsection{Second direction} \label{section on the direction direction in the proof for meaning preservation}

The next lemma shows that this diagram commutes:
\begin{displaymath}
\xymatrix{
\zeta \ar[r]^-{T_{f}} & \zeta' \\
\gamma \ar[u]^{T^{\mathcal{M}}_{P: \Delta}} \ar[r]_-{T_{f}} & \gamma' \ar[u]_{T^{\mathcal{M}}_{P': \Delta'}}
}
\end{displaymath}

\begin{restatable}{lemma}{relationshipbetweensourceandtargetonestepconsequenceoperators} \label{relationship between source and target one-step consequence operators}
Given a valuation $\gamma$ of $P$ and a valuation $\gamma'$ of $P'$, suppose $\gamma' = T_{f} (\gamma)$ holds. If $\zeta = T^{\mathcal{M}}_{P: \Delta} (\gamma)$ and $\zeta' = T^{\mathcal{M}}_{P': \Delta'} (\gamma')$, then
$\zeta' = T_{f} (\zeta)$. 
\end{restatable}

The next lemma states that $T_{f}$ holds between the lowest upper bounds of two increasing sequences whose valuations are related by $T_{f}$. 

\begin{restatable}{lemma}{betaandbetaarerelatedbyTf} \label{beta and beta' are related by T-f}
Assume $\beta = \bigsqcup \{ f^{n}_{1} (\alpha) \mid n \in \mathbb{N} \}$, where $f_1 = T^{\mathcal{M}}_{P: \Delta}$, and $\beta' = \bigsqcup \{ f^{n}_{2} (\alpha') \mid n \in \mathbb{N} \}$, where $f_2 = T^{\mathcal{M}}_{P': \Delta'}$.
If $\alpha' = T_{f} (\alpha)$, then $\beta' = T_{f} (\beta)$. 
\end{restatable}

The next theorem establishes soundness of the defunctionalization algorithm, albeit under the extra assumption that one-step consequence operators for the source and target problems are continuous. 

\begin{restatable}{theorem}{seconddirectionofthemeaningpreservationunderthemonotonesemantics} \label{second direction of the meaning preservation under the monotone semantics}
Given that $T^{\mathcal{M}}_{P: \Delta}$ and $T^{\mathcal{M}}_{P': \Delta'}$ are continuous, if $\mathcal{P}'$ is solvable, then so is $\mathcal{P}$. 
\end{restatable}

\subsubsection{Continuous semantics}

Recent work by Jochems \cite{Jochems18} studies the continuous semantics\index{continuous semantics}, which uses continuous function spaces to interpret goal terms. 
In his working paper, it is shown that one-step consequence operators in the continuous semantics are continuous:

\begin{theorem} \label{one-step consequence operators in the continuous semantics are continuous}
$T^{\mathcal{C}}_{P: \Delta}$ is continuous for all programs $P$ in the continuous semantics. 
\end{theorem}

Further, Jochems \cite{Jochems18} proves the equivalence between the monotone and continuous semantics:

\begin{theorem} \label{equivalence between the monotone and continuous semantics}
The HoCHC safety problem $(\Delta, P, G)$ is solvable under the monotone interpretation, if and only if it is solvable under the continuous interpretation. 
\end{theorem}

In Subsection~\ref{section on the direction direction in the proof for meaning preservation}, the key result is Lemma~\ref{relationship between source and target one-step consequence operators}, which in turn hinges on
Lemma~\ref{relationship between the semantics of source and target goal terms}. 
Valuation extraction works correctly even if we start with a source ``continuous'' problem (as opposed to a source monotone problem). 
Because the defunctionalized target problem is first order, it has the same meaning regardless of which of the standard, monotone, and continuous semantics we use to interpret the problem.
Therefore, Lemma~\ref{relationship between the semantics of source and target goal terms} can be adapted to the continuous semantics. 
Also, Lemma~\ref{relationship between source and target one-step consequence operators} can be adapted to the continuous semantics. 
As a consequence, adapting Theorem~\ref{second direction of the meaning preservation under the monotone semantics} to the continuous semantics yields

\begin{theorem} \label{second direction of the meaning preservation under the continuous semantics}
If $\mathcal{P}'$ is solvable under the continuous semantics, then $\mathcal{P}$ is also solvable under the continuous semantics.
\end{theorem}

This is because one-step consequence operators in the continuous semantics are continuous by Theorem~\ref{one-step consequence operators in the continuous semantics are continuous}. 

Finally, this gives
\begin{alignat*}{2}
\mathcal{P}' \text{ is solvable under } \mathcal{M} & \implies \mathcal{P}' \text{ is solvable under } \mathcal{C}  && \qquad \text{by Theorem~\ref{equivalence between the monotone and continuous semantics}} \\
& \implies \mathcal{P} \text{ is solvable under } \mathcal{C} && \qquad \text{by Theorem~\ref{second direction of the meaning preservation under the continuous semantics}} \\
& \implies \mathcal{P} \text{ is solvable under } \mathcal{M} && \qquad \text{by Theorem~\ref{equivalence between the monotone and continuous semantics}},
\end{alignat*}
where $\mathcal{M}$ and $\mathcal{C}$ denote the monotone and continuous semantics, respectively. Therefore, the defunctionalization algorithm is sound. 


\section{Implementation and evaluation} \label{chapter on implementation and evaluation}

This document describes how the defunctionalization algorithm for monotone problems is implemented. 
The source code, including a test suite, is available at \url{https://github.com/LongPham7/Defunctionalization-of-monotone-problems}. 
A web interface is available at \url{http://mjolnir.cs.ox.ac.uk/dfhochc/}. 
This web interface feeds the defunctionalization algorithm's output into Z3, an SMT solver developed by Microsoft Research, to verify the defunctionalized target problems.

\subsection{Implementation}

\subsubsection{Input format}

By way of example, consider a monotone safety problem $\mathcal{P} = (\Delta, P, G)$, where
\begin{align*}
\Delta & = \{add: {\tt int} \to {\tt int} \to {\tt int} \to {\tt bool}, \\
& \qquad twice: ({\tt int} \to {\tt int} \to {\tt bool}) \to {\tt int} \to {\tt int} \to {\tt bool}  \} \\
P & = \{add = \lambda x {:} {\tt int}, y {:} {\tt int}, z {:} {\tt int}. x + y = z, \\
& \qquad twice = \lambda f {:} {\tt int} \to {\tt int} \to {\tt bool}, x {:} {\tt int}, y {:} {\tt int}. (\exists_{\tt int} z. f \ x \ z \land f \ z \ y) \} \\
G & = \exists_{\tt int} x. add \ 1 \ 2 \ x. 
\end{align*}
For simplicity, DefMono only handles the background theory of linear integer arithmetic (ZLA). 

An input file corresponding to the monotone problem above is
\begin{minted}[linenos]{text}
# This is a sample comment.

environment
add: int -> int -> int -> bool
twice: (int -> int -> bool) -> int -> int -> bool

program
add := \x: int. \y: int. \z:int. x + y = z;
twice := \f: int -> int -> bool. \x: int. \y:int. E z:int. f x z && f z y;

goal
E x: int. add 1 2 x
\end{minted}
As can be seen in line 1, single-line comments start with \mintinline{text}{#}. 
Multiline comments are not supported. 

A sort environment is placed in the \mintinline{text}{environment} section. 
Each statement in the sort environment is allowed to span multiple lines, without endmarkers. 
By contrast, under the \mintinline{text}{program} section, each equation defining a top-level relational variable must end with a semicolon. This restriction is placed to make parsing easier. 

The binding operator $\lambda$ in a lambda calculus is written as \mintinline{text}{\}, and $\exists$ is written as \mintinline{text}{E}. 
The sorts of variables bound by $\lambda$ and $\exists$ must be specified. 

Following the notation in Haskell, conjunction is written as \mintinline{text}{&&}, and disjunction is written as \mintinline{text}{||}. 

For first-order formulas from ZLA, the following operators are included: \mintinline{text}{<}, \mintinline{text}{<=}, \mintinline{text}{=}, \mintinline{text}{>}, and \mintinline{text}{>=}. 
Inequality such as $a \neq b$ can be expressed by
\mintinline{text}{a < b || a > b}. 

\subsubsection{Output format}

DefMono supports two output formats. One is the same format as that of inputs, which is preferable if a readable output is desired. 
The other format is the `pure' SMT-LIB2 format, and it allows outputs to be readily fed into Z3. 

Since target problems produced by the defunctionalization algorithm involve closures (i.e.~entities of the {\bf closr} sort), it is necessary to encode them. 
This is achieved by using a list-like algebraic data type with equality. 
The following example demonstrates how to define closures in a suitable manner for Z3. $Twice$ and $Add$ are top-level relational variables in this example. 
\begin{minted}[linenos]{text}
(declare-datatypes () ((Closr
  Twice
  Add
  (boolCons (boolHd Bool) (boolTl Closr))
  (intCons (intHd Int) (intTl Closr))
  (closrCons (closrHd Closr) (closrTl Closr)) )))
\end{minted}
The name of the algebraic data type, \mintinline{text}{Closr}, is stated in line 1. In lines 2 and 3, \mintinline{text}{Twice} and \mintinline{text}{Add} represent $(Twice) \in A'_{\bf closr}$ and $(Add) \in A'_{\bf closr}$. 
In lines 4--6, \mintinline{text}{boolCons}, \mintinline{text}{intCons}, and \mintinline{text}{closrCons} are data constructors that append Booleans, integers, and closures, respectively, to input closures. 
Each of these constructors comes with selector functions for heads and tails of lists. Note that in the encoding of $(X, t_1, \ldots, t_k)$, where $X$ is a top-level relational variable, the head is $t_n$ rather than $X$. 

\subsection{Evaluation}

In this subsection, I evaluate the performance of DefMono in respect of its verification capability and its running time. 
Additionally, its performance is compared with that of two other higher-order verification tools:
\begin{itemize}
\item HORUS\footnotemark\ by \cite{Ramsay2017}: this runs a refinement type-based algorithm on higher-order Horn clause problems. 
\footnotetext{The source code of HORUS can be found at \url{https://github.com/penteract/HigherOrderHornRefinement}. 
The web interface is available at \url{http://mjolnir.cs.ox.ac.uk/horus/}.}
\item MoCHi\footnotemark\ by \cite{Kobayashi2011,Sato2012}: this runs a CEGAR-based model checking algorithm on higher-order verification problems written in OCaml. 
\footnotetext{The web interface of MoCHi is available at \url{http://www-kb.is.s.u-tokyo.ac.jp/~ryosuke/mochi/}. 
Since the original paper \cite{Kobayashi2011} on MoCHi was published, this web interface has incorporated an extension described in \cite{Unno2013}.}
\end{itemize}

The test suite for DefMono is obtained from that for HORUS by adding one additional test case: `hold'.
`hold' is originally presented in Section~5.3 of \cite{Ramsay2017} as an example that is beyond HORUS's verification capability. 
HORUS's test suite is obtained from MoCHi's. As some of them use the list datatype, which is not supported by HORUS and DefMono, such test cases are disregarded. 
The remaining test cases were then translated from OCaml into Horn clause problems by Cathcart Burn et al.~\cite{Ramsay2017}. 

\subsubsection{Verification capability}

The verification outcomes are summarised in Table~\ref{verification outcomes}. 
\begin{table}
\centering
\begin{tabular}{ c | c | c | c }
Test case & HORUS & MoCHi & DefMono \\
\hline
ack & sat & safe & sat \\
a-max & sat & safe & time out \\
a-max-e & sat & safe & time out \\
herc & sat & safe & sat \\
max & sat & safe & sat \\
mc91 & sat & safe & sat \\
mc91-e & unsat & unsafe & unsat \\
mult & sat & safe & sat \\
mult-e & unsat & unsafe & unsat \\
neg & unsat & safe & sat \\
repeat-e & unsat & unsafe & unsat \\
sum & sat & safe & sat \\
sum-e & unsat & unsafe & unsat \\
hold & unsat & safe & sat
\end{tabular}
\caption{Verification outcomes of HORUS, MoCHi, and DefMono}
\label{verification outcomes}
\end{table}
An input problem being solvable is indicated by \mintinline{text}{sat} in DefMono and HORUS and by \mintinline{text}{safe} in MoCHi. 
In fact, the output of HORUS is \mintinline{text}{unsat} when an input is solvable; however, for readability, it is reversed. 

For HORUS and MoCHi, I used their web interfaces to collect the results. The Z3 used in a web server running HORUS's web interface is version 4.4.1. As for DefMono, I used Z3 version 4.6.0. 

According to \cite{Kobayashi2011}, MoCHi verifies all test cases in HORUS correctly. Furthermore, because `hold' is solvable, MoCHi verifies it correctly as well \cite{Ramsay2017}. In all test cases except `a-max' and `a-max-e', 
because the outputs of DefMono coincide with those of MoCHi, DefMono verifies these test cases correctly as well. Regarding `a-max' and `a-max-e', DefMono does not terminate within two minutes. This shows that DefMono's outputs may be out of Z3's reach. 
In this test suite, MoCHi returns unsafe if and only if DefMono terminates and returns \mintinline{text}{unsat}. Hence, none of the test cases violates completeness or soundness of DefMono. 

With respect to HORUS, `neg' and `hold' demonstrate incompleteness of HORUS (i.e.~they are solvable, but their respective transforms are not typable). Thus, DefMono is more capable than HORUS with respect to `neg' and `hold'. 
On the other hand, HORUS correctly verifies `a-max' and `a-max-e', which cannot be handled by DefMono. 

\subsubsection{Running time}

The running time of DefMono  and HORUS on the test suite is presented in Table~\ref{running time of DefMono} and Table~\ref{running time of HORUS}.
The column `Def' shows the the running time of the defunctionalization algorithm. 
The column `Solving' shows the running time of Z3 v.4.6.0 to solve target monotone problems generated by the defunctionalization algorithm. 
The column `Trans' shows the execution time of transforming an input higher-order Horn clause problem into a first-order one using refinement types.
The experiment was conducted on Windows 10 using an Intel Core i7 CPU.

\begin{table}
\parbox{.45\linewidth}{
\centering
\begin{tabular}{ c | r | r }
Test case & Def (ms) & Solving (ms) \\
\hline
ack & 25.15 & 34.48 \\
a-max & 41.51 &  time out \\
a-max-e & 40.36 &  time out \\
hrec	 & 26.58 & 76.05 \\
max	 & 37.92 & 3347.60 \\
mc91 & 21.39 & 35.66 \\
mc91-e & 23.91 & 20.14 \\
mult & 22.06 & 58.19 \\
mult-e & 14.10 & 70.38 \\
neg & 16.76 & 367.21 \\
repeat-e & 16.32 & 361.41 \\
sum & 13.94 & 23.75 \\
sum-e & 14.05 & 16.72 \\
hold	 & 13.84 & 31.83
\end{tabular}
\caption{Running time of DefMono}
\label{running time of DefMono}
}
\hfill
\parbox{.45\linewidth}{
\begin{tabular}{ c | r | r }
Test case & Trans (ms) & Solving (ms) \\
\hline
ack & 14.37 & 24.03 \\
a-max & 16.36 & 36.60 \\
a-max-e & 16.54 & 38.33 \\
hrec & 15.77 & 34.04 \\
max & 15.27 & 20.80 \\
mc91 & 13.45 & 27.15 \\
mc91-e & 18.94 & 22.80 \\
mult & 13.41 & 30.41 \\
mult-e & 13.43 & 24.39 \\
neg & 15.91 & 24.47 \\
repeat-e & 21.46 & 22.63 \\
sum & 13.72 & 25.49 \\
sum-e & 13.81 & 20.71 \\
hold & 13.39 & 18.02
\end{tabular}
\caption{Running time of HORUS}
\label{running time of HORUS}
}
\end{table}

The running time of Z3 to solve target problems varies greatly from test case to test case: the execution time ranges from 16.72 ms in `sum-e' to 3.35 s in `max'. 
Moreover, as explained before, Z3 does not terminate on `a-max' and `a-max-e' within two minutes.  

As for HORUS, all in all, it takes less time for transformation than DefMono does for defunctionalization. `repeat-e' is the only test case where DefMono is faster than HORUS.
In `repeat-e', the difference in their running time is 5.14 ms. In the remaining test cases, the differences fall between 0.22 ms (in `sum') and 25.15 ms (in `a-max').  

As for Z3's execution time, HORUS is mostly faster than DefMono. The only exceptions are `mc91-e', `sum', and `sum-e', although the differences between HORUS and HORUS in these test cases are insignificant. 
Moreover, the differences between HORUS and DefMono 
in Z3's execution time are considerable in some cases. For instance, in `max', it takes 3.35 s for Z3 to solve the target monotone problem generated by DefMono, whereas it only takes 20.80 ms in HORUS---several orders of magnitude smaller.


\section{Conclusion} \label{chapter on conclusion}

\subsection{Conclusion}

Reynolds's defunctionalization is a viable approach to reducing HoCHC to first-order constrained Horn clauses. In this paper, I have presented an algorithm to defunctionalize HoCHC into first-order constrained Horn clause problems. 
Additionally, I have proved the following:
\begin{enumerate}
\item Type preservation: outputs of the algorithm are well-sorted. 
\item Completeness: if a source HoCHC problem is solvable, the target first-order constrained Horn clause problem generated by the defunctionalization algorithm is also solvable.
\item Soundness: if the target problem is solvable, the source problem is also solvable.
\end{enumerate}
Therefore, type preservation and meaning preservation (i.e.~completeness and soundness) have been established in this work. 

In addition to the theoretical work, I have implemented a system named DefMono that uses the defunctionalization algorithm to verify programs. I have also compared DefMono's performance with that of other higher-order verification tools, HORUS and MoCHi. 
In respect of verification capability, DefMono is less capable than MoCHi because Z3 cannot solve defunctionalized problems of some test cases within two minutes. In comparison with HORUS, DefMono can correctly verify some test cases that HORUS cannot
handle. However, HORUS does not present any bottleneck in Z3's processing of HORUS's outputs, whilst DefMono can cause Z3 to time out. With respect to running time, the defunctionalization-based approach is slower than HORUS. 
This is probably because target problems produced by DefMono use a more complicated background theory than the background theory of source problems.

\subsection{Future work}

I propose three continuations of the present work. 

\paragraph{Continuity of one-step consequence operators}
Whether one-step consequence operators in the monotone semantics are continuous is an interesting question in its own right. I attempted to prove continuity of one-step consequence operators by structural induction on goal terms, as done in the proof
of their monotonicity. However, I encountered a difficulty in the inductive case of function applications: the least upper bound operator $\bigsqcup$ is not guaranteed to distribute over function applications. Hence, I believe this is a key to finding a counterexample. 
In fact, a counterexample to continuity of monotone one-step consequence operators has been found and is presented in a working paper by Jerome Jochems at the University of Oxford. This counterexample shows that the least upper bound operator
does not always distribute over function applications. 

\paragraph{Theory of closures}
One weakness of the defunctionalization-based reduction of higher-order Horn clause problems to first-order ones is that the background theories of target problems involve closures.
In DefMono, closures are implemented using an algebraic data type. 
Fortunately, algebraic data types can be handled by Z3, thanks to recent advances in Horn-clause solving technology. 
Without these advances, it would have been impossible to verify target problems produced by the defunctionalization algorithm. 
Hence, it is another avenue of future work to study, for instance, how ZLA coupled closures can be more efficiently handled in Horn-clause solving. 

\paragraph{Implementation}
One direction is to extend the test suite. As of now, all test cases have order at most 2. Hence, it will be interesting to investigate how DefMono handles test cases of higher order. 

Another direction is to investigate why DefMono does not seem to terminate on `a-max' and `a-max-e'. The run time statistics of Z3 show that only one Boolean variable is created when `a-max' is tested.
This is extremely odd because in other cases where Z3 terminates, many Boolean variables are created. It is therefore likely that Z3 never halts on `a-max'.

\appendix


\section{Supplements for the defunctionalization algorithm} \label{chapter on the algorithm in the appendix}

\subsection{Preprocessing} \label{preprocessing in the appendix}

Let the source monotone problem be $\mathcal{P} = (\Delta, P, G)$. Prior to defunctionalizing $\mathcal{P}$, we need to eliminate all anonymous functions in $P$ and $G$ and then perform $\eta$-expansion to fully expand the outermost lambda
abstractions defining top-level relational variables. 

Every equation in $P$ can be expressed as
\begin{equation}
X {:} \sigma_1 \to \cdots \to \sigma_m \to o = \lambda x_1, \ldots, x_n. E, \label{equation with an anonymous function}
\end{equation}
where $m \leq n$ and $E$ is not a lambda abstraction. 

Anonymous functions refer to lambda abstractions occurring inside $E$ in \eqref{equation with an anonymous function}. 
Suppose that $E$ contains the anonymous function 
\begin{equation*}
\Delta \vdash \lambda x {:} \sigma. F: \sigma \to \rho.
\end{equation*}
Further, assume that the set of free variables occurring in $\lambda x. F$ is
\begin{equation*}
{\tt FV}(\lambda x. F) = \{ n_1, \ldots, n_k \}
\end{equation*}
and that $\Delta \vdash n_i: \sigma_i$ for all $1 \leq i \leq k$. The definition of a fresh top-level relational variable $X'$ is then added to $P$:
\begin{equation}
X' = \lambda n_1 {:} \sigma_1, \ldots, n_k {:} \sigma_k, x {:} \sigma. F. \label{definition of X'}
\end{equation}
As the the actual parameters for the free variables $\{n_1, \ldots, n_k \}$ are specified outside $\lambda x.F$, we need to use lambda abstraction to pass these parameters. The anonymous function
$\lambda x. F$ is then replaced with
\begin{equation*}
X' \ n_1 \ \cdots \ n_k.
\end{equation*}
This process of moving local functions (that is, anonymous functions) into a global scope is called \emph{lambda lifting} in the literature.

We repeat the same step for all the remaining anonymous functions in $P$ and $G$. Notice that some anonymous functions may be inside the definition of $X'$. In order to use a fresh top-level
relational variable for each step, the anonymous functions are eliminated one by one sequentially rather than concurrently. This procedure terminates because the number of anonymous functions is finite. 

Once all anonymous functions are turned into equations, $\eta$-expansion is performed on the right hand side of every equation from $P$. This is guided by the following inference rules:
\begin{center}
\AxiomC{$E \leadsto_{\eta} F$}
\UnaryInfC{$\lambda x. E \leadsto_{\eta} \lambda x. F$}
\DisplayProof
\qquad
\AxiomC{$\Delta \vdash E: \sigma_1 \to \cdots \to \sigma_m \to o$}
\AxiomC{$E \neq \lambda x. F$ for any $F$}
\BinaryInfC{$E \leadsto_{\eta} (\lambda x_1 {:} \sigma_1, \ldots, x_m {:} \sigma_m. E \ x_1 \ \cdots \ x_m)$}
\DisplayProof
\end{center}
The result of $\eta$-expansion on $F$ is obtained by applying $\leadsto_{\eta}$ on $F$. The inference rule on the right encompasses the case when $m = 0$. 
In that case, we have $\lambda \overline x . E \leadsto_\eta \lambda \overline x . E$, where $E : o$.
This transformation is applied to the right hand side of every equation in $P$.

\subsection{Rationale for the algorithm design}

When a source term is of the form $E \ F$, either {\sc (App)} or {\sc (Match)} is applied, depending on whether the function application returns a term of an arrow sort or of base sort. One of its premises of {\sc (App)}
is $(E \ F) \leadsto_{A}^{X} H$, where $\leadsto_{A}^{X}$ is defined by {\sc (App-Base)} and {\sc (App-Arrow)}. Which of these two rules is applied is determined by whether $F$
has a base sort. In both {\sc (App-Base)} and {\sc (App-Arrow)}, neither premises nor conclusions use $\leadsto_{A}^{X}$. Thus, we could remove $\leadsto_{A}^{X}$ completely from the inference rules by merging {\sc (App)}
with each of {\sc (App-Base)} and {\sc (App-Arrow)}. The reason why I do not do this is that the resulting inference rules would be too long to fit the width of a page. This is why $\leadsto_{A}^{X}$ and $\leadsto_{M}$ are necessary.

It is worth observing that $\leadsto$ is only applicable when the source term is of base sort and $\leadsto^{X}$ is only applicable when the source term has an arrow sort. 
This is a rule I imposed on the inference rules to reduce their complexity. 

To explain my reasoning, consider the target term of $Add \ x$  in \eqref{defunctionalization of Add x}. Applying the identity $y = C^0_{Add}$, we can write the target term more succinctly as
\begin{equation*}
Apply_{\bf int} \ C^0_{Add} \ x \ X. 
\end{equation*}
In order to have the inference rules produce this succinct form, we need to split the rule {\sc (App-Base)} into two rules corresponding to two cases: the case when $E'$ 
is a logical formula and the case when $E'$ is a single variable symbol. In the first case, we cannot write $Apply \ E' \ F' \ X$, since $E'$ is a logical formula rather than a variable symbol. By contrast, in the second case, $Apply \ E' \ F' \ X$
is a valid target term. 

If this idea were implemented, we would have $Add \leadsto C^0_{Add}$ instead of $Add \leadsto^{X} X = C^0_{Add}$. 
The former is more natural and less confusing than the latter. 
However, it does not seem elegant to split {\sc (App-Base)}, because we would need to work out whether $E'$ consists only of a single symbol. 
Also, splitting {\sc (App-Base)} will increase the total number of inference rules. 
Therefore, I opted to enforce the rule that whenever the source term has an arrow sort, the parameter $X$ can be passed. 
Consequently, when a source goal term is a top-level relational variable, the term has an arrow sort and hence its target term must accept a parameter. 
This is the reason behind the bizarre looking {\sc (TopVar)}.


\section{Monotonicity of extracted valuations} \label{chapter on monotonicity of alpha'}

This section presents how to establish monotonicity of $\alpha'$, which is formally defined in Subsection~\ref{formalization of valuation extraction}. 

\subsection{Preliminaries}

First, I prove a lemma that characterizes orders of higher-order elements.

\begin{lemma} \label{characterization of orders of higher-order elements}
Assume that $f_1$ and $f_2$ have sort $\sigma_1 \to \cdots \to \sigma_k \to o$, where $k \geq 0$ and each $\sigma_i$ is either a relational arrow sort or a base sort. Then $f_1 \subseteq f_2$ if and only if for each 
$t \in \mathcal{M} \llbracket \sigma_1 \rrbracket \times \cdots \times \mathcal{M} \llbracket \sigma_k \rrbracket$, we have $f_1(t) \subseteq_{o} f_2(t)$.
\end{lemma}
\begin{proof}
For both directions, the claim is proved by induction on $k$. In this proof, I use curried notation and non-curried notation interchangeably. Hence, if an $n$-tuple is input to a function, the $n$ components of the tuple are fed into the function separately. 

First, I prove $(\Rightarrow)$. For the base case, when $k = 0$, we have $f_1, f_2: o$. By assumption, $f_1 \subseteq_{o} f_2$ and hence the claim holds. 

For the inductive case, suppose that $f_1 \subseteq f_2$. By the definition of $\subseteq$, for all $c_1 \in \mathcal{M} \llbracket \sigma_1 \rrbracket$, we have $f_1 \ c_1 \subseteq f_2 \ c_1$. 
This is true regardless of whether $\sigma_1$ is a relational sort or a non-propositional base sort. Now by the inductive hypothesis, as $f_1 \ c_1 \subseteq f_2 \ c_1$, for all $(c_2, \ldots, c_k) \in \mathcal{M} \llbracket \sigma_2 \rrbracket \times \cdots \times \mathcal{M} \llbracket \sigma_k \rrbracket$, we obtain
\begin{equation*}
(f_1 \ c_1)(c_2, \ldots, c_k) \subseteq_{o} (f_2 \ c_1)(c_2, \ldots, c_k).
\end{equation*}
Thus, for all $c_1, \ldots, c_k$ of appropriate sorts,
\begin{equation*}
f_1(c_1, \ldots, c_k) \subseteq_{o} f_2(c_1, \ldots, c_k),
\end{equation*}
as required.

Now I turn to $(\Leftarrow)$. For the base case, when $k = 0$, if $f_1 \subseteq_{o} f_2$, the claim immediately follows.

For the inductive case, suppose that for all for all $(c_1, \ldots, c_k) \in \mathcal{M} \llbracket \sigma_1 \rrbracket \times \cdots \times \mathcal{M} \llbracket \sigma_k \rrbracket$, we have 
\begin{equation*}
f_1(c_1, \ldots, c_k) \subseteq_{o} f_2(c_1, \ldots, c_k). 
\end{equation*}
Fix arbitrary $c_1 \in \mathcal{M} \llbracket \sigma_1 \rrbracket$. Then for all $(c_2, \ldots, c_k) \in \mathcal{M} \llbracket \sigma_2 \rrbracket \times \cdots \times \mathcal{M} \llbracket \sigma_k \rrbracket$, we have
\begin{equation*}
(f_1 \ c_1)(c_2, \ldots, c_k) \subseteq_{o} (f_2 \ c_1)(c_2, \ldots, c_k).
\end{equation*}
Hence, by the inductive hypothesis, $f_1 \ c_1 \subseteq f_2 \ c_1$. Because $c_1$ is arbitrary, by the definition of $\subseteq$, $f_1 \subseteq f_2$. This concludes the proof.
\end{proof}

The next lemma characterizes monotone functions. 

\begin{lemma} \label{characterization of monotone functions}
Assume $f \in \mathcal{M} \llbracket \sigma_1 \rrbracket \Rightarrow \cdots \Rightarrow \mathcal{M} \llbracket \sigma_k \rrbracket \Rightarrow \mathbbm{2}$, where $k \geq 0$ and each $\sigma_i$ is either a relational arrow sort or a base sort. $f$ is 
monotone if and only if for each $t_1, t_2 \in \mathcal{M} \llbracket \sigma_1 \rrbracket \times \cdots \times \mathcal{M} \llbracket \sigma_k \rrbracket$ and $t_1 \subseteq t_2$, we have $f(t_1) \subseteq_{o} f(t_2)$. Here, $t_1 \subseteq t_2$
holds if and only if the order holds in each component. 
\end{lemma}
\begin{proof}
For both directions, the claim is proved by induction on $k$. In this proof, I use curried notation and non-curried notation interchangeably. Hence, if an $n$-tuple is input to a function, the $n$ components of the tuple are fed into the function separately. 

I first start with $(\Rightarrow)$. For the base case, when $k = 0$, the claim is clearly true.

For the inductive case, suppose that $f$ is monotone. By definition, we have
\begin{equation*}
\mathcal{M} \llbracket \sigma_1 \to \cdots \to \sigma_k \to o \rrbracket = \mathcal{M} \llbracket \sigma_1 \rrbracket \Rightarrow_{m} \mathcal{M} \llbracket \sigma_2 \to \cdots \to \sigma_k \to o \rrbracket.
\end{equation*}
It follows from the definition of $\Rightarrow_{m}$ that for any $c_1, d_1 \in \mathcal{M} \llbracket \sigma_1 \rrbracket$, if $c_1 \subseteq d_1$, then $X \ c_1 \subseteq X \ d_1$. Thus, it follows from 
Lemma~\ref{characterization of orders of higher-order elements} that 
\begin{equation}
(f \ c_1)(c_2, \ldots, c_k) \subseteq_{o} (f \ d_1)(c_2, \ldots, c_k). \label{first equation in the first direction for the characterization of monotone functions}
\end{equation}
Furthermore, because $f \ d_1$ is monotone, by the inductive hypothesis, for any $(c_2, \ldots, c_k)$ and $(d_2, \ldots, d_k)$ from $\mathcal{M} \llbracket \sigma_2 \rrbracket \times \cdots \times \mathcal{M} \llbracket \sigma_k \rrbracket$
 such that $(c_2, \ldots, c_k) \subseteq (d_2, \ldots, d_k)$, we have
\begin{equation}
(f \ d_1)(c_2, \ldots, c_k) \subseteq_{o} (f \ d_1)(d_2, \ldots, d_k). \label{second equation in the first direction for the characterization of monotone functions}
\end{equation}
Combining \eqref{first equation in the first direction for the characterization of monotone functions} and \eqref{second equation in the first direction for the characterization of monotone functions} gives
\begin{equation*}
(f \ c_1)(c_2, \ldots, c_k) \subseteq_{o} (f \ d_1)(d_2, \ldots, d_k).
\end{equation*}
Therefore, for any $t_1, t_2 \in \mathcal{M} \llbracket \sigma_1 \rrbracket \times \cdots \times \mathcal{M} \llbracket \sigma_k \rrbracket$ such that 
$t_1 \subseteq t_2$, we have
\begin{equation*}
f(t_1) \subseteq_{o} f(t_2),
\end{equation*}
as required. 

Now I turn to $(\Leftarrow)$. For the base case, when $k = 0$, the claim is vacuously true. 

For the inductive case, by assumption, for any $t_1, t_2 \in \mathcal{M} \llbracket \sigma_1 \rrbracket \times \cdots \times \mathcal{M} \llbracket \sigma_k \rrbracket$ such that $t_1 \subseteq t_2$, we have $f(t_1) \subseteq_{o} f(t_2)$. 
Now fix $c_1, d_1 \in \mathcal{M} \llbracket \sigma_1 \rrbracket$ such that $c_1 \subseteq d_1$. Then by the assumption, for any $(c_2, \ldots, c_k), (d_2, \ldots, d_k) \in \mathcal{M} \llbracket \sigma_2 \rrbracket \times \cdots \times \mathcal{M} \llbracket \sigma_k \rrbracket$ such that $(c_2, \ldots, c_k) \subseteq (d_2, \ldots, d_k)$, we have
\begin{align}
(f \ c_1)(c_2, \ldots, c_k) & \subseteq_{o} (f \ c_1)(d_2, \ldots, d_k) \label{first equation in the second direction of the characterization of monotone functions} \\
(f \ d_1)(c_2, \ldots, c_k) & \subseteq_{o} (f \ d_1)(d_2, \ldots, d_k) \label{second equation in the second direction of the characterization of monotone functions} \\
(f \ c_1)(c_2, \ldots, c_k) & \subseteq_{o} (f \ d_1)(c_2, \ldots, c_k). \label{third equation in the second direction of the characterization of monotone functions}
\end{align}
Applying the inductive hypothesis to \eqref{first equation in the second direction of the characterization of monotone functions} yields that $f \ c_1$ is monotone. Likewise, by the application of the inductive hypothesis to \eqref{second equation in the second direction of the characterization of monotone functions}, $X \ d_1$ is also monotone. Further, from Lemma~\ref{characterization of orders of higher-order elements} and 
\eqref{third equation in the second direction of the characterization of monotone functions}, we obtain
\begin{equation*}
f \ c_1 \subseteq f \ d_1.
\end{equation*}
To summarise, $f \ c_1$ and $f \ d_1$ are both monotone, and $f \ c_1 \subseteq f \ d_1$ whenever $c_1 \subseteq d_1$. Therefore, $f$ is monotone by definition. This concludes the proof. 
\end{proof}

\subsection{Monotonicity of $\alpha'$}

Thus, in order for $\alpha'(X)$ to be monotone, where $\Delta' \vdash X: \sigma_1 \to \cdots \to \sigma_m \to o$, for any $t_1, t_2 \in \mathcal{M} \llbracket \sigma_1 \rrbracket \times \cdots \times \mathcal{M} \llbracket \sigma_m \rrbracket$ such that
$t_1 \subseteq t_2$, we should have
\begin{equation*}
\alpha' (X) (t_1) \subseteq_{o} \alpha' (X) (t_2). 
\end{equation*}

This holds for $X = IOMatch_{B}$, where $B \in \mathbb{B}'$. If $B \neq o$ and $t_1, t_2 \in \mathcal{M} \llbracket {\bf closr} \to B \rrbracket$, then $t_1 \subseteq t_2$ implies 
$t_1 = t_2$. Otherwise, if $B = o$, by the monotonicity of $\alpha$, $iomatch_{B}$ is monotone as well. 

However, this does not hold for $X = Apply_{o}$. For instance, suppose that $Y \in {\tt dom}(\Delta)$ and that $\Delta \vdash Y: {\bf nat} \to o \to {\bf nat} \to o$. Then, 
\begin{equation*}
((Y, 2), 0, (Y, 2, 0)) \subseteq ((Y, 2), 1, (Y, 2, 0))
\end{equation*}
and yet
\begin{equation*}
apply_{o} \ (Y, 2) \ 0 \ (Y, 2, 0) \nsubseteq_{o} apply_{o} \ (Y, 2) \ 1 \ (Y, 2, 0)
\end{equation*}
as the left hand side evaluates to 1, whereas the right hand side evaluates to 0. Therefore, $\alpha'$ is not monotone. 

In this way, $\alpha'$ is ``nearly'' monotone, apart from $apply_{o}$. $apply_{B}$ augments an input to an input closure, thereby simulating function application that still yields a strictly partially applied function. $\mathcal{P}$'s rough equivalent
of $apply_{B}$ is function application. However, monotonicity of function application in $\mathcal{P}$ does not carry over to $\mathcal{P}'$, for the way $apply_{B}$ simulates function application is different from genuine function application in
$\mathcal{P}$. 

By contrast, there is a nice correspondence between branches of $IOMatch_{B}$ and top-level relational variables from $\Delta$. The monotonicity of $\alpha(X)$, where $X \in {\tt dom}(\Delta)$, carries over to $iomatch_{B}$
that corresponds to $X$, although this is true only when $B = o$; if $B \neq o$, $iomatch_{B}$ is monotone regardless of monotonicity of $\alpha(X)$. 

To fix the issue of monotonicity of $apply_{o}$, observe that $\alpha'$ can be interpreted as a standard valuation. This can be established by the next proposition.

\begin{proposition} \label{relationship between standard and monotone functions when order is 2}
Assume $f: b_1 \to \cdots \to b_m$, where each $b_i$ is a base sort. If $f \in \mathcal{M} \llbracket b_1 \rrbracket \Rightarrow \cdots \Rightarrow \mathcal{M} \llbracket b_m \rrbracket$, where $f$ is not necessarily monotone, 
then $f \in \mathcal{S} \llbracket b_1 \to \cdots \to b_m \rrbracket$.
\end{proposition}
\begin{proof}
Immediately follows from the fact that $\mathcal{M} \llbracket b \rrbracket = \mathcal{S} \llbracket b \rrbracket$ if $b$ is a base sort.
\end{proof}

Despite its triviality, this proposition is important. For example, consider $X: (o \to o) \to o$, which has order 3. Also, let $\beta$ a ``nearly'' monotone valuation for $X$ in the sense that 
\begin{align*}
\beta(X) & \in \mathcal{M} \llbracket o \to o \rrbracket \Rightarrow \mathcal{M} \llbracket o \rrbracket \\
& = (\mathbbm{2} \Rightarrow_{m} \mathbbm{2}) \Rightarrow \mathbbm{2},
\end{align*}
where the second $\Rightarrow$ on the second line is not $\Rightarrow_{m}$. When we want to extend the monotone interpretation of $X$ to the standard semantics, there is no straightforward way to do so, since $\beta$ does not define the result of
$\alpha(X)$ applied to $f$ when $f \in (\mathbbm{2} \Rightarrow \mathbbm{2}) \setminus (\mathbbm{2} \Rightarrow_{m} \mathbbm{2})$.

In contrast, if the sort of $X$ has order 2, we can extend the monotone interpretation of $X$ to the standard semantics in a straightforward fashion. 

Proposition~\ref{relationship between standard and monotone functions when order is 2} can be applied to any function occurring in $P'$ because any $f \in {\tt dom}(\Delta')$ has order 2 (by convention, it is assumed that all top-level relational variables
have arrow sorts) and any $f \in {\tt dom}(\mathbb{S})$ has order at most 2. Furthermore, we do not have existential quantifiers over higher-order variables in $P'$. For these two reasons, the standard semantics of $P'$ coincides with the monotone
semantics of $P'$. That is, $\mathcal{M} \llbracket s \rrbracket (\alpha') = \mathcal{S} \llbracket s \rrbracket (\alpha')$ holds, given that $s$ contains no existential quantifiers over higher-order variables and all symbols occurring in $s$ have order at most 2.
Therefore, $\alpha'$ can be viewed as a standard valuation of $P'$. 

In addition, $T^{\mathcal{M}}_{P': \Delta'}$ is equivalent to $T^{\mathcal{S}}_{P': \Delta'}$, although $T^{\mathcal{M}}_{P': \Delta'}$ is not guaranteed to be monotone if an input is not
drawn from $\mathcal{M} \llbracket \Delta' \rrbracket$. Hence, Lemma~2 in \cite{Ramsay2017} does not apply if a valuation is nearly but not monotone:

\begin{lemma} \label{one-step consequence operator is monotone}
$\mathcal{M} \llbracket \Delta \vdash G: \rho \rrbracket \in \mathcal{M} \llbracket \Delta \rrbracket \Rightarrow_{m} \mathcal{M} \llbracket \rho \rrbracket$, where $G$ is any goal term. 
Also, $T^{\mathcal{M}}_{P: \Delta} \in \mathcal{M} \llbracket \Delta \rrbracket \Rightarrow_{m} \mathcal{M} \llbracket \Delta \rrbracket$. 
\end{lemma}
\begin{proof}
``Immediately follows from the fact that {\tt mexists}, {\tt and} and {\tt or} are monotone and all the construction [in the inductive definition of $\mathcal{M} \llbracket \Delta \vdash G: \rho \rrbracket$] are monotone combinations''. \cite{Ramsay2017}
Note that the interpretations of constant symbols from $\mathbb{S}$ are required to be monotone as well. 
\end{proof}

Now suppose that $\alpha'$ is a prefix of $T^{\mathcal{M}}_{P': \Delta'} = T^{\mathcal{S}}_{P': \Delta'}$ and satisfies $\mathcal{M} \llbracket G' \rrbracket (\alpha') = \mathcal{S} \llbracket G' \rrbracket (\alpha') = 0$, where $G'$ is the goal formula component of $\mathcal{P}'$. In other words, suppose that $\alpha'$ is a solution to $\mathcal{P}'$ under the standard semantics. 
I restate Theorem~\ref{relationship between Horn clause problems and monotone problems original} with a slightly different notation (this is originally Theorem~2 in \cite{Ramsay2017}):

\begin{theorem} \label{relationship between Horn clause problems and monotone problems}
The higher-order constrained Horn clause problem $(\Delta', D', G')$ is solvable if and only if the monotone problem $(\Delta', P_{D'}, G')$ is solvable. 
\end{theorem}

For each monotone problem $(\Delta', P', G')$, there exists a higher-order constrained Horn clause problem $(\Delta', D', G')$ such that $P' = P_{D'}$. Further, Horn clause problems are interpreted using the standard semantics. Consequently, 
we obtain Lemma~1 from \cite{Ramsay2017}:

\begin{theorem} \label{prefixed points of the standard one-step consequence operator are models of D'}
For definite formula $D'$, the prefixed points of $T^{\mathcal{S}}_{P_{D'}}$ are exactly the models of $D'$. 
\end{theorem}

Finally, the next theorem ensures the existence of a monotone solution to $\mathcal{P}'$, provided that $\alpha'$ is a solution to $\mathcal{P}'$ under the standard semantics. 

\begin{theorem} \label{a monotone solution exists if a standard solution exists}
If $\alpha'$ is a solution to $\mathcal{P}'$ under the standard semantics, then $\mathcal{P}'$ is solvable under the monotone semantics. 
\end{theorem}
\begin{proof}
Let $\alpha'$ be a standard solution to $\mathcal{P}'$.
Suppose $(\Delta', D', G')$ is the higher-order constrained Horn clause problem that is equivalent to $\mathcal{P}'$; i.e.~$P' = P_{D'}$. Such a Horn clause problem is well-defined as there is one-one correspondence between higher-order constrained 
Horn clause problems and monotone problems. 

Since it is given that $\alpha'$ is a solution to $\mathcal{P}'$ and hence is a prefixed point of $T^{\mathcal{S}}_{P': \Delta'}$, by 
Theorem~\ref{prefixed points of the standard one-step consequence operator are models of D'}, $\alpha'$ is also a model of $D'$. Further, $\mathcal{S} \llbracket G' \rrbracket (\alpha') = 0$. Hence, $\alpha'$ is a solution to $(\Delta', D', G')$. 

Lastly, it follows from Theorem~\ref{relationship between Horn clause problems and monotone problems} that $\mathcal{P'} = (\Delta', P', G')$ is solvable under the monotone semantics because $P' = P_{D'}$ by assumption. 
\end{proof}


\section{Supplements for meaning preservation} \label{chapter on the supplements for meaning preservation}

This section presents the proofs for the results introduced in Subsection~\ref{section on meaning preservation} and explains difficulties with applying valuation extraction to the proof of soundness.

\subsection{First direction}

\relationshipbetweenthesemanticsofsourceandtargetgoalterms*
\begin{proof}
The proof proceeds by structural induction on $s$. First, I consider the case when $\sigma$ is an arrow sort. 

For the base case, if $s = x$, where $x$ is an ordinary variable of an arrow sort, $t$ is equal to $X = x$ due to {\sc (Var-Arrow)}. Because $x$ is a free variable, it must be included in ${\tt dom}(\alpha)$. Therefore, $c = \alpha'(x) \in A_{\bf closr}'$ works.
This $c$ is unique because if $c_1 \neq_{\bf closr} c$, then $c = \alpha(x)$ and $c_1 = \alpha(x)$ cannot hold simultaneously (this is due to the fact that $(=)$ is the same as $(=_{\bf closr})$ in this setting). 

By the definition of $\alpha'$, we have $\text{expand}_{\alpha}(\alpha'(x)) = \alpha(x)$. It follows that
\begin{align*}
\text{expand}_{\alpha}(c) & = \text{expand}_{\alpha}(\alpha'(x)) \\
& = \alpha (x) \\
& = \mathcal{M} \llbracket \Gamma \vdash x: \sigma \rrbracket (\alpha) \\
& = \mathcal{M} \llbracket \Gamma \vdash s: \sigma \rrbracket (\alpha). 
\end{align*}
Therefore, $\text{expand}_{\alpha}(c) = \mathcal{M} \llbracket \Gamma \vdash s: \sigma \rrbracket (\alpha)$ holds. 

Another base case is when $s \in \Delta$; i.e.~$s$ is a top-level relational variable. By {\sc (TopVar)}, $t$ is equal to $X = C^0_{s}$. The only value of $c$ that satisfies $\mathcal{M} \llbracket X = C^0_{s} \rrbracket (\alpha' \cup [X \mapsto c]) = 1$
is $(s) \in A_{\bf closr}'$ because $C^0_{s}$, which is a constant symbol, is by default interpreted as $(s)$. Here, $(s) \in A_{\bf closr}'$ is a 1-tuple containing $s$. Thus, such $c$ is unique. Further, we have $\text{expand}_{\alpha}(c) = \alpha (s)$. 
Therefore, the claim holds.

For the inductive case, $s$ is transformed into $t$ by either {\sc (App-Base)} or  {\sc (App-Arrow)}. 

Assume $s = E \ F$, where $E$ has an arrow sort and $F$ has a base sort. By {\sc (App-Base)}, $t$ is equal to 
\begin{equation*}
\exists_{\bf closr} x. (E' \land Apply_{\bf closr} \ x \ F' \ X)),
\end{equation*}
where $E \leadsto^{x} E'$ and $F \leadsto F'$. Applying the inductive hypothesis to $E$, we have a unique $c_1 \in A_{\bf closr}'$ such that
\begin{equation*}
\mathcal{M} \llbracket \Gamma' \vdash E': {\bf closr} \rrbracket (\alpha' \cup [x \mapsto c_1]) = 1. 
\end{equation*}
Additionally, this $c_1$ satisfies 
\begin{equation*}
\text{expand}_{\alpha}(c_1) = \mathcal{M} \llbracket E \rrbracket (\alpha). 
\end{equation*}
Further, applying the inductive hypothesis to $F$, we have $\mathcal{M} \llbracket F \rrbracket (\alpha) = \mathcal{M} \llbracket F' \rrbracket (\alpha')$. 

Consequently, we obtain
\begin{align*}
& \mathcal{M} \llbracket \exists_{\bf closr} x. (E' \land Apply_{\bf closr} \ x \ F' \ X)) \rrbracket (\alpha' \cup [X \mapsto c]) = 1 \\
\iff {} & \mathcal{M} \llbracket Apply_{\bf closr} \ x \ F' \ X \rrbracket (\alpha' \cup [x \mapsto c_1, X \mapsto c]) = 1 \\
\iff {} & apply_{\bf closr} \ c_1 \ \mathcal{M} \llbracket F' \rrbracket (\alpha') \ c = 1 \\
\iff {} & (c == append(c_1, \mathcal{M} \llbracket F' \rrbracket (\alpha'))) = 1 \\
\iff {} & c = append(c_1, \mathcal{M} \llbracket F' \rrbracket (\alpha'))
\end{align*}
Here, $(==)$ is a comparator. The second line follows from the uniqueness of $c_1$. The third line follows from the fact that $\alpha'$ interprets $Apply_{\bf closr}$ as $apply_{\bf closr}$. The fourth line follows form the definition of $apply_{\bf closr}$. 

Thus, to satisfy $\mathcal{M} \llbracket Apply_{\bf closr} \ x \ F' \ X \rrbracket (\alpha' \cup [x \mapsto c_1, X \mapsto c]) = 1$, we should set $c$ to $\text{append}(c_1, \mathcal{M} \llbracket F' \rrbracket (\alpha'))$. 
Hence, there is indeed a unique $c \in A_{\bf closr}'$ that satisfies $\mathcal{M} \llbracket \Gamma' \vdash t: {\bf closr} \rrbracket (\alpha' \cup [X \mapsto c]) = 1$. 

Furthermore, from $c = \text{append}(c_1, \mathcal{M} \llbracket F' \rrbracket (\alpha'))$, we derive
\begin{align*}
\text{expand}_{\alpha}(c) & = \text{expand}_{\alpha}(\text{append}(c_1, \mathcal{M} \llbracket F' \rrbracket (\alpha'))) \\
& = \text{expand}_{\alpha}(c_1) \ \text{expand}_{\alpha}(\mathcal{M} \llbracket F' \rrbracket (\alpha')) \\
& = \mathcal{M} \llbracket E \rrbracket (\alpha) \ \mathcal{M} \llbracket F' \rrbracket (\alpha')) \\
& = \mathcal{M} \llbracket E \rrbracket (\alpha) \ \mathcal{M} \llbracket F \rrbracket (\alpha)) \\
& = \mathcal{M} \llbracket E \ F \rrbracket (\alpha). 
\end{align*}
The second equality follows from the inductive definition of the $\text{expand}_{\alpha}$ function. 

Lastly, if {\sc (App-Arrow)} is used, we have $s = E \ F$, and $t$ is equal to
\begin{equation*}
\exists_{\bf closr} x. (E' \land \exists_{\bf closr} y. (F' \land Apply_{\bf closr} \ x \ y \ X)),
\end{equation*}
where $E \leadsto^{x} E'$ and $F \leadsto^{y} F'$. By the inductive hypothesis, we have unique $c_1, c_2 \in A_{\bf closr}'$ such that 
\begin{align*}
& \mathcal{M} \llbracket \Gamma' \vdash E': {\bf closr} \rrbracket (\alpha' \cup [x \mapsto c_1]) = 1 \\
& \mathcal{M} \llbracket \Gamma' \vdash F': {\bf closr} \rrbracket (\alpha' \cup [y \mapsto c_2]) = 1.
\end{align*}
Additionally, $c_1$ and $c_2$ satisfy
\begin{align*}
\text{expand}_{\alpha}(c_1) & = \mathcal{M} \llbracket E \rrbracket (\alpha) \\
\text{expand}_{\alpha}(c_2) & = \mathcal{M} \llbracket F \rrbracket (\alpha). 
\end{align*}

As a consequence, we have
\begin{align*}
& \mathcal{M} \llbracket \exists_{\bf closr} x. (E' \land \exists_{\bf closr} y. (F' \land Apply_{\bf closr} \ x \ y \ X)) \rrbracket (\alpha' \cup [X \mapsto c]) = 1 \\
\iff {} & \mathcal{M} \llbracket Apply_{\bf closr} \ x \ y \ X \rrbracket ([x \mapsto c_1, y \mapsto c_2, X \mapsto c]) = 1 \\
\iff {} & apply_{\bf closr} \ c_1 \ c_2 \ c = 1 \\
\iff {} & (c == \text{append}(c_1, c_2)) = 1 \\
\iff {} & c = \text{append}(c_1, c_2). 
\end{align*}
The second line follows from the uniqueness of $c_1$ and $c_2$. The third line follows from the the interpretation of $Apply_{\bf closr}$ by $\alpha'$. The fourth line follows from the definition of $apply_{\bf closr}$. 

Therefore, the only value of $c$ that satisfies $\mathcal{M} \llbracket \Gamma' \vdash t: {\bf closr} \rrbracket (\alpha' \cup [X \mapsto c]) = 1$ is $\text{append}(c_1, c_2)$.

Moreover, $c = \text{append}(c_1, c_2)$ yields
\begin{align*}
\text{expand}_{\alpha}(\text{append}(c_1, c_2)) & = \text{expand}_{\alpha}(c_1) \ \text{expand}_{\alpha}(c_2) \\
& = \mathcal{M} \llbracket E \rrbracket (\alpha) \ \mathcal{M} \llbracket F \rrbracket (\alpha) \\
& = \mathcal{M} \llbracket E \ F \rrbracket (\alpha).
\end{align*}
The first equality follows from the inductive definition of the $\text{expand}_{\alpha}$ function. Therefore, the claim is true. 

Next, consider the case when $\sigma$ is a base sort. 

For the base case, if $s \in Fm \cup Tm$, we have $s \leadsto s$ by {\sc (ConstrLan)}. All free variables occurring in first-order terms from a constraint language have base sorts (this is proved in Theorem~\ref{free variables in first-order terms}). 
Further, since $\text{expand}_{\alpha}(c) = c$ when $c$ is of base sort from $\mathbb{B}$, $\alpha$ and $\alpha'$ have the same interpretation of all free variables in $s$. Also, $A$ and $A'$ have the same universes
for each $b \in \mathbb{B}$ and have the same interpretation of constant symbols from $\mathbb{S}$. Therefore, $s$ has the same meaning in both $\langle A, \alpha \rangle$ and $\langle A', \alpha' \rangle$. Thus, the claim is true.

The case when {\sc (Var-Base)} is used can be proved straightforwardly. 

For the inductive case, if $s = E \land F$, by {\sc (LogSym)}, $t$ is equal to $E' \land F'$, where $E \leadsto E'$ and $F \leadsto F'$. It follows from the inductive hypothesis that
\begin{align*}
\mathcal{M} \llbracket E \rrbracket (\alpha) & = \mathcal{M} \llbracket E' \rrbracket (\alpha') \\
\mathcal{M} \llbracket F \rrbracket (\alpha) & = \mathcal{M} \llbracket F' \rrbracket (\alpha').
\end{align*}
Therefore, we obtain
\begin{align*}
\mathcal{M} \llbracket E' \land F' \rrbracket (\alpha') & = \mathcal{M} \llbracket E' \rrbracket (\alpha') \land \mathcal{M} \llbracket F' \rrbracket (\alpha') \\
& = \mathcal{M} \llbracket E \rrbracket (\alpha) \land \mathcal{M} \llbracket F \rrbracket (\alpha) \\
& = \mathcal{M} \llbracket E \land F \rrbracket (\alpha)
\end{align*}
as required. The case for $s$ being $E \lor F$ can be proved in the same manner. 

If $s = \exists_{b} x. F$, by {\sc (Exi)}, we have $t = \exists_{b} x. F'$, where $F \leadsto F'$. By the inductive hypothesis,
\begin{equation*}
\mathcal{M} \llbracket F \rrbracket (\alpha \cup [x \mapsto c]) = \mathcal{M} \llbracket F' \rrbracket (\alpha' \cup [x \mapsto c])
\end{equation*}
for any $c \in A_{b} = A_{b}'$. Thus, we obtain
\begin{align*}
\mathcal{M} \llbracket \exists_{b} x. F' \rrbracket (\alpha') & = \exists c \in A_{b}'. \mathcal{M} \llbracket F' \rrbracket (\alpha' \cup [x \mapsto c]) \\
& = \exists c \in A_{b}. \mathcal{M} \llbracket F \rrbracket (\alpha \cup [x \mapsto c]) \\
& = \mathcal{M} \llbracket \exists_{b} x. F \rrbracket (\alpha).
\end{align*}
Therefore, the claim holds. 

It is essential that the existential quantifier is bound to a variable of base sort as opposed to an arrow sort. If $\exists_{\sigma} x. F \leadsto \exists_{\bf closr} x. F'$, where $\sigma$ is an arrow sort, 
it is possible that $\mathcal{M} \llbracket \exists_{\sigma} x. F \rrbracket (\alpha) = 1$ and yet $\mathcal{M} \llbracket \exists_{\bf closr} x. F' \rrbracket (\alpha') = 0$. This is because $\mathcal{M} \llbracket \sigma \rrbracket$ contains
functions that cannot be represented by any element of $A_{\bf closr}'$. This is why we need to eliminate existential quantifiers over higher-order variables. 

Next, assume $s = E \ F$. If $F$ is of base sort, {\sc (Match-Base)} is applied to defunctionalize $s$ into $t$, yielding
\begin{equation*}
t = \exists_{\bf closr} x. (E' \land IOMatch_{\sigma} \ x \ F'),
\end{equation*}
where $E \leadsto^{x} E'$ and $F \leadsto F'$. Applying the inductive hypothesis to $E$, we have a unique $c \in A_{\bf closr}'$ such that
\begin{equation*}
\mathcal{M} \llbracket \Gamma' \vdash E': {\bf closr} \rrbracket (\alpha' \cup [x \mapsto c]) = 1. 
\end{equation*}
Also, this $c$ satisfies $\text{expand}_{\alpha}(c) = \mathcal{M} \llbracket E \rrbracket (\alpha)$. Additionally, applying the inductive hypothesis to $F$, we have $\mathcal{M} \llbracket F \rrbracket (\alpha) = \mathcal{M} \llbracket F' \rrbracket (\alpha')$.

By the uniqueness of $c$,
\begin{equation*}
\mathcal{M} \llbracket \exists_{\bf closr} x. (E' \land IOMatch_{\sigma} \ x \ F') \rrbracket (\alpha') = \mathcal{M} \llbracket IOMatch_{\sigma} \ x \ F') \rrbracket (\alpha' \cup [x \mapsto c]). 
\end{equation*}
Furthermore, we obtain
\begin{align*}
\mathcal{M} \llbracket IOMatch_{\sigma} \ x \ F') \rrbracket (\alpha' \cup [x \mapsto c]) & =iomatch_{\sigma} \ c \ \mathcal{M} \llbracket F' \rrbracket (\alpha') \\
& = iomatch_{\sigma} \ c \ \mathcal{M} \llbracket F \rrbracket (\alpha) \\
& = \text{expand}_{\alpha}(c) \ \text{expand}_{\alpha}(\mathcal{M} \llbracket F \rrbracket (\alpha)) \\
& = \mathcal{M} \llbracket E \rrbracket (\alpha) \ \mathcal{M} \llbracket F \rrbracket (\alpha) \\
& = \mathcal{M} \llbracket E \ F \rrbracket (\alpha).
\end{align*}
The second equality follows from the identity $\mathcal{M} \llbracket F \rrbracket (\alpha) = \mathcal{M} \llbracket F' \rrbracket (\alpha')$. The third equality follows from the definition of $iomatch_{\sigma}$. The fourth equality follows from
the definition of $c$ and the definition of $\text{expand}_{\alpha}$ when the input has a base sort. 

Therefore, $\mathcal{M} \llbracket \Gamma \vdash s: b \rrbracket (\alpha) = \mathcal{M} \llbracket \Gamma' \vdash t: b \rrbracket (\alpha')$ holds. 

Otherwise, if $F$ is of an arrow sort, {\sc (Match-Arrow)} is applied to defunctionalize $s$ into $t$, where $t$ is 
\begin{equation*}
\exists_{\bf closr} x. (E' \land \exists_{\bf closr} y. IOMatch_{\bf closr} \ x \ y),
\end{equation*}
where $E \leadsto^{x} E'$ and $F \leadsto^{y} F'$. By the inductive hypothesis, we have unique $c_1, c_2 \in A_{\bf closr}'$ such that 
\begin{align*}
& \mathcal{M} \llbracket \Gamma' \vdash E': {\bf closr} \rrbracket (\alpha' \cup [x \mapsto c_1]) = 1 \\
& \mathcal{M} \llbracket \Gamma' \vdash F': {\bf closr} \rrbracket (\alpha' \cup [y \mapsto c_2]) = 1.
\end{align*}
Further, $c_1$ and $c_2$ satisfy 
\begin{align*}
\text{expand}_{\alpha}(c_1) & = \mathcal{M} \llbracket E \rrbracket (\alpha) \\
\text{expand}_{\alpha}(c_2) & = \mathcal{M} \llbracket F \rrbracket (\alpha). 
\end{align*}
By the uniqueness of $c_1$ and $c_2$,
\begin{align*}
\mathcal{M} \llbracket \exists_{\bf closr} x. (E' \land \exists_{\bf closr} y. IOMatch_{\bf closr} \ x \ y) \rrbracket (\alpha') & = \mathcal{M} \llbracket IOMatch_{\bf closr} \ x \ y) \rrbracket (\alpha' \cup [x \mapsto c_1, y \mapsto c_2]) \\
& = \mathcal{M} \llbracket IOMatch_{\bf closr} \ x \ y) \rrbracket ([x \mapsto c_1, y \mapsto c_2]).
\end{align*}
The only free variables in $IOMatch_{\bf closr} \ x \ y$ are $x$ and $y$. Hence, $\alpha'$ does not affect its semantics; thus, the second equality follows. The above expression can be further reduced to
\begin{align*}
\mathcal{M} \llbracket IOMatch_{\bf closr} \ x \ y) \rrbracket ([x \mapsto c_1, y \mapsto c_2]) & = iomatch_{\bf closr} \ c_1 \ c_2 \\
& = \text{expand}_{\alpha}(c_1) \ \text{expand}_{\alpha}(c_2) \\
& = \mathcal{M} \llbracket E \rrbracket (\alpha) \ \mathcal{M} \llbracket F \rrbracket (\alpha) \\
& = \mathcal{M} \llbracket E \ F \rrbracket (\alpha). 
\end{align*}
Therefore, the claim holds. This concludes the proof. 
\end{proof}

\begin{theorem} \label{alpha' is a model of P'}
If $\alpha$ is a model of $P$, then $\alpha' = T_{f}(\alpha)$ is a model for $P'$.
\end{theorem}
\begin{proof}
Assume that $\alpha$ is a model of $P$; that is, $\alpha$ is a prefixed point of $T^{\mathcal{M}}_{P: \Delta}$. It is given that
\begin{equation}
\mathcal{M} \llbracket \Delta \vdash P(X): \Delta(X) \rrbracket (\alpha) \subseteq_{\Delta(X)} \alpha (X) \label{P(X) is monotone in the first direction}
\end{equation}
for each $X \in {\tt dom}(\Delta)$. In addition, since $\alpha \in \mathcal{M} \llbracket \Delta \rrbracket$, we have ${\tt dom}(\alpha) = {\tt dom}(\Delta)$. Suppose $P(X)$ is of the form
\begin{equation*}
\lambda x_1, \ldots, x_m. F,
\end{equation*}
where $\Delta \vdash X: \sigma_1 \to \cdots \to \sigma_m \to o$. $X$ gives rise to
\begin{equation*}
IOMatch_{\sigma_m'} = \lambda x, x_m. (\exists x_1, \ldots, x_{m-1}. x = C^{m-1}_{X} \ x_1 \ \cdots \ x_{m-1} \land F'),
\end{equation*}
where $\sigma_m \leadsto_{T} \sigma_m'$ and $F \leadsto F'$. The sort of $IOMatch_{\sigma_m'}$ is ${\bf closr} \to \sigma_m' \to o$. 

Now suppose that for some $c \in A_{\bf closr}'$ and $c_m \in A_{\sigma_m'}'$, we have
\begin{equation*}
\mathcal{M} \llbracket \exists x_1, \ldots, x_{m-1}. x = C^{m-1}_{X} \ x_1 \ \cdots \ x_{m-1} \land F' \rrbracket (\alpha' \cup [x \mapsto c, x_m \mapsto c_m]) = 1. 
\end{equation*}
This means there exists $c_i \in A_{\sigma_i'}'$ for each $1 \leq i < m$ such that 
\begin{gather}
c = (X, c_1, \ldots, c_{m - 1}) \\
\mathcal{M} \llbracket F' \rrbracket (\alpha' \cup \{ [x_i \mapsto c_i] \mid 1 \leq i \leq m \} ) = 1. \label{F' under alpha' and beta' is 1}
\end{gather}
Let $\beta'$ be the valuation $\{ (x_i, c_i) \mid 1 \leq i \leq m \}$. For simplicity, I write $\{ (x_i, c_i) \mid 1 \leq i \leq m \}$ for $\{ [x_i \mapsto c_i] \mid 1 \leq i \leq m \}$. Also, let $\beta$ be $\{ (x_i, \text{expand}_{\alpha}(c_i)) \mid 1 \leq i \leq m \}$. 
Because ${\tt FV}(F) \subseteq {\tt dom}(\alpha) \cup {\tt dom}(\beta)$, $\alpha \cup \beta$ is a valid valuation for $F$. Similarly, $\alpha' \cup \beta'$ is a valid valuation of $F'$. 

By Lemma~\ref{relationship between the semantics of source and target goal terms} and \eqref{F' under alpha' and beta' is 1},
\begin{align*}
\mathcal{M} \llbracket F \rrbracket (\alpha \cup \beta) & = \mathcal{M} \llbracket F' \rrbracket (\alpha' \cup \beta') \\
& = 1. 
\end{align*}
Because $\mathcal{M} \llbracket F \rrbracket (\alpha \cup \beta) = 1$ and $P(X) = F$, it follows from \eqref{P(X) is monotone in the first direction} that
\begin{equation*}
\mathcal{M} \llbracket X \ x_1 \ \cdots \ x_m \rrbracket (\alpha \cup \beta) = 1
\end{equation*}
Therefore, we obtain
\begin{align*}
\mathcal{M} \llbracket IOMatch_{\sigma_m'} \ x \ x_m \rrbracket (\alpha' \cup \beta') & = iomatch_{\sigma_m'} \ c \ c_m \\
& = \text{expand}_{\alpha}(c) \ \text{expand}_{\alpha}(c_m) \\
& = \text{expand}_{\alpha}((X, c_1, \ldots, c_{m - 1})) \ \text{expand}_{\alpha}(c_m) \\
& = \alpha(X) \ \text{expand}_{\alpha}(c_1) \ \cdots \ \text{expand}_{\alpha}(c_m) \\
& = \mathcal{M} \llbracket X \ x_1 \ \cdots \ x_m \rrbracket (\alpha \cup \beta) \\
& = 1. 
\end{align*}
Thus, for all $c \in A_{\bf closr}'$ and $c_m \in A_{\sigma_m'}'$, 
\begin{equation} \label{P'(IOMatch) is monotone in the first direction}
\begin{split}
& \mathcal{M} \llbracket \exists x_1, \ldots, x_{m-1}. x = C^{m-1}_{X} \ x_1 \ \cdots \ x_{m-1} \land F' \rrbracket (\alpha' \cup [x \mapsto c, x_m \mapsto c_m]) \\
\subseteq_{o} {} & \mathcal{M} \llbracket IOMatch_{\sigma_m'} \ x \ x_m \rrbracket (\alpha' \cup [x \mapsto c, x_m \mapsto c_m]). 
\end{split}
\end{equation}
Hence, we obtain
\begin{align*}
& \mathcal{M} \llbracket \lambda x, x_m. \exists x_1, \ldots, x_{m-1}. x = C^{m-1}_{X} \ x_1 \ \cdots \ x_{m-1} \land F' \rrbracket (\alpha') \\
\subseteq_{{\bf closr} \to \sigma_m' \to o} {} & \mathcal{M} \llbracket IOMatch_{\sigma_m'} \rrbracket (\alpha').
\end{align*}
Even if $IOMatch_{B}$ has multiple branches corresponding to different top-level relational variables from $\Delta$, the disjunction of the left hand side of \eqref{P'(IOMatch) is monotone in the first direction} for each $X$ contributing to 
$IOMatch_{B}$ is smaller than or equal to the right hand side of \eqref{P'(IOMatch) is monotone in the first direction}. It therefore follows that
\begin{equation*}
\mathcal{M} \llbracket P'(IOMatch_{B}) \rrbracket (\alpha') \subseteq_{{\bf closr} \to B \to o} \alpha' (IOMatch_{B}).
\end{equation*}

If $X = Apply_{B}$, by the definition of $apply_{B}$, 
\begin{gather*}
\mathcal{M} \llbracket P'(Apply_{B}) \rrbracket (\alpha') = \alpha' (Apply_{B}) \\
\therefore \mathcal{M} \llbracket P'(Apply_{B}) \rrbracket (\alpha') \subseteq_{{\bf closr} \to B \to {\bf closr} \to o} \alpha' (Apply_{B}).
\end{gather*}
Hence, for every $X \in {\tt dom}(\Delta')$,
\begin{equation*}
\mathcal{M} \llbracket P'(X) \rrbracket (\alpha') \subseteq_{\Delta'(X)} \alpha' (X).
\end{equation*}
As $\alpha'$ is a prefixed point of $T^{\mathcal{M}}_{P': \Delta'}$ (which is equivalent to $T^{\mathcal{S}}_{P': \Delta'}$), it is indeed a model of $P'$. This concludes the proof. 
\end{proof}

\completenessofthedefunctionalizationalgorithm*
\begin{proof}
Let $\alpha$ be a solution to $\mathcal{P}$ and $\alpha'$ be a valuation for $P'$ derived from $\alpha$. By Theorem~\ref{alpha' is a model of P'}, $\alpha'$ is a model of $P'$. Furthermore, since $G \leadsto G'$, it follows from 
Lemma~\ref{relationship between the semantics of source and target goal terms} that
\begin{equation*}
\mathcal{M} \llbracket G \rrbracket (\alpha) = \mathcal{M} \llbracket G' \rrbracket (\alpha').
\end{equation*}
Because $\alpha$ is a solution to $\mathcal{P}'$, $\mathcal{M} \llbracket G \rrbracket (\alpha) = 0$. Therefore, $\mathcal{M} \llbracket G' \rrbracket (\alpha') = 0$ as well. Hence, $\alpha'$ is a solution to $\mathcal{P}'$ under the standard semantics.
By Theorem~\ref{a monotone solution exists if a standard solution exists}, $\mathcal{P}'$ is solvable under the monotone semantics.  
\end{proof}

\subsection{Difficulties with valuation extraction in the second direction} \label{section on difficulties with valuation extracting in the second direction}

For the second direction of meaning preservation, I explain some difficulties in extracting solutions to $\mathcal{P}$ from solutions to
$\mathcal{P}'$ as we did for the first direction. Consider the example introduced in Subsection~\ref{demonstration of valuation extraction}. Suppose that a solution to $\mathcal{P}'$ is
\begin{equation*}
\alpha' = \{ IOMatch_{\bf nat} \mapsto iomatch_{\bf nat}, Apply_{\bf nat} \mapsto apply_{\bf nat}, Apply_{\bf closr} \mapsto apply_{\bf closr} \},
\end{equation*}
where $apply_{\bf nat}$ and $apply_{\bf closr}$ are defined (independently of $\alpha$) in Subsection~\ref{demonstration of valuation extraction}. The interpretation of $IOMatch_{\bf nat}$ is 
\begin{equation*}
iomatch_{\bf nat} = add' \cup twice',
\end{equation*}
where the functions $add': A'_{\bf closr} \to \mathbb{N} \to \mathbbm{2}$ and $twice': A'_{\bf closr} \to \mathbb{N} \to \mathbbm{2}$ are
\begin{equation*}
add' \ m \ n = 
\begin{cases}
1 & \text{if } m = (Add, n_1, n_2), n = n_1 + n_2 \\
0 & \text{otherwise}
\end{cases}
\end{equation*}
and
\begin{equation*}
twice' \ m \ n =
\begin{cases}
1 & \text{if } m = (Twice, f, n_1) \\
& {} \land \exists n_2. ((\exists n_3. apply_{\bf nat} \ f \ n_1 \ n_3 \land iomatch_{\bf nat} \ n_3 \ n_2) \\
& {} \land (\exists n_4. apply_{\bf nat} \ f \ n_2 \ n_4 \land iomatch_{\bf nat} \ n_4 \ n)) \\
0 & \text{otherwise}.
\end{cases}
\end{equation*}

There are three issues with extracting a valuation for $P$ from $\alpha'$. 
\begin{enumerate}
\item It is not straightforward to define a valuation for $P$ that has the same structure as $\alpha$. For instance, because the sort of $Add$ in $\mathcal{P}'$ has order 2 (i.e.~not a higher-order function), $add'$ can be straightforwardly transferred to 
the interpretation of $Add$, yielding $add: \mathbb{N} \to \mathbb{N} \to \mathbb{N} \to o$ given as
\begin{equation*}
add \ n_1 \ n_2 \ n_3 = 
\begin{cases}
1 & \text{if } n_3 = n_1 + n_2 \\
0 & \text{otherwise}. 
\end{cases}
\end{equation*}
On the other hand, it is not easy to extract an interpretation for $Twice$ from $twice'$. This is because $twice'$ is defined in terms of $iomatch_{\bf nat}$, which is in turn defined in terms of $twice'$ (and $add'$). 

Due to this recursive nature of the definition of $twice'$, it is not clear how to construct a valuation $\alpha$ for $P$ that satisfies $\alpha' = T_{f}(\alpha)$. It is crucial for $\alpha$ to have the
same structure as $\alpha'$ because it lets us apply Lemma~\ref{relationship between the semantics of source and target goal terms} to prove $\mathcal{M} \llbracket G \rrbracket (\alpha) = 0$.

\item Suppose that the first issue is overcome and that $\alpha$ that satisfies $\alpha' = T_{f}(\alpha)$ has been obtained. With the same example as above, it is natural to have 
\begin{equation*}
\alpha(Twice) \ f \ n_1 \ n_2 = 0
\end{equation*}
whenever $f$ is not expressible in $\mathcal{P}$; that is, whenever $f$ cannot be expressed by combination of $Twice$ and $Add$. This creates an issue that $\alpha$ is not monotone. For example, 
\begin{equation*}
(\lambda a, b. a + 1 = b) \subseteq_{\rho} U_{\rho}
\end{equation*}
but
\begin{equation}
\alpha(Twice) \ (\lambda a, b. a + 1 = b) \ 2 \ 4 \nsubseteq_{o} \alpha(Twice) \ U_{\rho} \ 2 \ 4. \label{model_Twice is not monotone}
\end{equation}
As $\lambda a, b. a + 1 = b$ can be expressed by $Add \ 1$, the left hand side of \eqref{model_Twice is not monotone} evaluates to 1. However, since $U_{\rho}$ cannot be expressed by any element of $A_{\bf closr}'$, 
the right hand side of \eqref{model_Twice is not monotone} evaluates to 0. Thus, $\alpha(Twice)$ is not monotone; hence, neither is $\alpha$.

\item The third problem with $\alpha$ is that it is not necessarily a prefixed point of $T^{\mathcal{M}}_{P: \Delta}$. In the above example, 
\begin{equation*}
\mathcal{M} \llbracket \lambda f, a, b. (\exists c. f \ a \ c \land f \ c \ b) \rrbracket (\alpha) \nsubseteq_{\Delta(Twice)} \alpha (Twice)
\end{equation*}
holds since the left hand side can take $f = U_{\rho}$ and produces 1 for any $a$ and $b$, whilst the right hand side does not. 
\end{enumerate}

\subsection{Continuity of one-step consequence operators}

\finitedomainsimplycontinuity*
\begin{proof}
I will prove that for every directed subset $R \subseteq \mathcal{M} \llbracket \Delta \rrbracket$, $\bigsqcup \{ T^{\mathcal{M}}_{P: \Delta} (x) \mid x \in R \}$ exists and equals $T^{\mathcal{M}}_{P: \Delta} (\bigsqcup R)$. 

Fix $R \subseteq \mathcal{M} \llbracket \Delta \rrbracket$. Because $T^{\mathcal{M}}_{P: \Delta}$ is monotone, it is given
\begin{gather*}
\forall x \in R. x \subseteq \bigsqcup R \\
\therefore \forall x \in R. T^{\mathcal{M}}_{P: \Delta} (x) \subseteq T^{\mathcal{M}}_{P: \Delta} (\bigsqcup R). 
\end{gather*}
Note that the order of valuations is denoted by $\subseteq$ rather than $\leq$. Thus, 
\begin{equation}
\bigsqcup \{ T^{\mathcal{M}}_{P: \Delta} (x) \mid x \in R \} \subseteq T^{\mathcal{M}}_{P: \Delta} (\bigsqcup R), \label{first inequality between two sides in the proof for continuous of the one-step consequence operator}
\end{equation}
where the left hand side exists as $\mathcal{M} \llbracket \Delta \rrbracket$ is a complete lattice.

It remains to prove that both sides of the above inequality are in fact equal. If $\mathcal{M} \llbracket \Delta \rrbracket$ is finite, then $R$ must be finite as well. Since $R$ is directed by assumption and is finite, $\bigsqcup R \in R$. Hence,
\begin{equation}
T^{\mathcal{M}}_{P: \Delta} (\bigsqcup R) \subseteq \bigsqcup \{ T^{\mathcal{M}}_{P: \Delta} (x) \mid x \in R \}. \label{second inequality between two sides in the proof for continuous of the one-step consequence operator}
\end{equation}
Combining \eqref{first inequality between two sides in the proof for continuous of the one-step consequence operator} and \eqref{second inequality between two sides in the proof for continuous of the one-step consequence operator}, we obtain
\begin{equation*}
\bigsqcup \{ T^{\mathcal{M}}_{P: \Delta} (x) \mid x \in R \} = T^{\mathcal{M}}_{P: \Delta} (\bigsqcup R).
\end{equation*}
Therefore, $T^{\mathcal{M}}_{P: \Delta}$ is indeed continuous if $\mathcal{M} \llbracket \Delta \rrbracket$ is finite. 
\end{proof}

\iterativeapplicationofacontinuousfunctiononthebottomyieldstheleastfixedpoint*
\begin{proof}
Since $f$ is continuous, it is monotone. Therefore, $\langle f^{n} (\bot) \mid n \in \mathbb{N} \rangle$ is an increasing sequence. As $f$ is continuous, 
\begin{equation*}
f (\bigsqcup \{ f^{n} (\bot) \mid n \in \mathbb{N} \}) = \bigsqcup \{ f^{n+1} (\bot) \mid n \in \mathbb{N} \}.
\end{equation*}
Because $\bot$ cannot be larger than any element from $\{ f^{n+1} (\bot) \mid n \in \mathbb{N} \}$, 
\begin{equation*}
\bigsqcup \{ f^{n+1} (\bot) \mid n \in \mathbb{N} \} = \bigsqcup \{ f^{n} (\bot) \mid n \in \mathbb{N} \}.
\end{equation*}
Combining the above two equations gives
\begin{equation*}
f (\bigsqcup \{ f^{n} (\bot) \mid n \in \mathbb{N} \}) = \bigsqcup \{ f^{n} (\bot) \mid n \in \mathbb{N} \}.
\end{equation*}
Therefore, $\bigsqcup \{ f^{n} (\bot) \mid n \in \mathbb{N} \}$ is a fixed point of $f$. 

If $y$ is also a least fixed point of $f$, we have $\bot \leq y$. The monotonicity of $f$ gives that
\begin{equation*}
\forall n \in \mathbb{N}. f^{n} (\bot) \leq y,
\end{equation*}
which means that $y$ is also an upper bound of $\{ f^{n} (\bot) \mid n \in \mathbb{N} \}$. As $\bigsqcup \{ f^{n} (\bot) \mid n \in \mathbb{N} \}$ is the least upper bound of $\{ f^{n} (\bot) \mid n \in \mathbb{N} \}$ by definition, 
$\bigsqcup \{ f^{n} (\bot) \mid n \in \mathbb{N} \} \leq y$. Thus, $\bigsqcup \{ f^{n} (\bot) \mid n \in \mathbb{N} \}$ is the least fixed point of $f$. 
\end{proof}

\subsection{Second direction}

\relationshipbetweensourceandtargetonestepconsequenceoperators*
\begin{proof}
The claim can be depicted by the following commutative diagram:
\begin{displaymath}
\xymatrix{
\zeta \ar[r]^-{T_{f}} & \zeta' \\
\gamma \ar[u]^{T^{\mathcal{M}}_{P: \Delta}} \ar[r]_-{T_{f}} & \gamma' \ar[u]_{T^{\mathcal{M}}_{P': \Delta'}}
}
\end{displaymath}
Fix $X \in {\tt dom}(\Delta)$ and assume $\Delta \vdash X: \sigma_1 \to \cdots \to \sigma_m \to o$. Also, suppose that $P$ contains $X = \lambda x_1, \ldots, x_m. F$. In addition, for each $1 \leq i \leq m$, consider
$c_i \in \mathcal{M} \llbracket \sigma_i \rrbracket$ such that there exists $c_i' \in A_{\sigma_i'}'$ that satisfies $\text{expand}_{\gamma}(c_i') = c_i$. Here, $\sigma_i \leadsto_{T} \sigma_i'$ for each $i$. 

Now let $\eta$ be the valuation $\{ (x_i, c_i) \mid 1 \leq i \leq m \}$ and $\eta'$ be $\{ (x_i, c_i') \mid 1 \leq i \leq m \}$. Because all free variables, including top-level relational variables, in $F$ are in the domains of $\gamma$ and $\eta$, 
$\gamma \cup \eta$ is a valid valuation of $F$. Similarly, $\gamma' \cup \eta'$ is a valid valuation of $F'$, where $F \leadsto F'$. Since $F$ does not contain lambda abstractions, we can apply 
Lemma~\ref{relationship between the semantics of source and target goal terms} to obtain
\begin{equation}
\mathcal{M} \llbracket F \rrbracket (\gamma \cup \eta) = \mathcal{M} \llbracket F' \rrbracket (\gamma' \cup \eta'). \label{equality between the semantics of F under gamma cup eta and F' under gamma' cup eta'}
\end{equation}
The left hand side of \eqref{equality between the semantics of F under gamma cup eta and F' under gamma' cup eta'} is equal to
\begin{equation} \label{transformation of the left hand side of the first equality in the proof for the relationship between source and target one-step consequence operators}
\begin{split}
\mathcal{M} \llbracket F \rrbracket (\gamma \cup \eta) & = \mathcal{M} \llbracket \lambda x_1, \ldots, x_m. F \rrbracket (\gamma) \ c_1 \ \cdots \ c_m \\
& = \mathcal{M} \llbracket P(X) \rrbracket (\gamma) \ c_1 \ \cdots \ c_m \\
& = T^{\mathcal{M}}_{P: \Delta} (\gamma) (X) \ c_1 \ \cdots \ c_m \\
& = \zeta (X) \ c_1 \ \cdots \ c_m,
\end{split}
\end{equation}
where the third equality follows from the definition of $T^{\mathcal{M}}_{P: \Delta}$ and the last equality follows from the definition of $\zeta$. The right hand side of 
\eqref{equality between the semantics of F under gamma cup eta and F' under gamma' cup eta'} can be transformed into
\begin{equation} \label{transformation of the right hand side of the first equality in the proof for the relationship between source and target one-step consequence operators}
\begin{split}
& \mathcal{M} \llbracket F' \rrbracket (\gamma' \cup \eta') \\
= {} & \mathcal{M} \llbracket \lambda x, x_m. \exists x_1, \ldots, x_{m-1}. x = C^{m-1}_{X} \ x_1 \ \cdots \ x_{m-1} \land F' \rrbracket (\gamma') \ (X, c_1', \ldots, c_{m-1}') \ c_{m}' \\
= {} & \mathcal{M} \llbracket P'(IOMatch_{\sigma_m'}) \rrbracket (\gamma') \ (X, c_1', \ldots, c_{m-1}') \ c_{m}' \\
= {} & T^{\mathcal{M}}_{P': \Delta'} (\gamma') (IOMatch_{\sigma_m'}) \ (X, c_1', \ldots, c_{m-1}') \ c_{m}' \\
= {} & \zeta' (IOMatch_{\sigma_m'}) \ (X, c_1', \ldots, c_{m-1}') \ c_{m}'. 
\end{split}
\end{equation}
Combining \eqref{equality between the semantics of F under gamma cup eta and F' under gamma' cup eta'}, \eqref{transformation of the left hand side of the first equality in the proof for the relationship between source and target one-step consequence operators}, and \eqref{transformation of the right hand side of the first equality in the proof for the relationship between source and target one-step consequence operators},
we obtain
\begin{equation*}
\zeta (X) \ c_1 \ \cdots \ c_m = \zeta' (IOMatch_{\sigma_m'}) \ (X, c_1', \ldots, c_{m-1}') \ c_{m}'.
\end{equation*}
Therefore, it follows from the definition of $T_{f}$ that $\zeta' = T_{f} (\zeta)$. This concludes the proof.
\end{proof}

\betaandbetaarerelatedbyTf*
\begin{proof}
As usual, fix $X \in {\tt dom}(\Delta)$ and assume $\Delta \vdash X: \sigma_1 \to \cdots \to \sigma_m \to o$. Also, suppose that $P$ contains $X = \lambda x_1, \ldots, x_m. F$. In addition, for each $1 \leq i \leq m$, consider
$c_i \in \mathcal{M} \llbracket \sigma_i \rrbracket$ such that there exists $c_i' \in A_{\sigma_i'}'$ that satisfies $\text{expand}_{\alpha}(c_i') = c_i$. Here, $\sigma_i \leadsto_{T} \sigma_i'$ for each $i$. 

Since $\beta$ and $\beta'$ are the least upper bounds of $\{ f^{n}_{1} (\alpha) \mid n \in \mathbb{N} \}$ and $\{ f^{n}_{2} (\alpha') \mid n \in \mathbb{N} \}$, respectively, it is given that
\begin{equation}
\beta (X) \ c_1 \ \cdots \ c_m = \max \{ f^{n}_{1} (\alpha) (X) \ c_1 \ \cdots \ c_m \mid n \in \mathbb{N} \} \label{equality between beta and maximum of an increasing sequence of alpha}
\end{equation}
and
\begin{equation} \label{equality between beta' and maximum of an increasing sequence of alpha'}
\begin{split}
& \beta' (IOMatch_{\sigma_m'}) \ (X, c_1', \ldots, c_{m-1}') \ c_m' \\
= {} & \max \{ f^{n}_{2} (\alpha') (IOMatch_{\sigma_m'}) \ (X, c_1', \ldots, c_{m-1}') \ c_m' \mid n \in \mathbb{N} \}. 
\end{split}
\end{equation}
As $\alpha' = T_{f} (\alpha)$, by Lemma~\ref{relationship between source and target one-step consequence operators}, $f^{n}_{2}(\alpha') = T_{f} (f^{n}_{1} (\alpha))$ for every $n \in \mathbb{N}$. Hence,
\begin{equation*}
f^{n}_{1} (\alpha) (X) \ c_1 \ \cdots \ c_m = f^{n}_{2} (\alpha') (IOMatch_{\sigma_m'}) \ (X, c_1', \ldots, c_{m-1}') \ c_m'
\end{equation*}
for each $n$. Consequently,
\begin{equation}
\max \{ f^{n}_{1} (\alpha) (X) \ c_1 \ \cdots \ c_m \} = \max \{ f^{n}_{2} (\alpha') (IOMatch_{\sigma_m'}) \ (X, c_1', \ldots, c_{m-1}') \ c_m' \}. \label{equality between the maxima of two increasing sequences}
\end{equation}
Combining \eqref{equality between beta and maximum of an increasing sequence of alpha}, \eqref{equality between beta' and maximum of an increasing sequence of alpha'}, and \eqref{equality between the maxima of two increasing sequences} yields
\begin{equation*}
\beta (X) \ c_1 \ \cdots \ c_m = \beta' (IOMatch_{\sigma_m'}) \ (X, c_1', \ldots, c_{m-1}') \ c_m'.
\end{equation*}
Therefore, $\beta' = T_{f} (\beta)$ indeed holds. 
\end{proof}

\seconddirectionofthemeaningpreservationunderthemonotonesemantics*
\begin{proof}
Let $\bot$ be the least element from $\mathcal{M} \llbracket \Delta \rrbracket$ and $\bot'$ be the least element from $\mathcal{M} \llbracket \Delta' \rrbracket$. It is clear that $\bot' = T_{f} (\bot)$. 

Suppose $\beta = \bigsqcup \{ f^{n}_{1} (\bot) \mid n \in \mathbb{N} \}$, where $f_{1} = T^{\mathcal{M}}_{P: \Delta}$, and $\beta' = \bigsqcup \{ f^{n}_{2} (\bot') \mid n \in \mathbb{N} \}$, where $f_{2} = T^{\mathcal{M}}_{P': \Delta'}$. 
Because it is given that $T^{\mathcal{M}}_{P: \Delta}$ and $T^{\mathcal{M}}_{P': \Delta'}$ are continuous, $\beta$ and $\beta'$ are fixed points of the respective one-step consequence operators. In other words, they are models of 
$P$ and $P'$, respectively.

Furthermore, because $\beta$ is obtained by iteratively applying $T^{\mathcal{M}}_{P': \Delta'}$ to $\bot'$, which is the least element of $\mathcal{M} \llbracket \Delta' \rrbracket$, $\beta'$ is the least fixed point of 
$T^{\mathcal{M}}_{P': \Delta}$ by Theorem~\ref{iterative application of a continuous function on the bottom yields the least fixed point}. Moreover, it is the least prefixed point. This statement is not too difficult to prove, although I will not provide
its formal proof. 

Assume that $\mathcal{P}'$ is solvable and let its solution be $\alpha'$. Then $\beta' \subseteq \alpha'$ because $\beta'$ is the least model of $P$. Moreover, by the monotonicity of $\mathcal{M} \llbracket G \rrbracket$, we should have
\begin{equation*}
\mathcal{M} \llbracket G \rrbracket (\beta') \subseteq_{o} \mathcal{M} \llbracket G \rrbracket (\alpha').
\end{equation*}
The right hand side of this equation is 0 since $\alpha'$ is a solution to $\mathcal{P}'$. Thus, $\mathcal{M} \llbracket G \rrbracket (\beta') = 0$. 

It follows from Lemma~\ref{beta and beta' are related by T-f} that $\beta' = T_{f} (\beta)$. Hence, by Lemma~\ref{relationship between the semantics of source and target goal terms}, we have
\begin{align*}
\mathcal{M} \llbracket G \rrbracket (\beta) & = \mathcal{M} \llbracket G \rrbracket (\beta') \\
& = 0.
\end{align*}
Therefore, $\beta$ is a solution to $\mathcal{P}$. This concludes the proof. 
\end{proof}


\section{Type preservation} \label{chapter on type preservation}

Before I prove type preservation, I revisit the basic concepts of first-order terms from constraint languages and goal terms. 

\subsection{Defining terms and formulas}

In this subsection, I formally define first-order terms and first-order formulas in constraint languages. This is necessary because to prove type preservation, I need to use some properties of terms. 

\subsubsection{Terms and formulas}
Let $\Sigma = (\mathbb{B}, \mathbb{S})$ be a first-order signature. Since $\Sigma$ is first-order, the sort of each symbol in $\mathbb{S}$ has order at most 2. The class of well-sorted first-order terms over $\Sigma$ is given by
\begin{center}
\AxiomC{}
\LeftLabel{\sc (TCst)}
\RightLabel{$c \in {\tt dom}(\mathbb{S}) $}
\UnaryInfC{$\Delta \vdash c: \mathbb{S}(c)$}
\DisplayProof
\qquad
\AxiomC{}
\LeftLabel{\sc (TAnd)}
\UnaryInfC{$\Delta \vdash \land: o \to o\to o$}
\DisplayProof
\end{center}
\begin{center}
\AxiomC{}
\LeftLabel{\sc (TNeg)}
\UnaryInfC{$\Delta \vdash \neg : o \to o$}
\DisplayProof
\qquad
\AxiomC{}
\LeftLabel{\sc (TVar)}
\UnaryInfC{$\Delta_1, x: b, \Delta_2 \vdash x: b$}
\DisplayProof
\end{center}
\begin{center}
\AxiomC{$\Delta, x: b \vdash t: o$}
\LeftLabel{\sc (TExi)}
\UnaryInfC{$\Delta \vdash \exists_{b} x. t: o$}
\DisplayProof
\qquad
\AxiomC{$\Delta \vdash t_1: b \to \beta$}
\AxiomC{$\Delta \vdash t_2: b$}
\LeftLabel{\sc (TApp)}
\BinaryInfC{$\Delta \vdash t_1 \ t_2: \beta$}
\DisplayProof
\end{center}
Here, $b$ is a base sort $t$ (with or without subscripts) is a first-order term, and $\beta$ is a sort of order at most 2; i.e.~sort of the form $b_1 \to \cdots \to b_n$, where $b_i \in \mathbb{B}$ for each $1 \leq i \leq n$.

It is important that $\Delta$ contains no conflicts; i.e.~no variable is associated with multiple types. Henceforth, it is implicitly assumed that sort environments for first-order terms are free of conflicts. 

Well-sorted first-order formulas are defined as well-sorted first-order terms of sort $o$. Notice that unlike in usual presentation of first-order logic, where formulas and terms are disjoint, according to the above definition, terms include formulas. 

When a first-order term $s$ is well-sorted under sort environment $\Delta$ and has sort $\beta$, I write $\Delta \vdash s: \beta$. From now on, I assume that first-order terms are well-sorted. 

When a typing judgement $\Delta \vdash s: \beta$ is created by {\sc (TCst)}, {\sc (TAnd)}, or {\sc (TNeg)}, the sort of $s$ is independent of $\Delta$. In that case, to work out the sort of $s$, we need to check $\mathbb{S}$ and {\tt LSym}. When $\mathbb{S}$
is unclear, I write $\mathbb{S}, \Delta \vdash s: \beta$ to make $\mathbb{S}$ explicit. However, whenever $\mathbb{S}$ is clear from the context, I will omit it from sort environments. 

\subsubsection{Properties of terms and formulas}

\begin{proposition} \label{free variables in first-order terms}
Every free variable occurring in a first-order term has a base sort.
\end{proposition}
\begin{proof}
Variables can only be introduced into first-order terms by {\sc (TVar)}. The rule requires variables to be of base sort. Hence, the claim is true. 
\end{proof}

\begin{proposition} \label{uniqueness of sorts of first-order terms}
Given $\Delta \vdash s: \beta$, the sort of $s$ under $\Delta$ is unique; that is, we cannot have $\Delta \vdash s: \beta'$, where $\beta \neq \beta'$.
\end{proposition}
\begin{proof}
The proof goes by structural induction on $s$. 

For the base case, if $\Delta \vdash s: \beta$ is created by {\sc (TCst)}, {\sc (TAnd)}, or {\sc (TNeg)}, the sort of $s$ is unique (and is independent of $\Delta$). If $\Delta \vdash s: \beta$ is created by {\sc (TVar)}, the sort of $s$ is uniquely determined by
$\Delta$. 

For the inductive case, if $\Delta \vdash s: \beta$ is created by {\sc (TExi)}, we have $\beta = o$. Thus, the sort of $s$ is uniquely determined. 

Finally, if $\Delta \vdash s: \beta$ is generated by {\sc (TApp)}, we know that $s = t_1 \ t_2$. By the inductive 
hypothesis, the sorts of $t_1$ and $t_2$ under $\Delta$ are uniquely determined. Therefore, the sort of $s$ under $\Delta$ is also uniquely determined. 
\end{proof}

The following proposition states that each well-sorted first-order term has a unique way to assign sorts to all symbols occurring in the term such that the term is well-sorted.

\begin{theorem} \label{type annotation of first-order terms}
If $\Delta \vdash s: \beta$ holds, where $s$ is a first-order term, then every symbol occurring in $s$ can
be annotated with a unique sort. 
\end{theorem}
\begin{proof}
The claim is proved by structural induction on $s$. 

For the base case, if $\Delta \vdash s: \beta$ is created by {\sc (TCst)}, {\sc (TAnd)}, or {\sc (TNeg)}, the sort of $s$ is given by $\mathbb{S}$ or {\tt LSym} and is unique. If $s$ is created by {\sc (TVar)}, the sort of $s$ is given by 
$\Delta(s)$ and is unique because $\Delta$ is assumed to contain no conflicts. Thus, in all base cases, the sort of $s$ can be uniquely identified. Alternatively, we can use Proposition~\ref{uniqueness of sorts of first-order terms} to
prove the base case. As the only symbol appearing in $s$ is itself, the claim reduces to Proposition~\ref{uniqueness of sorts of first-order terms}. 

For the inductive case, if $s$ is created by {\sc (TExi)}, $s$ is in the form of $\exists_{b} x. t$, where $\Delta, x: b \vdash t: o$. The sort of $x$ is stored in the subscript of $\exists_{b}$ in $s$. Hence, from $\Delta \vdash \exists_{b} x. t: o$,
we can uniquely derive $\Delta, x: b \vdash t: o$. In other words, from a conclusion of {\sc (TExi)}, we can uniquely deduce the corresponding premise of {\sc (TExi)}. By the inductive hypothesis, every symbol in $\Delta, x: b \vdash t: o$ can be annotated
with a unique sort. If there exist two distinct ways to assign sorts to the symbols occurring in $\Delta \vdash \exists_{b} x. t: o$, there should be two distinct ways to assign sorts to $\Delta, x: b \vdash t: o$ as well, which contradicts the
inductive hypothesis. Hence, all symbols in $\exists_{b} x. t$ can be annotated with a unique sort.

Finally, if $s$ is created by {\sc (TApp)}, we have $s = t_1 \ t_2$, where $\Delta \vdash t_1: b \to \beta$ and $\Delta \vdash t_2: b$. We cannot determine the typing judgements $\Delta \vdash t_1: b \to \beta$ and $\Delta \vdash t_2: b$ uniquely
by the mere appearance of $\Delta \vdash t_1 \ t_2: \beta$, without any calculation. However, we can evaluate the sorts of $t_1$ and $t_2$ under the sort environment $\Delta$ by repeatedly applying the six typing rules listed above. Furthermore, by 
Proposition~\ref{uniqueness of sorts of first-order terms}, the sorts of $t_1$ and $t_2$ under $\Delta$ are unique. By the inductive hypothesis, every symbol in $t_1$ and $t_2$ can be annotated with a unique symbol. 
For the sake of contradiction, assume that there are two distinct ways to assign sorts to $t_1 \ t_2$. Then at least one of $\Delta \vdash t_1$ and $\Delta \vdash t_2$ has two distinct sort assignments. This contradicts the inductive
hypothesis. Therefore, the claim holds for $t_1 \ t_2$ as well. This concludes the proof. 
\end{proof}

In effect, Theorem~\ref{type annotation of first-order terms} proves uniqueness of typing derivation trees of first-order terms by showing that given the root of a derivation tree, the root's successor(s) can be uniquely determined. Because all constants
and variables appear at the leaves of a tree, their sort assignment is uniquely determined. As for logical constants, their sorts are given by {\tt LSym} and hence are unique. 

The syntax and typing rules of first-order terms allow us to determine the sort of each symbol in a term by simply consulting $\mathbb{S}$, {\tt LSym}, and $\Delta$. This nice property does not hold any longer if we omit subscripts from
$\exists$. For instance, consider $\vdash (\exists x. x= 2): o$. It is still possible to uniquely determine the sort of each symbol. However, we cannot apply the same proof as the one for Theorem~\ref{type annotation of first-order terms}, since
it is not straightforward to deduce the typing judgement $x: {\tt int} \vdash (x= 2): o$ (especially the left hand side of the judgement; i.e.~$x: {\tt int}$) from $\vdash (\exists x. x= 2): o$. To determine the sort of $x$, we need to carry out
type inference using $\vdash (=): {\tt int} \to {\tt int} \to o$.

\subsection{Redefining goal terms}

In this subsection, I redefine goal terms in order to fix my imprecise use of terminology. In my explanation of the defunctionalization algorithm (Section~\ref{section on the defunctionalization algorithm}), I call an input of transformation a `source goal term' and an
output a `target goal term'. A problem lies in the use of the word `goal term'. According to \cite{Ramsay2017}, elements of $Tm$, where $Tm$ is a set of first-order terms in a constraint language, do not qualify as goal terms. 
However, in my presentation of the defunctionalization algorithm, a `source goal term' can be an element from $Tm$. This issue is caused by the fact that although $t \in Tm$ can be a subexpression of a goal term, $t$ itself is not a goal
term. Hence, I need to find a suitable word to refer to a collection of both goal terms and terms from $Tm$. One solution I would suggest is to redefine goal terms to mean first-order terms from $Tm$ as well as goal terms (in the original definition). 

The next subsection is a revised version of Subsection~\ref{subsection on goal terms}. I will also introduce some useful theorems about goal terms. 

\subsubsection{Goal terms}

Fix a first-order signature $\Sigma = (\mathbb{B}, \mathbb{S})$ and a constraint language $(Tm, Fm, Th)$ over $\Sigma$. In the original paper \cite{Ramsay2017}, the class of well-sorted goal terms $\Delta\vdash G: \rho$, 
where $\rho$ is a relational sort, is given by these sorting rules:
\begin{prooftree}
\AxiomC{}
\LeftLabel{\sc (GCst)}
\RightLabel{$c \in \{\land, \lor, \exists_{b} \} \cup \{ \exists_{\rho} \mid \rho \}$}
\UnaryInfC{$\Delta \vdash c: \rho_{c}$}
\end{prooftree}
\begin{prooftree}
\AxiomC{}
\LeftLabel{\sc (GVar)}
\UnaryInfC{$\Delta_1, x: \rho, \Delta_2 \vdash x: \rho$}
\end{prooftree}
\begin{prooftree}
\AxiomC{}
\LeftLabel{\sc (GConstr)}
\RightLabel{$\Delta \vdash \varphi: o \in Fm$}
\UnaryInfC{$\Delta \vdash \varphi: o$}
\end{prooftree}
\begin{prooftree}
\AxiomC{$\Delta, x: \sigma \vdash G: \rho$}
\LeftLabel{\sc (GAbs)}
\RightLabel{$x \notin \text{\tt dom}(\Delta)$}
\UnaryInfC{$\Delta \vdash \lambda x {:} \sigma. G: \sigma \to \rho$}
\end{prooftree}
\begin{prooftree}
\AxiomC{$\Delta \vdash G: b \to \rho$}
\LeftLabel{\sc (GAppl)}
\RightLabel{$\Delta \vdash N: b \in Tm$}
\UnaryInfC{$\Delta \vdash G \ N: \rho$}
\end{prooftree}
\begin{prooftree}
\AxiomC{$\Delta \vdash G: \rho_1 \to \rho_2$}
\AxiomC{$\Delta \vdash H: \rho_1$}
\LeftLabel{\sc (GAppR)}
\BinaryInfC{$\Delta \vdash G \ H: \rho_2$}
\end{prooftree}
Throughout the above six rules, $b$ denotes a base sort from $\mathbb{B}$, $\rho$ (with or without subscripts) denotes a relational sort, and $\sigma$ is either a base sort or a relational sort.

Despite being a subexpression of a goal term, a first-order term $t \in Tm$ is not a goal term according to the definition above. As the defunctionalization algorithm I developed works compositionally, I need a word to
refer to not only goal terms but also their subexpressions (excluding subexpressions of elements from $Tm \cup Fm$). Therefore, I will redefine goal terms to encompass first-order terms from $Tm$:
\begin{prooftree}
\AxiomC{}
\LeftLabel{\sc (GCst)}
\RightLabel{$c \in \{\land, \lor, \exists_{b} \} \cup \{ \exists_{\rho} \mid \rho \}$}
\UnaryInfC{$\Delta \vdash c: \rho_{c}$}
\end{prooftree}
\begin{prooftree}
\AxiomC{}
\LeftLabel{\sc (GVar)}
\UnaryInfC{$\Delta_1, x: \rho, \Delta_2 \vdash x: \rho$}
\end{prooftree}
\begin{center}
\AxiomC{}
\LeftLabel{\sc (GFml)}
\RightLabel{$\Delta \vdash \varphi: o \in Fm$}
\UnaryInfC{$\Delta \vdash \varphi: o$}
\DisplayProof
\qquad
\AxiomC{}
\LeftLabel{\sc (GTerm)}
\RightLabel{$\Delta \vdash t: b \in Tm$}
\UnaryInfC{$\Delta \vdash t: b$}
\DisplayProof
\end{center}
\begin{prooftree}
\AxiomC{$\Delta, x: \sigma \vdash G: \rho$}
\LeftLabel{\sc (GAbs)}
\RightLabel{$x \notin \text{\tt dom}(\Delta)$}
\UnaryInfC{$\Delta \vdash \lambda x {:} \sigma. G: \sigma \to \rho$}
\end{prooftree}
\begin{prooftree}
\AxiomC{$\Delta \vdash G: \sigma \to \rho$}
\AxiomC{$\Delta \vdash H: \sigma$}
\LeftLabel{\sc (GApp)}
\BinaryInfC{$\Delta \vdash G \ H: \rho$}
\end{prooftree}
As before, throughout the new six rules, $b$ denotes a base sort from $\mathbb{B}$, $\rho$ denotes a relational sort, and $\sigma$ is either a base sort or a relational sort. 

As is true of first-order terms, it is important that $\Delta$ contains no conflicts; i.e.~no variable is associated with multiple types. Henceforth, it is implicitly assumed that sort environments for goal terms are free of conflicts. 

When a goal term $t$ is well-sorted under the sort environment $\Delta$ and has sort $\sigma$, I write $\Delta \vdash t: \sigma$.

The next three propositions establish the relationship between the original and modified definitions of goal terms. 

\begin{proposition} \label{goal terms in the original def can be generated by the new one}
If $s$ is a goal term in the original definition, $s$ can be generated by the new definition. Further, if $\Delta \vdash s: \rho$ in the original definition, then $\Delta \vdash s: \rho$ holds in the new definition as well.
\end{proposition}
\begin{proof}
The claim is proved by structural induction on $s$. 

For the base case, if $s$ is generated by {\sc (GCst)} or {\sc (GVar)}, $s$ can be generated by the same rules in the new definition. If $s$ is generated by {\sc (GConstr)} in the 
original definition, it can be generated by {\sc (GFml)} in the new definition. In both cases, the sort is preserved. 

For the inductive case, suppose that $s$ is generated by {\sc (GAbs)} in the original definition. Then it follows from (GAbs) that $s$ is in the form
\begin{equation*}
s = \lambda x. G,
\end{equation*}
where $G$ is a goal term in the original definition. Also, if $\Delta, x: \sigma \vdash G: \rho$, then we have $\Delta \vdash \lambda x. G: \sigma \to \rho$. By the inductive hypothesis, $G$ can be generated by the new
definition, and $\Delta, x: \sigma \vdash G: \rho$ holds. Hence, by {\sc (GAbs)} in the new definition, $\Delta \vdash \lambda x. G: \sigma \to \rho$ can be established. Thus, the claim is true in this case.

Consider the case when $s$ is generated by {\sc (GAppl)} in the original definition. From {\sc (GAppl)}, we know that $s = G \ N$, where $G$ is a goal term and $N \in Tm$. Furthermore, if $\Delta \vdash G: b \to \rho$, then $\Delta \vdash G \ N: \rho$
holds. By the inductive hypothesis, $\Delta \vdash G: b \to \rho$ can be established by the new definition. Also, $\Delta \vdash N: b$ holds in the new definition. Therefore, by {\sc (GApp)} in the new definition, we obtain
\begin{prooftree}
\AxiomC{$\Delta \vdash G: b \to \rho$}
\AxiomC{$\Delta \vdash N: b$}
\LeftLabel{\sc (GApp)}
\BinaryInfC{$\Delta \vdash G \ N: \rho$}
\end{prooftree}
Thus, the claim is true when $s$ is generated by {\sc (GAppl)}.

The case when $s$ is generated by {\sc (GAppR)} in the original definition can be proved in the same manner as the case when $s$ is created by {\sc (GAbs)}. 
\end{proof}

\begin{proposition} \label{relational goal terms in the new def can be generated by the original one}
If $\Delta \vdash s: \rho$ in the new definition, where $\rho$ is a relational sort, the typing judgement holds in the old definition as well. 
\end{proposition}
\begin{proof}
By structural induction on $s$. 
\end{proof}

\begin{proposition} \label{equivalence between original goal terms and relational goal terms}
If $A$ is the set of goal terms in the original definition and $B$ is the set of goal terms in the new definition with relational sorts, then $A = B$ holds. 
\end{proposition}
\begin{proof}
By Proposition~\ref{goal terms in the original def can be generated by the new one} and the fact that goal terms in the original definition have relational sorts, we have $A \subseteq B$. 
Additionally, from Proposition~\ref{relational goal terms in the new def can be generated by the original one}, we know $B \subseteq A$. Therefore, by double inclusion, $A = B$. 
\end{proof}

Due to Proposition~\ref{equivalence between original goal terms and relational goal terms}, I use the word `relational goal terms' to mean goal terms in the original definition. Henceforth, I will use the new definition of goal terms. 

\subsubsection{Properties of goal terms}

The first proposition is the goal terms' counterpart of Proposition~\ref{uniqueness of sorts of first-order terms}. 

\begin{proposition} \label{uniqueness of sorts of goal terms}
Given $\Delta \vdash s: \sigma$, the sort of $s$ is unique; that is, we cannot have $\Delta \vdash s: \sigma'$, where $\sigma \neq \sigma'$.
\end{proposition}
\begin{proof}
The proof proceeds by structural induction on $s$. 

For the base case, if {\sc (GCst)} or {\sc (GVar)} is used, the sort of $s$ is uniquely determined by {\tt LSym} or $\mathbb{S}$. If $\Delta \vdash s: \sigma$ is created by {\sc (GFml)} or {\sc (GTerm)}, the sort of $s$ is uniquely determined
due to Proposition~\ref{uniqueness of sorts of first-order terms}. 

For the inductive case, suppose {\sc (GAbs)} is used. Hence, we have $s = \lambda x {:} \sigma. G$. Regardless of the sort of $s$, we can always uniquely determine the sort of $x$ because it is recorded in the lambda abstraction $\lambda x {:} \sigma. G$.
Therefore, the left hand side of $\Delta, x: \sigma \vdash G: \rho$ is fixed. It follows from the inductive hypothesis that the sort of $G$ is uniquely determined. Hence, the sort of $\lambda x {:} \sigma. G$ is unique as well.

Finally, if {\sc (GApp)} is used, we have $s = G \ H$. Since the sorts of $G$ and $H$ under $\Delta$ are uniquely determined by the inductive hypothesis, the claim holds for $G \ H$. 
\end{proof}

Similarly, the next theorem is the goal terms' counterpart of Theorem~\ref{type annotation of first-order terms}. 

\begin{theorem}
If $\Delta \vdash s: \sigma$ holds, where $s$ is a goal term, each symbol in $s$ can be annotated with a unique sort. 
\end{theorem}
\begin{proof}
The proof goes by by structural induction on goal terms.

For the base case, when $s$ is created by {\sc (GCst)} or {\sc (GVar)}, we can simply apply Proposition~\ref{uniqueness of sorts of goal terms} since $s$ only contains one symbol. If {\sc (GFml)} or {\sc (GTerm)} is used, the claim follows from
Theorem~\ref{type annotation of first-order terms}.  

For the inductive case, if $s$ is created by {\sc (GAbs)}, we know $s = \lambda x {:} \sigma. G$. From $\Delta \vdash \lambda x {:} \sigma. G: \sigma \to \rho$, we can uniquely deduce $\Delta, x: \sigma \vdash G: \rho$. 
By the inductive hypothesis, every symbol in $G$ can be annotated with a unique symbol. Thus, the claim holds in this case.

Finally, if {\sc (GApp)} is used, we know $s =G \ H$. By Proposition~\ref{uniqueness of sorts of goal terms}, we can uniquely determine the sorts of $G$ and $H$ under $\Delta$; that is, we can uniquely deduce typing judgements
$\Delta \vdash G: \rho_1$ and $\Delta \vdash H: \rho_2$. It follows from the inductive hypothesis that every symbol in $G$ and $H$ can be annotated with a unique sort. Therefore, the claim holds for every symbol in $G \ H$. 
\end{proof}

The following proposition saves us the need to be concerned about defunctionalizing partially applied instances of functions from $\mathbb{S}$ because they are never strictly partially applied in goal terms.

\begin{proposition}
Functions (i.e.~constants of arrow sort) from $\mathbb{S}$ cannot be strictly partially applied inside goal terms. 
\end{proposition}
\begin{proof}
Functions from $\mathbb{S}$ are introduced into goal terms by {\sc (GFml)} and {\sc (GTerm)}. Let $s$ be a first-order term (or formula) introduced by these two rules. Also, let $f \in \mathbb{S}$ be a function and $t$
be a first-order term $f \ t_1 \ \cdots t_k$, where $k < {\bf ar}(f)$. Hence, $t$ is strictly partially applied. In addition, assume that $t$ cannot be applied to another first-order term. This means that $t$ is maximal with respect to function application.
Since {\sc (GFml)} and {\sc (GTerm)} require $s$ to be of base sort, $s$ itself cannot be strictly partially applied. Thus, $t$ could only possibly appear (strictly) inside $s$.

Furthermore, because $t$ is assumed to be maximal with respect to function application, inside $s$, we cannot have $t \ u$ for some first-order term $u$. Thus, the only possibility for $u$ being located inside $s$ is that 
$s$ contains $u \ t$ for some $u$. However, as indicated by the conclusions in the six typing rules, first-order terms have sorts of order at most 2. Every subexpression of a first-order term is also a first-order term
and hence has a sort of order at most 2. Thus, the sort of $u$ has order at most 2; that is, the sort of $u$ looks like $b_1 \to \cdots \to b_n$, where $n > 1$ and $b_i \in \mathbb{B}$ for each $1 \leq i \leq n$. Since $u$ is applied
to $t$, the sort of $t$ must be $b_1$; that is, $t$ cannot have an arrow sort. Therefore, $t$ cannot appear inside $s$. This concludes the proof. 
\end{proof}

\subsection{Type preservation proof}

In this subsection, I prove that in an output of the defunctionalization algorithm, the logic program and the goal formula are well-sorted. Let $\mathcal{P} = (\Delta, P, G)$ be a source monotone problem and $\Sigma = (\mathbb{B}, \mathbb{S})$
be a first-order signature for $\mathcal{P}$. $P$ and $G$ are assumed to be well-sorted. Further, let $\mathcal{P}' = (\Delta', P', G')$ be the result of defunctionalizing $\mathcal{P}$ and $\Sigma' = (\mathbb{B}', \mathbb{S}')$ 
be a signature for $\mathcal{P}'$. 

The first theorem establishes well-sortedness of equations defining $Apply_{A}$. 

\begin{theorem} \label{well-sortedness of P' Apply}
Every equation in $P'_{\text{Apply}}$ is well-sorted. 
\end{theorem}
\begin{proof}
By \eqref{new P for Apply}, every equation in $P'_{\text{Apply}}$ takes the form
\begin{equation}
Apply_{\sigma_{n+1}'} = \lambda x, y, z. (\exists a_1, \ldots, a_n. x = C^{n}_{X} \ a_1 \ \cdots \ a_n \land z = C^{n+1}_{X} \ a_1 \ \cdots \ a_n \ y), \label{equation in P' Apply}
\end{equation}
where $X: \sigma_1 \to \cdots \to \sigma_m \to o \in \Delta$, $0\leq n \leq m - 2$, and $\sigma_i \leadsto_{T} \sigma_i'$ for all $1 \leq i \leq n + 1$. In \eqref{equation in P' Apply}, the equality between objects of sort {\bf closr} refers to $(=_{\bf closr})$ declared
in $\mathbb{S}'$. The sort of $(=_{\bf closr})$ is
\begin{equation*}
\vdash (=_{\bf closr}): {\bf closr} \to {\bf closr} \to o.
\end{equation*}
Note that I omit $\mathbb{S}'$ from typing judgements whenever its omission does not cause confusion. 

From \eqref{new Delta}, we know
\begin{align*}
\Delta' & \vdash C^{n}_{X}: \sigma_1' \to \cdots \to \sigma_{n}' \to {\bf closr} \\
\Delta' & \vdash C^{n+1}_{X}: \sigma_1' \to \cdots \to \sigma_{n+1}' \to {\bf closr}. 
\end{align*}
Let us denote $\{ a_i: \sigma_i' \mid 1 \leq i \leq n \}$ by $\{a_i: \sigma_i' \}$ for brevity. Applying (GApp) repeatedly, we can build the following typing derivations:
\begin{prooftree}
\AxiomC{$\Delta', \{a_i: \sigma_i' \} \vdash \{a_i: \sigma_i' \}$}
\AxiomC{$\Delta', \{a_i: \sigma_i' \} \vdash C^{n}_{X}: \sigma_1' \to \cdots \to \sigma_{n}' \to {\bf closr}$}
\BinaryInfC{$\Delta', \{a_i: \sigma_i' \} \vdash C^{n}_{X} \ a_1 \ \cdots \ a_n : {\bf closr}$}
\end{prooftree}
\begin{prooftree}
\AxiomC{$\Delta', \{a_i: \sigma_i' \}, y: \sigma_{n+1}' \vdash \{a_i: \sigma_i' \}, y: \sigma_{n+1}'$}
\AxiomC{$\Delta', \{a_i: \sigma_i' \}, y: \sigma_{n+1}' \vdash C^{n+1}_{X}: \sigma_1' \to \cdots \to \sigma_{n+1}' \to {\bf closr}$}
\BinaryInfC{$\Delta', \{a_i: \sigma_i' \}, y: \sigma_{n+1}' \vdash C^{n+1}_{X} \ a_1 \ \cdots \ a_n \ y: {\bf closr}$}
\end{prooftree}
Hence, for the two disjuncts in \eqref{equation in P' Apply}, we have
\begin{align*}
\Delta', \{a_i: \sigma_i' \}, x: {\bf closr}: \sigma_n' & \vdash (x = C^{n}_{X} \ a_1 \ \cdots \ a_n) : o \\
\Delta', \{a_i: \sigma_i' \}, y: \sigma_{n+1}', z: {\bf closr}: \sigma_n' & \vdash (z = C^{n+1}_{X} \ a_1 \ \cdots \ a_n \ y) : o. 
\end{align*}
These two typing judgements yield
\begin{equation*}
\Delta', x: {\bf closr}, y: \sigma_{n+1}', z: {\bf closr} \vdash (\exists a_1, \ldots, a_n. x = C^{n}_{X} \ a_1 \ \cdots \ a_n \land z = C^{n+1}_{X} \ a_1 \ \cdots \ a_n \ y): o.
\end{equation*}
Finally, by {\sc (GAbs)}, we obtain
\begin{equation*}
\Delta' \vdash \lambda x, y, z. (\exists a_1, \ldots, a_n. x = C^{n}_{X} \ a_1 \ \cdots \ a_n \land z = C^{n+1}_{X} \ a_1 \ \cdots \ a_n \ y): {\bf closr} \to \sigma_{n+1}' \to {\bf closr} \to o.
\end{equation*}
Whether $\sigma_{n+1} \in \mathbb{B}$ or $\sigma_{n+1} = {\bf closr}$, $\sigma_{n+1}' \in \mathbb{B}'$ holds by the definition of $\mathbb{B}'$. Thus, it is given by \eqref{new Delta} that
\begin{equation*}
\Delta' \vdash Apply_{\sigma_{n+1}'}: {\bf closr} \to \sigma_{n+1}' \to {\bf closr} \to o.
\end{equation*}
Therefore, the left and right hand sides of \eqref{equation in P' Apply} have the same sort as required. 
\end{proof}

The next lemma plays a pivotal role in proving that all equations in $P'_{\text{IOMatch}}$ are well-sorted. 

\begin{lemma} \label{lemma for well-sortedness of P' IOMatch}
Let $s$ be a well-sorted source goal term over $\Sigma = (\mathbb{B}, \mathbb{S})$ that contains no lambda abstraction. Also, suppose $App = \{Apply_{A}: {\bf closr} \to A \to {\bf closr} \to o \mid A \in \mathbb{B} \cup \{ {\bf closr} \} \}$ and 
$IO = \{IOMatch_{A}: {\bf closr} \to A \to o \mid A \in \mathbb{B} \cup \{{\bf closr} \} \}$. 

If $\Gamma \vdash s: b$, where $b \in \mathbb{B}$, then $s \leadsto t$ holds for some goal term $t$. Furthermore, we have $\Gamma', App, IO \vdash t: b$, where $\Gamma' = \{v: \sigma' \mid v: \sigma \in \Gamma, \sigma \leadsto_{T} \sigma' \}$. 
Here, we use the fact that $\leadsto_{T}$ is a function.

Otherwise, if $\Gamma \vdash s: \rho$, where $\rho \notin \mathbb{B}$, then $s \leadsto^{X} t$ holds for some goal term $t$. Furthermore, we have $\Gamma', App, IO, X: {\bf closr} \vdash t: o$, where 
$\Gamma' = \{v: \sigma' \mid v: \sigma \in \Gamma, \sigma \leadsto_{T} \sigma' \}$. 
\end{lemma}
\begin{proof}
The proof proceeds by structural induction on $s$. 

For the base case, suppose $s \in Fm \cup Tm$. Then the only inference rule that is applicable is {\sc (ConstrLan)}, which gives $s \leadsto s$. Because $s$ is well-sorted, all free variables in $s$ should be included in $\Gamma$. 
This can be formally proved, but I will not do it here. Additionally, by Proposition~\ref{free variables in first-order terms}, every free variable occurring in first-order terms have base sorts. As $b \leadsto_{T} b$ for any $b \in \mathbb{B}$,
we have
\begin{align*}
\Gamma' & =\{v: \sigma' \mid v: \sigma \in \Gamma, \sigma \leadsto_{T} \sigma' \} \\
& = \{u: b \mid u \in {\tt FV}(s), u: b \in \Gamma \} \\
& \quad \cup \{v: \sigma' \mid v \notin {\tt FV}(s), v: \sigma \in \Gamma, \sigma \leadsto_{T} \sigma' \}.
\end{align*}
Hence, free variables in $s$ have the same sorts in $\Gamma'$ as in $\Gamma$. As the sort of $s$ depends only on the sorts of free variables in $s$, we obtain
\begin{gather*}
\Gamma' \vdash s: b \\
\therefore \Gamma', App, IO \vdash s: b.
\end{gather*}
Thus, the claim holds in this case. The case for {\sc (Var-Base)} can be proved analogously.

Next, consider the case of {\sc (Var-Arrow)}. According to the rule, we have $s = x$, where $\Gamma \vdash x: \rho$ and $\rho$ is a relational arrow sort. {\sc (Var-Arrow)} yields that $x \leadsto^{X} X = x$. As $s$ is
well-sorted under $\Gamma$, it is given by {\sc (GVar)} that $x \in {\tt dom}(\Gamma)$. 
Because $\rho \leadsto_{T} {\bf closr}$ for any relational arrow sort $\rho$, we have $(x: {\bf closr}) \in \Gamma'$. It is straightforward to see that $x: {\bf closr}, X: {\bf closr} \vdash (X = x): o$ holds. It thus follows 
that $\Gamma', App, IO, X: {\bf closr} \vdash (X = x): o$ holds. Therefore, the claim is true in this case. The case for {\sc (TopVar)} can be proved in the same fashion. 

For the inductive case, assume $s = c \ E \ F$, where $c \in \{\land, \lor \}$. $s$ is thus defunctionalized by {\sc (LogSym)}. Since $s$ is well-sorted, by {\sc (GCst)} and {\sc (GApp)}, we have
\begin{align*}
\Gamma & \vdash c: o \to o \to o \\
\Gamma & \vdash E: o \\
\Gamma & \vdash F: o. 
\end{align*}
By the inductive hypothesis, $\Gamma', App, IO \vdash E': o$ and $\Gamma', App, IO \vdash F': o$ hold, where $E \leadsto E'$ and $F \leadsto F'$. It follows that $\Gamma', App, IO \vdash (c \ E' \ F'): o$. 

Next, suppose $s = E \ F$, where $E \ F$ and $F$ have arrow sorts. This case of $s$ is handled by {\sc (App)} and {\sc (App-Arrow)}. Thus, we have
\begin{equation*}
s \leadsto^{X} t,
\end{equation*}
where $t= \exists_{\bf closr} x. (E' \land \exists_{\bf closr} y. (F' \land Apply_{\bf closr} \ x \ y \ X))$ and $E \leadsto^{x} E'$ and $F \leadsto^{y} F'$. Because both $E$ and $F$ have arrow sorts, the inductive hypothesis gives
\begin{align*}
\Gamma', App, IO, x: {\bf closr} & \vdash E': o \\
\Gamma', App, IO, y: {\bf closr} & \vdash F': o. 
\end{align*}
It is therefore possible to construct a typing derivation tree for $\Gamma', App, IO, X: {\bf closr} \vdash t: o$, although I omit it because it takes a lot of space. 

The remaining three cases when $s = E \ F$ can be proved analogously.
\end{proof}

\begin{theorem} \label{well-sortedness of P' IOMatch}
Every equation from $P'_{\text{IOMatch}}$ is well-sorted. 
\end{theorem}
\begin{proof}
By \eqref{new P for IOMatch}, each rule in $P'_{\text{IOMatch}}$ has the form
\begin{equation*}
IOMatch_{\sigma_{m}'} = \lambda x, x_{m}. (\exists x_1, \ldots, x_{m-1}. x = C^{m-1}_{X} \ x_1 \ \cdots \ x_{m-1} \land F'),
\end{equation*}
where $X = \lambda x_1 {:} \sigma_1, \ldots, x_m {:} \sigma_m. F$ is in $P$. Here, ${\bf ar}(X) = m$ and $F \leadsto F'$. As ${\bf ar}(X) = m$, $F$ cannot be a lambda abstraction. 
Further, equations in $P$ are assumed to be well-sorted. Thus, we obtain
\begin{gather*}
\Delta \vdash (\lambda x_1 {:} \sigma_1, \ldots, x_m {:} \sigma_m. F): \sigma_1 \to \cdots \to \sigma_m \to o \\
\therefore \Delta, \{x_i: \sigma_i \mid 1 \leq i \leq m \} \vdash F: o. 
\end{gather*}
Lemma~\ref{lemma for well-sortedness of P' IOMatch} yields that
\begin{equation*}
\Gamma, \{x_i: \sigma_i' \mid 1 \leq i \leq m, \sigma_i \leadsto_{T} \sigma_i' \}, \Delta' \vdash F': o,
\end{equation*}
where $\Gamma = \{X: \sigma' \mid X: \sigma \in \Delta, \sigma \leadsto_{T} \sigma' \}$ and $\Delta'$ is given by \eqref{new Delta}. It is relatively straightforward to prove that $F'$ does not contain any top-level
relational variable symbol from $\Delta$. Therefore, $\Gamma$ does not affect the sort of $F'$. Consequently, we obtain
\begin{equation*}
\{x_i: \sigma_i' \mid 1 \leq i \leq m, \sigma_i \leadsto_{T} \sigma_i' \}, \Delta' \vdash F': o. 
\end{equation*}
It is possible to construct a valid typing derivation tree for
\begin{equation*}
\Delta' \vdash \lambda x, x_{m}. (\exists x_1, \ldots, x_{m-1}. x = C^{m-1}_{X} \ x_1 \ \cdots \ x_{m-1} \land F'): {\bf closr} \to \sigma_m' \to o.
\end{equation*}
This is consistent with the sort of $IOMatch_{\sigma_m'}$ given by \eqref{new Delta}. Therefore, each equation in $P'_{\text{IOMatch}}$ is indeed well-sorted. 
\end{proof}

\begin{theorem}
Each equation in $P'$ and $G'$ is well-sorted. 
\end{theorem}
\begin{proof}
Well-sortedness of equations in $P'$ follows from Theorem~\ref{well-sortedness of P' Apply} and Theorem~\ref{well-sortedness of P' IOMatch}. As for $G'$, because each $s \in G$ is free of lambda abstractions, by 
Lemma~\ref{lemma for well-sortedness of P' IOMatch}, we have $\Gamma, \Delta' \vdash t: o$, where $s \leadsto t$ and $\Gamma = \{X: \sigma' \mid X: \sigma \in \Delta\}$. Since $G'$ does not contain any top-level relational variable symbols
from $\Delta$, $\Gamma$ can be removed from the typing judgement. This results in $\Delta' \vdash G': o$. 
\end{proof}

\bibliographystyle{splncs04} 
\bibliography{references}

\end{document}